\documentclass[12pt,a4paper]{amsart}
\usepackage[a4paper,inner=2.5cm,outer=2.5cm,top=2.5cm,bottom=2.5cm]{geometry}
\usepackage{amsmath,amssymb,amsthm,enumerate,mathtools,stmaryrd
}
\usepackage{hyperref}
\hypersetup{colorlinks=true,linkcolor=blue,citecolor=teal,filecolor=magenta,urlcolor=cyan}

\usepackage[T2A,T1]{fontenc}
\usepackage[utf8]{inputenc}
\usepackage[russian,english]{babel}
\usepackage{epigraph}
\usepackage{csquotes}

\usepackage{makecell}
\usepackage{graphicx}

\usepackage{float}

\usepackage{extarrows}

\usepackage{xcolor}

\newcommand{\Cf}{\mathbb{C}}

\usepackage{xurl}
\usepackage[backend=bibtex,style=alphabetic,sorting=nyt,isbn=false,url=false,doi=true,maxalphanames=10,minalphanames=4,mincitenames=4,maxcitenames=10,minnames=4,maxnames=10,giveninits=true,maxbibnames=99]{biblatex}
\theoremstyle{plain}
\newtheorem{theorem}{Theorem}[section]
\newtheorem{proposition}[theorem]{Proposition}

\newtheorem{lemma}[theorem]{Lemma}

\newtheorem{conjecture}[theorem]{Conjecture}
\theoremstyle{definition}
\newtheorem{definition}[theorem]{Definition}

\makeatletter
\@addtoreset{proofpart}{theorem}
\makeatother
\theoremstyle{remark}
\newtheorem{remark}[theorem]{Remark}

\newcommand{\Z}{\mathbb{Z}}
\newcommand{\J}{\mathbb{J}}
\newcommand{\cD}{\mathcal{D}}
\newcommand{\cS}{\mathcal{S}}
\newcommand{\cA}{\mathcal{A}}
\newcommand{\cE}{\mathcal{E}}
\newcommand{\cF}{\mathcal{F}}

\newcommand{\cW}{\mathcal{W}}
\newcommand{\cT}{\mathcal{T}}

\newcommand{\VEV}[1]{{\big\langle \! 0 \big| {#1} \big| 0 \! \big\rangle}}
\newcommand{\VEVc}[1]{{\big\langle \! 0 \big| {#1} \big| 0 \!\big\rangle^\circ}}
\newcommand{\vac}{{\big| 0\! \big\rangle}}
\newcommand{\covac}{{\big\langle \! 0 \big| }}

\newcommand{\res}{\mathop{\rm res}}

\newcommand{\coeff}[2]{\mathop{[{#1}^{#2}]}}
\newcommand{\restr}[2]{\mathop{\big\lfloor_{{#1}\to {#2}}}}
\newcommand{\set}[1]{\llbracket {#1} \rrbracket}

\def\mfS{\mathfrak{S}}

\def\Expr{\mathsf{Expr}}

\numberwithin{equation}{section}

\title[Universal formula for $x-y$ swap]{A universal formula for the $x-y$ swap in topological recursion}

\author[A.~Alexandrov]{A.~Alexandrov}
\address{A.~A.: Center for Geometry and Physics, Institute for Basic Science (IBS), Pohang 37673, Korea
}
\email{alex@ibs.re.kr}

\author[B.~Bychkov]{B.~Bychkov}
\address{B.~B.: Department of Mathematics, University of Haifa, Mount Carmel, 3498838, Haifa, Israel}
\email{bbychkov@hse.ru}

\author[P.~Dunin-Barkowski]{P.~Dunin-Barkowski}
\address{P.~D.-B.: Faculty of Mathematics, HSE University, Usacheva 6, 119048 Moscow, Russia; HSE--Skoltech International Laboratory of Representation Theory and Mathematical Physics, Skoltech, Bolshoy Boulevard 30 bld. 1, 121205 Moscow, Russia; and NRC “Kurchatov Institute” -- ITEP, 117218 Moscow, Russia}
\email{ptdunin@hse.ru}

\author[M.~Kazarian]{M.~Kazarian}
\address{M.~K.: Faculty of Mathematics, HSE University, Usacheva 6, 119048 Moscow, Russia; and Igor Krichever Center for Advanced Studies, Skoltech, Bolshoy Boulevard 30 bld. 1, 121205 Moscow, Russia}
\email{kazarian@mccme.ru}

\author[S.~Shadrin]{S.~Shadrin}
\address{S.~S.: Korteweg-de Vries Institute for Mathematics, University of Amsterdam, Postbus 94248, 1090GE Amsterdam, The Netherlands}
\email{S.Shadrin@uva.nl}	

\begin{document}
	
	\begin{abstract} We prove a recent conjecture of Borot et al.~that a particular universal closed algebraic formula recovers the correlation differentials of topological recursion after the swap of $x$ and $y$ in the input data. We also show that this universal formula can be drastically simplified (as it was already done by Hock).
		
		As an application of this general $x-y$ swap result, we prove an explicit closed formula for the topological recursion differentials for the case of any spectral curve with unramified $y$ and arbitrary rational $x$.
	\end{abstract}
	
	\maketitle
	\setcounter{tocdepth}{2}

	\tableofcontents
	
	
	\epigraph{\foreignlanguage{russian}{Егор Алексеич берёт в руки грифель и начинает решать. Он заикается, краснеет, бледнеет.
			\\		
			— Эта задача, собственно говоря, алгебраическая, — говорит он. — Её с иксом и игреком решить можно. Впрочем, можно и так решить.}}{\textit{\foreignlanguage{russian}{Репетитор}}
		
		\textsc{\foreignlanguage{russian}{А. П. Чехов}}
	}
	
	
	\epigraph{\emph{English translation of the epigraph:}
		
		\bigskip	
		
		Gregory takes the pencil and begins figuring. He hiccoughs and flushes and pales.
		
		``The fact is, this is an algebraical problem," he says. "It ought to be solved with x and y. But it can be done in this way, too.''}
	{\textit{The Tutor}
		
		\textsc{A.~P.~Chekhov}
	}

\section{Introduction}

\subsection{Topological recursion}\label{sec:toprecintro}
The Chekhov--Eynard--Orantin topological recursion~\cite{EynardOrantin-toporec} associates to an input that consists of
\begin{itemize}
	\item a compact Riemann surface $S$ with a chosen basis of $\mathfrak{A}$- and $\mathfrak{B}$-cycles (the so-called spectral curve);
	\item two meromorphic functions $x$ and $y$ on $S$ such that all critical points $p_1,\dots,p_a\in S$ of $x$ are simple, $y$ is regular at $p_i$ and $dy|_{p_i}\not=0$, $i=1,\dots,a$;
	\item a symmetric bi-differential $B$ on $S^2$ with a double pole on the diagonal with bi-residue $1$ and holomorphic at all other points, normalized in such a way that its $\mathfrak A$-periods in both variables vanish (the so-called Bergman kernel)
\end{itemize}
an output that consists of a system of symmetric meromorphic differentials $\omega^{(g)}_{m,0}$ on $S^m$, $g\geq 0$, $m\geq 1$, which are called the \emph{correlation differentials}. The recursion is given by setting $\omega^{(0)}_{1,0}(z_1)  \coloneqq -
y(z_1)dx(z_1)$, $\omega^{(0)}_{2,0}(z_1,z_2) \coloneqq B(z_1,z_2)$, and then for $2g-2+m>0$ we have
\begin{align} \label{eq:CEO-TR-original}
	\omega^{(g)}_{m,0}(z_{\llbracket m \rrbracket}) & = \frac 12 \sum_{i=1}^a \res_{z=p_i} \frac{\int_z^{\sigma_i (z)} \omega^{(0)}_{2,0}(z_1,\bullet) }{\omega^{(0)}_{1,0}({\sigma_i (z)}) - \omega^{(0)}_{1,0}(z)} \Bigg( \omega^{(g-1)}_{m+1,0} (z,\sigma_i(z),z_{\llbracket m \rrbracket \setminus \{1\}})
	\\ \notag & \qquad + \sum_{\substack{g_1+g_2=g \\ I_1\sqcup I_2 = \llbracket m \rrbracket \setminus \{1\}\\ (g_i,|I_i|)\not= (0,0)}} \omega^{(g_1)}_{|I_1|+1,0}(z,z_{I_1})\omega^{(g_2)}_{|I_2|+1,0}(\sigma_i(z),z_{I_2}) \Bigg).
\end{align}
Here we use notation $\llbracket m \rrbracket \coloneqq \{1,\dots,m\}$, $z_I = \{z_i\}_{i\in I}$ for any $I\subseteq \llbracket m \rrbracket$, and $\sigma_i$ denotes the deck transformation with respect to the function $x$ near the point $p_i$. The extra seemingly superfluous subscript $0$ in the notation $\omega^{(g)}_{m,0}$ is explained below.

\begin{remark}
	The requirements for $x$ and $y$ can be substantially relaxed, for instance, it is sufficient to assume that $ydx$ is meromorphic (and the function $x$ itself might have essential singularities), and in fact for $2g-2+m>0$ the recursion only uses the local germs of $x$ and $y$ at the points $p_1,\dots,p_a$. However, in this paper we focus on the situation with $x$ and $y$ both globally defined meromorphic functions.
\end{remark}

\begin{remark} The formula for the recursion~\eqref{eq:CEO-TR-original} is given just for historical reasons and plays by itself almost no role in this paper. It should be rather considered as a way to recursively resolve the so-called loop equations (see Section~\ref{sec:RecollectionTR}) in the special space of symmetric meromorphic differentials that have poles only at the points $p_1,\dots,p_a$ and are uniquely reconstructed from the singular parts of their expansions at these points, cf.~\cite{BEO-loop,BS-blobbed}.
\end{remark}

\begin{remark}
	It is totally beyond the scope of this paper to overview the origins and numerous applications of topological recursion. To motivate its study it is enough to mention that it provides a unifying interface between matrix models, mirror symmetry, integrable systems, enumerative geometry, enumerative combinatorics, quantum knot invariants, free probability theory, and many further areas of mathematics and mathematical physics.
\end{remark}

\subsection{The \texorpdfstring{$x-y$}{x-y} swap} The roles of functions $x$ and $y$ in topological recursion are not symmetric. Assume, however, that the function $y$ satisfies the same properties as $x$: assume that all critical points $q_1,\dots,q_b\in S$ of $y$ are simple, $x$ is regular at $q_i$ and $dx|_{q_i}\not=0$, $i=1,\dots, b$. Then we can start a new recursion with the roles of $x$ and $y$ interchanged. This way we produce a system of symmetric meromorphic differentials $\omega^{(g)}_{0,n}$ on $S^n$, $g\geq 0$, $n\geq 1$, with $\omega^{(g)}_{0,1}(z_1)\coloneqq -x(z_1)dy(z_1)$, $\omega^{(0)}_{0,2}(z_1,z_2)\coloneqq B(z_1,z_2)$, and  for $2g-2+n>0$ we have
\begin{align} \label{eq:CEO-TR-xyswap}
	\omega^{(g)}_{0,n}(z_{\llbracket m \rrbracket}) & = \frac 12 \sum_{i=1}^b \res_{z=q_i} \frac{\int_z^{\tau_i (z)} \omega^{(0)}_{0,2}(z_1,\bullet) }{\omega^{(0)}_{0,1}({\tau_i (z)}) - \omega^{(0)}_{0,1}(z)} \Bigg( \omega^{(g-1)}_{0,n+1} (z,\tau_i(z),z_{\llbracket n \rrbracket \setminus \{1\}})
	\\ \notag & \qquad + \sum_{\substack{g_1+g_2=g \\ I_1\sqcup I_2 = \llbracket n \rrbracket \setminus \{1\}\\ (g_i,|I_i|)\not= (0,0)}} \omega^{(g_1)}_{0,|I_1|+1}(z,z_{I_1})\omega^{(g_2)}_{0,|I_2|+1}(\tau_i(z),z_{I_2}) \Bigg).
\end{align}
Here $\tau_i$ denotes the deck transformation with respect to the function $y$ near the point $q_i$.

\begin{remark}\label{rem:MixedCorrDiff} Now we can explain what is the meaning of the extra subscripts $0$ in the notation for $\omega^{(g)}_{m,0}$ and $\omega^{(g)}_{0,n}$. There is a natural family of meromorphic $n+m$ differentials $\omega^{(g)}_{m,n}$ on $S^{m+n}$, defined for $g\geq 0$ and $m+n\geq 1$, symmetric under the action of $\mfS_m\times \mfS_n$ interchanging the first $m$ and the last $n$ arguments. These ``mixed'' correlation differentials are interpolating between the outputs of two topological recursions presented above, and we claim that they are very natural objects, at least from the point of view of $2$-matrix model~\cite{EynardOrantin-xysymmetry} and KP integrability, though they are so far not present in the literature other than in the context of $2$-matrix model. 
\end{remark}

\subsection{Historical perspective} Initially, the concept of $x-y$ swap was studied only in the framework of the so-called symplectic invariance of the free energies associated to the correlation differentials of topological recursion, see~\cite{EynardOrantin-toporec,EynardOrantin-xysymmetry,bouchard-sulk,eynard2013xy}.
In this framework the correlation differentials on the dual side, and more general $n+m$ differentials $\omega^{(g)}_{m,n}$ played a role of an auxiliary tool. In particular, no meaningful enumerative interpretation was known outside of the example of the $2$-matrix model that we mentioned above. 

This situation changed with the work of Borot and Garcia-Failde~\cite{ElbaPHD,BGF-fullysimple} who conjectured an instance of $x-y$ swap such that the expansions of correlation differentials on both sides were combinatorially meaningful (their conjecture is proved in~\cite{borot2021topological,BDKS-FullySimple}). To this end they introduced the concept of the so-called fully simple maps. They also suggested that their approach provided the right framework for the moment-cumulant correspondence in the free probability theory, which was later proved in~\cite{BCGFLS-Free}.

These developments led to a conjecture by Borot et al. that the universal functional relations used in~\cite{BDKS-FullySimple} in the context of fully simple maps and in~\cite{BCGFLS-Free} in the context of the moment-cumulant correspondence are universal for $x-y$ swap in topological recursion in general, see~\cite[Conjecture 3.14]{BCGFLS-Free}. This conjecture was also supported by an old idea of Eynard and Orantin that the $2$-matrix model is rich enough to provide a universal system of relations for the correlation differentials under the $x-y$ swap. 

\subsection{Conjecture of Borot et al.} The main purpose of this paper is to give a proof of a conjectural formula proposed in~\cite[Conjecture 3.14]{BCGFLS-Free} that expresses $\omega^{(g)}_{0,n}$ in terms of $\omega^{(g)}_{m,0}$ as universal algebraic expressions. Our immediate goal is to recall this conjecture.

In order to do this we have to say a few words about the background of the formulas that will emerge throughout the paper. They all come from contractions of vertex operators in the Fock space studied in~\cite{BDKS-OrlovScherbin,BDKS-toporec-KP,BDKS-FullySimple,BDKS-symplectic}, and we explain their origin in Section~\ref{sec:FormalPowerSeries} and, later on, in Section~\ref{sec:MotivationFPS}. In principle, the reader can ignore the initial source of our formulas and skip these sections, but then the reader will have to face some cumbersome expressions presented with no prior motivation or explanation. Moreover, in the proofs we have constantly to switch from dealing with globally defined meromorphic differentials to their local expressions to be able to use the techniques of vertex operators in the Fock space. So, the Fock space background of our formulas is also necessary to follow the proofs. 

\subsubsection{Universal algebraic expressions} \label{sec:FormalPowerSeries}
In order to define the universal algebraic expressions and explain their origin, we introduce some extra assumptions and notation. Assume that there exists a point $P\in S$ such that $x$ has a simple zero at $P$ and $y$ has a simple pole at $P$ and $xy|_P = 1$.

The functions $x$ and $y^{-1}$ define local coordinates near the point $P$. By the construction of topological recursion the form $\omega^{(g)}_{m,0}$ is holomorphic at $P$ with an appropriate singular correction for unstable cases and we can consider the expansions of these forms at the local coordinates $x_i=x(z_i)$
\begin{align}
	& \omega^{(g)}_{m,0}(z_{\llbracket m\rrbracket})
	-\delta_{(g,m),(0,2)}\tfrac{dx_1dx_2}{(x_1-x_2)^2}-\delta_{(g,m),(0,1)}\big(-\tfrac{dx_1}{x_1}\big)\\ \notag & \sim
	\sum_{k_1,\dots,k_m=1}^\infty C^{(g)}_{k_1,\dots,k_m} \prod_{i=1}^m x_i^{k_i} \Big(-\frac{dx_i}{x_i}\Big),
\end{align} 
where $\delta_{(\dots),(\dots)}$ is the Kronecker delta function of the multi-indexes. 

\begin{remark} The extra sign $(-1)^m$ in the expansion on the right hand side might be captured in the definition of $\omega^{(0)}_{1,0}$, which is sometimes natural to define as $ydx$ rather than $-ydx$. This would results in the extra sign $(-1)^m$ for each $ \omega^{(g)}_{m,0}$. This introduction of extra sign is natural from the point of view of $x-y$ duality, and more generally, symplectic duality, cf.~\cite{BDKS-symplectic}.
\end{remark}

The above expansion can be represented in the following equivalent way
\begin{align}
	& \sum_{g=0}^{\infty} \hbar^{2g-2+m}\omega^{(g)}_{m,0} -\delta_{m,2}\tfrac{dx_1dx_2}{(x_1-x_2)^2}-\delta_{m,1}\hbar^{-1}\Bigl(-\tfrac{dx_1}{x_1}\Bigr)
	\\ \notag & \quad = \VEVc{\prod_{i=1}^m \sum_{\ell_i=1}^\infty x_i^{\ell_i}\Bigl(-\frac{dx_i}{x_i}\Bigr) J_{\ell_i}\;Z},
	\\
	& Z\coloneqq  \exp\Big(\sum_{\substack{g\geq 0,\ m\geq 1\\ k_1,\dots,k_m\geq 1}}\tfrac{\hbar^{2g-2+m}}{m!}
	C^{(g)}_{k_1,\dots,k_m} \prod_{j=1}^m \frac{J_{-k_i}}{k_i}\Big).\label{def:Z}
\end{align}  
This formula is written as a vacuum expectation value on the Fock space 
\begin{align}
	\cF=\Cf[[p_1,p_2,\dots]].
\end{align} 
Here $\vac \coloneqq 1\in \cF$, for any $f\in \cF$ we have $\covac f \coloneqq f|_{p_1=p_2=\cdots=0}$, and for any $k\geq 1$ we have $J_{-k} \coloneqq p_k\cdot$ and $J_k \coloneqq k\partial_{p_k}$. The symbol $\VEVc{\cdots}$ denotes the so-called connected vacuum expectation that is obtained from the usual vacuum expectation $\VEV{\cdots}$ by inclusion-exclusion formulas.

Define the operator $\cD$ acting diagonally in the basis of Schur functions as
\begin{align}
	\cD \colon s_\lambda(p_1,p_2,\dots) \mapsto s_\lambda(p_1,p_2,\dots) \cdot \prod_{(i,j)\in \lambda} (1+\hbar(i-j)).	
\end{align}
We introduce a new system of differentials $\tilde\omega^{(g)}_{0,n}$ defined as the following formal expansions in the local coordinates $y_i^{-1}$
\begin{align} \label{eq:UExpr-FormalPowerSeries}
	\sum_{g=0}^{\infty} \hbar^{2g-2+n}\tilde\omega^{(g)}_{0,n} \coloneqq \VEVc{\prod_{i=1}^n \sum_{\ell_i=1}^\infty y_i^{-\ell_i} \Bigl(-\frac{dy_i}{y_i}\Bigr)  J_{\ell_i}\;\cD\;Z}.
\end{align}

It is proved in~\cite{BDKS-FullySimple}, see also~\cite{BCGFLS-Free}, that there exists a universal closed dif\-fe\-ren\-tial-algebraic expression that expresses the differentials $\tilde\omega^{(g)}_{0,n}$ in terms of the differentials $\omega^{(h)}_{|I|,0}(z_I)$, $I\subseteq \llbracket n \rrbracket$, $2h-2+|I| \leq 2g-2+n$, and their iterated $x_i\partial_{x_i}$  and $y_i\partial_{y_i}$ derivatives, and in addition to them we also use $dx_i/x_i$, $dy_i/y_i$, and $dx_idx_j/(x_i-x_j)^2$ in these expressions:
\begin{multline} \label{eq:UExpr-Algebraic}
	\tilde\omega^{(g)}_{0,n}+\delta_{(g,n),(0,2)}\tfrac{dy_1dy_2}{(y_1-y_2)^2}+\delta_{(g,n),(0,1)}\big(-\tfrac{dy_1}{y_1}\big)
	\\=
	\Expr_{g,n}\Big(\big\{\omega^{(h)}_{m,0}\big\}_{2h-2+m\leq 2g-2+n}, \big\{\tfrac{dx_i}{x_i},\tfrac{dy_i}{y_i}\big\}_{i=1,\dots,n}, \big\{\tfrac{dx_idx_j}{(x_i-x_j)^2}\big\}_{\substack {i,j=1,\dots,n\\ i\not=j}}\Big),
\end{multline}
see the explicit formula~\eqref{eq:MainFormula} below. It follows that the forms $\tilde\omega^{(g)}_{0,n}$ defined initially as formal power expansions extend as global meromorphic $n$-differentials that we denote also by~$\tilde\omega^{(g)}_{0,n}$.

\begin{remark} Note that both the formal power series~\eqref{eq:UExpr-FormalPowerSeries} and the universal algebraic expressions~\eqref{eq:UExpr-Algebraic} appear quite often in the literature. These formulas are used in enumeration of double strictly monotone Hurwitz numbers (in various special situations turning into enumeration of Gro\-then\-dieck's dessins d'enfants, higher genera Catalan numbers, ribbon graphs, lattice points in the moduli spaces of curves, hypermaps, etc.~etc.). The universal expressions $\Expr_{g,n}$ occur also as universal relations between cumulants of two matrices in the Hermitian two matrix model, and in the moment-cumulant correspondence in the free probability theory~\cite{BCGFLS-Free}.
\end{remark}

\subsubsection{The conjecture}\label{sec:theconj} The previous subsection explains the origin of some so far not explicitly presented expression $\Expr_{g,n}$ implied by the constructions given there.

\begin{conjecture}[{\cite[Conjecture 3.14]{BCGFLS-Free}}] \label{conj:Main}
	For $g\geq 0$, $n\geq 1$ we have:
	\begin{equation}
		\omega^{(g)}_{0,n} = \label{eq:MainConjecture}
		\Expr_{g,n}\Big(\big\{\omega^{(h)}_{m,0}\big\}_{2h-2+m\leq 2g-2+n}, \Big\{\frac{dx_i}{x_i},\frac{dy_i}{y_i}\Big\}_{i=1,\dots,n}, \Big\{\frac{dx_idx_j}{(x_i-x_j)^2}\Big\}_{\substack {i,j=1,\dots,n\\ i\not=j}}\Big).
	\end{equation}
\end{conjecture}

In other words, we forget all assumptions that we made in order to define and explain the origin of $\Expr_{g,n}$ and apply these expressions directly to globally defined symmetric $m$-differentials on $S^m$. As a result, we obtain a symmetric $n$-differential on $S^n$, and the conjecture states that it is precisely 
the correlation differential obtained with interchanged roles of~$x$ and~$y$.


\begin{theorem}\label{thm:mainconjholds} Conjecture~\ref{conj:Main} holds. That is, assume the differentials $\omega^{(g)}_{m,0}$ satisfy the topological recursion for the input data $(S,x,y,B)$ with meromorphic $x$ and $y$, all zeros of $dy$ are simple, and $x$ is regular at the zeros of $dy$. Then the differentials $\omega^{(g)}_{0,n}$ given by~\eqref{eq:MainConjecture} satisfy the topological recursion for the input data $(S,y,x,B)$.
\end{theorem}

We prove Theorem~\ref{thm:mainconjholds} in Section~\ref{sec:proofofmainthms}. Note that the assumption that the differentials $\omega^{(g)}_{m,0}$ are globally defined is used there in Proposition~\ref{prop:trueTRformulation} in order to replace the topological recursion by some uniqueness property. 

\subsubsection{Explicit formula}\label{sec:ExplFormula} Let us recall the explicit formula for $\Expr_{g,n}$ given in~\cite{BCGFLS-Free,BDKS-FullySimple} specified to our situation as it is required in Equation~\eqref{eq:MainConjecture}, that is, applied to globally defined symmetric meromorphic differentials.

We use graphs with multiedges. A graph with multiedges is given by its set of vertices $V$, its set of multiedges $E$, the set of red flags $R$, and the set of blue flags $L$, equipped with maps $r\colon R\to V$, $l\colon L\to E$, and an isomorphism $\iota \colon L\to R$. 
The index of a vertex $v\in V$ is $|r^{-1}(v)|$, the index of a multiedge $e\in E$ is $|l^{-1}(e)|$. We call elements of $l^{-1}(e)$ the \emph{legs} of a multiedge $e$, and we say that a particular leg $h$ of $e$ is attached to a vertex $v$ if $v=r(\iota(h))$.

Let  $x_i = x(z_i)$,  $\tilde x_i = x(\tilde z_i)$ and $y_i=y(z_i)$. For $g\geq 0$, $n\geq 1$ we have:
\begin{align} \label{eq:MainFormula}
	& \omega_{0,n}^{(g)} (z_{\llbracket n\rrbracket})\prod_{i=1}^n \Big(-\frac{y_i}{dy_i}\Big) 
	=
	\\ \notag &
	[\hbar^{2g}] \sum_{\Gamma} \frac{
		\hbar^{2g(\Gamma)}}{|\mathrm{Aut}(\Gamma)|} \prod_{i=1}^n
	\sum_{k_i=0}^\infty ( -y_i \partial_{y_i})^{k_i} [v_i^{k_i}]
	\\ \notag & 
	\Big(-\frac{y_i}{dy_i}\frac{dx_i}{x_i} \Big)
		\restr{\theta_i}{ - \frac{x_i}{dx_i}\omega^{(0)}_{1,0} (z_i)}
	\sum_{r_i=0}^\infty \Big(\partial_{\theta_i} + \frac{v_i}{\theta_i}\Big)^{r_i}
	e^{v_i\frac{\cS(\hbar v_i \partial_{\theta_i})}{\cS(\hbar \partial_{\theta_i})} \log\theta_i - v_i\log\theta_i}
	[u_i^{r_i}]
	\\ \notag &
	\frac{1}{ u_i \cS(\hbar u_i)} e^{u_i \cS(\hbar u_i x_i \partial_{x_i}) \sum_{\tilde g=0}^\infty \hbar^{2\tilde g} \big(-\frac{x_i}{dx_i}\omega^{(\tilde g)}_{1,0} (z_i)\big)-u_i \big(-\frac{x_i}{dx_i}\omega^{(0)}_{1,0} (z_i)\big)}
	\\ \notag &
	\prod_{e\in E(\Gamma)} \prod_{j=1}^{|e|
	} \restr{(\tilde u_j, \tilde x_j) }{u_{e(j)},x_{e(j)}} \tilde u_j \cS(\hbar \tilde u_j \tilde x_j \partial_{\tilde x_j})  \Big(-\frac{\tilde x_j}{d\tilde x_j}\Big) \sum_{\tilde g=0}^\infty \hbar^{2\tilde g}\tilde \omega^{(\tilde g)}_{|e|,0}(\tilde z_{\llbracket |e|\rrbracket})
	\\ \notag &
	+\delta_{n,1}[\hbar^{2g}] \sum_{k=0}^\infty (-y_1\partial_{y_1})^k [v^{k+1}]
	\\ \notag & \quad 
	\restr{\theta}{ - \frac{x_1}{dx_1}\omega^{(0)}_{1,0} (z_1)} e^{v\frac{\cS(\hbar v \partial_{\theta})}{\cS(\hbar \partial_{\theta})} \log\theta - v\log\theta} (-y_1\partial_{y_1})\Big(-\frac{x_1}{dx_1}\omega^{(0)}_{1,0} (z_1)\Big)
	\\ \notag &
	+ \delta_{(g,n),(0,1)} \Big(-\frac{x_1}{dx_1}\omega^{(0)}_{1,0} (z_1)\Big).
\end{align} 
Here
\begin{itemize}

	\item The sum is taken over all connected graphs $\Gamma$ with $n$ labeled vertices (with labels from $1$ to $n$) and multiedges of index $\geq 2$. 
	\item For convenience, for a given such graph, we also label all legs of every given multiedge $e$ from $1$ to $|e|$ in an arbitrary way.
	\item For a multiedge $e$ with index $|e|$  we control its attachment to the vertices by the associated map $e\colon \llbracket |e| \rrbracket \to 
	\llbracket n \rrbracket
	$ that we denote also by $e$, abusing notation (so $e(j)$ is the label of the vertex to which the $j$-th leg of the multiedge $e$ is attached). Do note that this map can be an arbitrary map from $\llbracket |e| \rrbracket$ to $\llbracket n \rrbracket$; in particular, it might not be injective, i.e. we allow a given multiedge to connect to a given vertex with several of its legs. 
	\item For a given multiedge $e$ with $|e|=2$ we define $\tilde \omega^{(0)}_{2,0}(\tilde x_1,\tilde x_2) :=  \omega^{(0)}_{2,0}(\tilde x_1,\tilde x_2) - \frac{d\tilde x_1d\tilde x_2}{(\tilde x_1-\tilde x_2)^2}$ if $e(1)=e(2)$, and $\tilde \omega^{(0)}_{2,0}(\tilde x_1,\tilde x_2) :=  \omega^{(0)}_{2,0}(\tilde x_1,\tilde x_2)$ otherwise. For all $(g,n)\not=(0,2)$ we simply have $\tilde \omega^{(g)}_{n,0} :=  \omega^{(g)}_{n,0}$.
	\item By $g(\Gamma)$ we denote the first Betti number of $\Gamma$.
	\item $|\mathrm{Aut}(\Gamma)|$ stands for the number of automorphisms of $\Gamma$.
	\item By $[\hbar^{2g}]$ (respectively, $[v_i^{k_i}]$, $[u_i^{r_i}]$) we denote the operator that extracts the corresponding coefficient from the whole expression to the right of it, that is,\\ $[x^m]\sum_{i=-\infty}^\infty a_ix^i \coloneqq a_m$. 
	\item By $\restr{a}{b}$ we denote the operator of substitution $a\to b$, that is, $\restr{a}{b}f(a)=f(b)$ for any function $f$. It is a bit unusual notation (sometimes one uses $f(a)|_{a\to b}$ instead), but it is very convenient for our purpose to consider the substitution as an operator and thus apply from the left as all other operators involved in the formula. 
	\item The function $\cS(z)$ is defined as
	\begin{align}
		\cS(z)\coloneqq \frac {e^{z/2}-e^{-z/2}} z.
	\end{align}
\end{itemize}

\begin{remark} \label{rem:MainFormula-DiffOper}For each $g\geq 0$, $n\geq 1$, Equation~\eqref{eq:MainFormula} is manifestly a finite sum of finite products of differential operators applied to $\omega^{(\tilde{g})}_{m,0}$ for $2\tilde g -2+m\geq 0$.
\end{remark}

\begin{remark} \label{rem:MainFormula-shifts} Conjecture~\ref{conj:Main} implies, for instance, that the universal formula, given by Equation~\eqref{eq:MainFormula}, should not change once we shift functions $x$ and $y$ by constants. However, this property is absolutely not clear from the formula --- the differential operators there involve the external differential $d$ and multiplications / divisions by $x$, $dx$, $y$, and $dy$.
\end{remark}

\begin{remark} The last line of Equation~\eqref{eq:MainFormula} is a regularization term for the next to the last line, and the last two lines together form a regularization term to the main part of the formula.
\end{remark}

\begin{remark}
	In the special case $(g,n)=(0,1)$
	Equation~\eqref{eq:MainFormula} implies that $-\omega^{(0)}_{0,1} \frac{y_1}{dy_1} = -\omega^{(0)}_{1,0} \frac{x_1}{dx_1}$. Since $ -\omega^{(0)}_{1,0} \frac{x_1}{dx_1}= x_1y_1$, we obtain $\omega^{(0)}_{0,1}=-x_1dy_1$. It is also straightforward to see from Equation~\eqref{eq:MainFormula} that $\omega^{(0)}_{0,2} = \omega^{(0)}_{2,0}$.	
\end{remark}

\subsection{A reformulation of the conjecture} Equation~\eqref{eq:MainFormula} admits a number of equivalent reformulations. For instance, it was ob\-served in \cite[Proposition 5.3]{Hockx-x-y} (and it also follows directly from~\cite[Lemma 4.1]{BDKS-toporec-KP}) that for $g=0$ Equation~\eqref{eq:MainFormula} can be substantially simplified. In fact it is the case for any $g$.

\begin{theorem}\label{thm:newformula} Equation~\eqref{eq:MainFormula} is equivalent to the following one:
	\begin{align} \label{eq:MainFormulaSimple}
		& \omega_{0,n}^{(g)} (z_{\llbracket n\rrbracket})\prod_{i=1}^n \frac{1}{dy_i} 
		=
		(-1)^n
		\coeff \hbar {2g} \sum_{\Gamma} \frac{\hbar^{2g(\Gamma)}}{|\mathrm{Aut}(\Gamma)|} \prod_{i=1}^n
		\sum_{k_i=0}^\infty  \partial_{y_i}^{k_i} [w_i^{k_i}] \frac{dx_i}{dy_i}
		\\ \notag &
		\frac{1}{w_i} e^{ w_i \cS(\hbar w_i \partial_{x_i}) \sum_{\tilde g=0}^\infty \hbar^{2\tilde g} \frac{1}{dx_i}\omega^{(\tilde g)}_{1,0} (z_i)-w_i \frac{1}{dx_i}\omega^{(0)}_{1,0} (z_i)}
		\\ \notag &
		\prod_{e\in E(\Gamma)} \prod_{j=1}^{|e|\geq 2}\restr{(\tilde w_j, \tilde x_j)}{ (w_{e(j)},x_{e(j)})} \tilde w_j \cS(\hbar \tilde w_j  \partial_{\tilde x_j}) \sum_{\tilde g=0}^\infty \hbar^{2\tilde g}\tilde \omega^{(\tilde g)}_{|e|,0}(\tilde z_{\llbracket |e|\rrbracket})  \prod_{j=1}^{|e|} \frac{1}{d\tilde x_j}
		\\ \notag &
		+\delta_{(g,n),(0,1)} (-x_1).
	\end{align}
\end{theorem}

A direct corollary of Theorems~\ref{thm:mainconjholds} and~\ref{thm:newformula} is the following statement:

\begin{theorem} \label{thm:simpleformulaxyswap} Equation~\eqref{eq:MainFormulaSimple} gives a universal formula for the $x-y$ swap. That is, assume the differentials $\omega^{(g)}_{m,0}$ satisfy the topological recursion for the input data $(S,x,y,B)$ with meromorphic $x$ and $y$, all zeros of $dy$ are simple, and $x$ is regular at the zeros of $dy$. Then the differentials $\omega^{(g)}_{0,n}$ given by~\eqref{eq:MainFormulaSimple} satisfy the topological recursion for the input data $(S,y,x,B)$.
\end{theorem}

A special case of Theorems~\ref{thm:mainconjholds} and~\ref{thm:simpleformulaxyswap} for $g=0$ is proved in~\cite{Hockx-x-y} under an additional assumption of the existence of the so-called loop insertion operators and their commutativity. We do not need this additional assumption. Special cases of this formula were also known for $(g,n)=(1,1)$ and $(g,n)=(1,2)$, see~ \cite{EynardOrantin-toporec,Hockx-x-y,BCGFLS-Free}.

While this paper was in preparation, Hock also independently discovered and proved Theorem~\ref{thm:newformula}, see~\cite{Hock-FullFormula}, as a generalization of his genus $0$ formula derived through the loop insertion operators. The methodological difference between our approach and the one taken by Hock is that we originally derived the equivalence of Equation~\eqref{eq:MainFormulaSimple} and~Equation~\eqref{eq:MainFormula} via a certain identity for vertex operators, see Section~\ref{sec:Computations-vertex-new-form}.

\begin{remark} For each $g\geq 0$, $n\geq 1$, Equation~\eqref{eq:MainFormulaSimple} is manifestly a finite sum of finite products of differential operators applied to $\omega^{(\tilde{g})}_{m,0}$ for $2\tilde g -2+m\geq 0$. This time, however, the differential operators involve only the external differential $d$ and multiplications / divisions by $dx$ and $dy$. Therefore, it is also straightforward to see that Equation~\eqref{eq:MainFormulaSimple} remains invariant under the shifts of $x$ and $y$ by constants, cf. Remarks~\ref{rem:MainFormula-DiffOper} and~\ref{rem:MainFormula-shifts}.
\end{remark}

\begin{remark} The last line of Equation~\eqref{eq:MainFormula} is again a regularization term for the the main part of the formula.
\end{remark}

\subsection{Mixed correlation differentials} The main tool to study the universal algebraic formulas that convert $\omega^{(g)}_{m,0}$'s into $\omega^{(g)}_{0,n}$'s is a system of mixed correlation differentials $\omega^{(g)}_{m,n}$ that we already mentioned in Remark~\ref{rem:MixedCorrDiff}. We associate such system of $\omega^{(g)}_{m,n}$'s to any input $(S,x,y,B)$, starting with the unstable cases
\begin{align} \label{eq:unstable}
	& \omega^{(0)}_{1,0} = -ydx; \qquad \omega^{(0)}_{0,1} = -xdy; \qquad \omega^{(0)}_{2,0} = -\omega^{(0)}_{1,1} = \omega^{(0)}_{0,2} = B.
\end{align}
This system of mixed correlation differentials is constructed from $\omega^{(g)}_{m,0}$'s via the formulas similar to Equations~\eqref{eq:MainFormula} and~\eqref{eq:MainFormulaSimple} (we have two equivalent formulations, see Section~\ref{sec:MixedDefinition}). We analyze the differential-algebraic relations between mixed correlation differentials, their singularities, and loop equations (see Section \ref{sec:higherloop}).
 
 Note that the explicit differential-algebraic relations and the pole analysis give us an algorithm to compute mixed correlation differentials, in particular, the original $\omega^{(g)}_{m,0}$, with\-out applying the topological recursion directly (see Section~\ref{sec:splitingofpoles} for details). This algorithm is in particular very efficient in the case of a rational spectral curve.

\subsubsection{Relation to the work of Eynard--Orantin} Conjecture~\ref{conj:Main} and Theorem~\ref{thm:mainconjholds} are not too surprising. Indeed, as it is discussed in~\cite{EynardOrantin-xysymmetry}, the formulas establishing the $x-y$ symmetry in general situation mimic in many ways the formulas for the Hermitian $2$-matrix model, cf.~\cite[Remarks~3.1 and~3.2]{EynardOrantin-xysymmetry}. On the other hand,  it is proved in~\cite{BDKS-symplectic} that Equation~\eqref{eq:MainFormula} gives the corresponding expression for the Hermitian $2$-matrix model.

In principle, this should give an alternative way to prove Theorem~\ref{thm:mainconjholds}. Indeed, in~\cite{EynardOrantin-xysymmetry} the $\omega^{( g)}_{m,0}$'s and $\omega^{(g)}_{0,n}$'s are also connected via a system of mixed correlation differentials $\omega^{(g)}_{m,n}$, and these mixed correlation differentials can most probably be identified with the mixed correlation differentials that we use in this paper (many properties are exactly the same). However, to work out the details to make this sketch of a proof into an actual proof is not so easy --- this would require an extension of the uniqueness arguments of the type given in~\cite[Appendix B]{EynardOrantin-xysymmetry} and the passage from a combination of loop equations and analytic properties as in~\cite[Section 1]{BEO-loop} and~\cite[Section 2]{BS-blobbed} to be extended to the more general situation of the mixed correlation functions.

\subsection{Review of the paper}
\newcommand{\oW}{{\overline W}}
The idea of this Section is that hopefully after reading it one will be able to use the results of the article (maybe without working out all the details of the proofs). All definitions and statements in this Section are explained in much more detail in the rest of the text, and the sole purpose of their presentation is to provide a comprehensive overview of the rest of the paper. This Section might be skipped entirely for the first reading, or, alternatively, be used as a guide that might simplify following the cross-references within the paper. 

We explain here the main ideas used in the proof of Theorems~\ref{thm:mainconjholds} and \ref{thm:simpleformulaxyswap}.  We start with the collection of universal functional identities that hold independently of the
topological recursion. Let $\{\omega^{(g)}_{m,0}\}$ be any collection of
symmetric meromorphic $m$-differentials  $\omega^{(g)}_{m,0}(z_{\set m})$ on $S^m$, $g\geq 0$, $m\geq 1$, such that $\omega^{(0)}_{1,0}=-ydx$ for some meromorphic functions $x$ and $y$ on $S$, and for $m\geq 2$ the differentials
\begin{align} \label{eq:assumption-diagonal-intro}
\omega^{(g)}_{m,0}(z_{\set{m}}) - \delta_{(g,m),(0,2)} \frac{dx(z_1)dx(z_2)}{(x(z_1)-x(z_2))^2}
\text{ have no poles on the diagonals}
\end{align}
(see Section \ref{sec:wmndef}). With these assumptions, we define the mixed $(m,n)$-differentials
\begin{equation}\label{eq:barW}
	\omega^{(g)}_{m,n}(z_1,\dots,z_m;z_{m+1},\dots,z_{m+n})=\oW^{(g)}_{m,n}(z_1,\dots,z_m;z_{m+1},\dots,z_{m+n})\prod_{i=1}^m dx_i\prod_{i=m+1}^{m+n}dy_i
\end{equation}
on $S^{m+n}$ for $n\ge1$ by the following recursive procedure. Denote $x=x_{m+1}=x(z_{m+1})$, $y=y_{m+1}=y(z_{m+1})$, $\partial_x=\frac{d}{dx}$, $\partial_y=\frac{d}{dy}$,
$M=\{1,\dots,m\}$, $N=\{m+2,\dots,m+n+1\}$.

Then, we define $\cT^x_{m+1,n}(u;z_{M};z;z_{N})$ and $\cW^{x,(g)}_{m+1,n}(u;z_{M};z;z_{N})$ (see Equations \eqref{eq:cT} and \eqref{Wgg})
\begin{multline}\label{eqIntro:cTx}
	\cT^x_{m+1,n}(u;z_{M};z;z_{N}):=
	\sum_{k=1}^\infty\frac{\hbar^{2(k-1)}u^k}{k!}\left(\prod_{i=1}^k
	\big\lfloor_{z_{\bar i}\to z}\cS(u\hbar \partial_{x_{\bar i}})\right)\\
	\sum_{g=0}^\infty\hbar^{2g}\bigl(
	\oW^{(g)}_{m+k,n}(z_{M},z_{\{\bar1,\dots,\bar k\}};z_{N})-
	\delta_{(g,m,n,k),(0,0,0,2)}\tfrac{1}{(x_{\bar1}-x_{\bar2})^2}\bigr),
\end{multline}
\begin{equation}\label{eqIntro:cWx}
	\cW^{x,(g)}_{m+1,n}(u;z):=
	\coeff \hbar {2g} \frac{e^{\cT^x_{1,0}(u;z)}}{u}\!\!\!\!\!\!
	\sum_{\substack{M\sqcup N=\sqcup_\alpha K_\alpha,~
			K_\alpha\ne\varnothing\\~I_\alpha=K_\alpha\cap M,~J_\alpha=K_\alpha\cap N}}
	\prod_{\alpha}\cT^x_{|I_\alpha|+1,|J_\alpha|}(u;z_{I_\alpha};z;z_{J_\alpha}),
\end{equation}
where the summation carries over the set of all partitions of the set $M\cup N$ into unordered collection of disjoint nonempty subsets $K_\alpha$, and where we denote
$I_\alpha=K_\alpha\cap M$, $J_\alpha=K_\alpha\cap N$. This sum has finitely many terms. With this notation, we set $\oW^{(0)}_{0,1}=-x$ and for all other triples $(g,m,n)$ we
define (see Equation \eqref{eq:Wxtoomega})
\begin{equation}\label{eqintro:Wxtoomega}
	\oW^{(g)}_{m,n+1}(z_{M};z,z_{N})=-\sum_{r\ge0} \partial_y^r\bigl(\tfrac{dx}{dy}[u^r]e^{u\,y}\cW^{x,(g)}_{m+1,n}(u,z)\bigr).
\end{equation}
Observe that $\cT^x_{1,0}=u\,\oW^{(0)}_{1,0}+O(\hbar^2)=-u\, y+O(\hbar^2)$ implying that the function $e^{u\,y}\cW^{x,(g)}_{m+1,n}$ is polynomial in $u$ for any triple
$(g,m,n)$ (with the exception $e^{u\,y}\cW^{x,(0)}_{1,0}(u;z)=\frac{1}{u}$). It follows that the sum on the right hand side of~\eqref{eqintro:Wxtoomega} has finitely many
nonzero terms and the functions $\oW^{(g)}_{m,n}$ defined by this iterative procedure are global meromorphic functions on $S^{m+n}$. The initial terms
$\oW^{(g)}_{m,0}$ of this recursion are given by~\eqref{eq:barW}.

Introduce also the dual functions obtained by the exchange of the role of~$x$ and~$y$,
\begin{multline}\label{eqIntro:cTy}
	\cT^y_{m,n+1}(u;z_{M};z;z_{N}):=
	\sum_{k=1}^\infty\frac{\hbar^{2(k-1)}u^k}{k!}\left(\prod_{i=1}^k
	\big\lfloor_{z_{\bar i}\to z}\cS(u\hbar \partial_{y_{\bar i}})\right)\\
	\sum_{g=0}^\infty\hbar^{2g}\bigl(
	\oW^{(g)}_{m,k+n}(z_{M};z_{\{\bar1,\dots,\bar k\}},z_{N})-
	\delta_{(g,m,n,k),(0,0,0,2)}\tfrac{1}{(y_{\bar1}-y_{\bar2})^2}\bigr),
\end{multline}
\begin{equation}\label{eqIntro:cWy}
\cW^{y,(g)}_{m,n+1}(u;z):=
	\coeff \hbar {2g} \frac{e^{\cT^y_{0,1}(u;z)}}{u}\!\!\!\!\!\!
	\sum_{\substack{M\cup N=\sqcup_\alpha K_\alpha,~
			K_\alpha\ne\varnothing\\~I_\alpha=K_\alpha\cap M,~J_\alpha=K_\alpha\cap N}}
	\prod_{\alpha}\cT^y_{|I_\alpha|,1+|J_\alpha|}(u;z_{I_\alpha};z;z_{J_\alpha}).
\end{equation}

\begin{theorem}\label{th:formalidentities} (see Propositions~\ref{prop:omega-g-n-graphsxy},~\ref{prop:regular-at-diagonals},~\ref{prop:xy-relations})

	1. The introduced functions $\oW^{(g)}_{m,n}$ and the corresponding differentials
	\begin{align}
		\omega^{(g)}_{m,n}=\oW^{(g)}_{m,n}\prod_{i=1}^m dx_i\prod_{i=m+1}^{m+n}dy_i
	\end{align}
are
	meromorphic, symmetric with respect to permutations of the arguments $z_1,\dots,z_m$, and, separately, with respect to permutations of the arguments $z_{m+1},\dots,z_{m+n}$ (and we use $x_i$ for $x(z_i)$ and $y_i$ for $y(z_i)$, as above). 
	Moreover, they have no poles on the diagonals $z_i=z_j$ such that either both~$i$ and~$j$ are smaller or equal to $m$ or both~$i$ and~$j$ are bigger than~$m$, with the exception
	\begin{equation}
		\omega^{(0)}_{2,0}(z_1,z_2)=-\omega^{(0)}_{1,1}(z_1;z_2)=\omega^{(0)}_{0,2}(z_1,z_2)
	\end{equation} 
	which has a pole (of order two) on the diagonal $z_1=z_2$. It follows that Equations~\eqref{eqIntro:cTy}, \eqref{eqIntro:cWy} make sense and the
	functions~$\cT^y_{m,n+1}$,~$\cW^{y,(g)}_{m,n+1}$ are well defined.
	
	2. In the case $m=0$ the forms $\omega^{(g)}_{0,n}$ are expressed in terms of the initial forms $\omega^{(g)}_{m,0}$ exactly by the graph summation formula~\eqref{eq:MainFormulaSimple}.
	
	3.  For $(m,n,g)\neq (0,0,0)$ the following identity 
	holds true (see Equation \eqref{eq:Wytoomega})
	\begin{equation}\label{eqintro:Wytoomega}
		\oW^{(g)}_{m+1,n}(z_{M},z;z_{N})=-\sum_{r\ge0} \partial_x^r\bigl(\tfrac{dy}{dx}[u^r]e^{u\,x}\cW^{y,(g)}_{m,n+1}(u;z)\bigr).
	\end{equation}
	
	4. For $(m,n,g)\neq (0,0,0)$ the following parametric extensions of the equalities~\eqref{eqintro:Wxtoomega} and~\eqref{eqintro:Wytoomega} also hold true (see Equations \eqref{eq:WxtoW}, \eqref{eq:WytoW})
	\begin{align}\label{eqintro:WxtoW}
		\cW^{y,(g)}_{m,n+1}(\tilde u;z)&=
		-\sum_{r\ge0} \partial_y^r\bigl(e^{-\tilde u\,x}\tfrac{dx}{dy}[u^r]e^{u\,y}\cW^{x,(g)}_{m+1,n}(u;z)\Bigr),\\
		\label{eqintro:WytoW}
		\cW^{x,(g)}_{m+1,n}(u;z)&=
		-\sum_{r\ge0} \partial_x^r\bigl(e^{-u\,y}\tfrac{dy}{dx}[{\tilde u}^r]e^{\tilde u\,x}\cW^{y,(g)}_{m,n+1}(\tilde u;z)\Bigr).
	\end{align}
\end{theorem}

Identities~\eqref{eqintro:WxtoW} and~\eqref{eqintro:WytoW} are understood as the equalities of the coefficients of any power of the parameter $\tilde u$ and $u$,
respectively, on both sides. In particular, taking the coefficients of $\tilde u^0$ and $u^0$, we obtain~\eqref{eqintro:Wxtoomega}
and~\eqref{eqintro:Wytoomega}, respectively.

We first prove a version of the identities of Theorem \ref{th:formalidentities} for the case when all functions involved are replaced by formal Laurent expansions by applying the well developed
formalism of vertex operators acting on the bosonic Fock space (see Proposition \ref{prop:XY-relations}). This covers a special case of Theorem  \ref{th:formalidentities} when there exists a point on the spectral curve where the function $y$
has a simple pole and $dx\neq 0$ (or, conversely, $x$~has a simple pole and $dy\ne0$ at this point). Since this restriction is an open condition, a general case can be
reduced to this one by applying suitable deformation arguments.

We expect that the identities of Theorem  \ref{th:formalidentities} admit a purely algebraic proof based on the combinatorics of their structure. In particular, it can be readily checked for
small $(g,m,n)$ using computer by expanding explicitly all the terms on both sides. However, such a proof would hardly ever explain the origin of these identities. In
opposite, they appear quite naturally in the context of vertex operators. Thus, we use the operator formalism both as a motivation to introduce these identities and also as a
tool in their proof in the context of global meromorphic functions.

\medskip
In the statement of Theorem \ref{th:formalidentities} we make no additional assumptions on the initial forms $\omega^{(g)}_{m,0}$. Assume, however, that they are chosen by applying the procedure of
topological recursion with the initial data $(S,x,y,B)$. Then $\omega^{(g)}_{0,n}$ solve the topological recursion with the spectral curve data $(S,y,x,B)$ with the functions~$x$ and~$y$ swapped (see Section \ref{sec:proofofmainthms}). In fact, this theorem has a stronger
form that we formulate now. It is well known that the topological recursion for the forms $\omega^{(g)}_{m,0}$ is essentially equivalent to certain relations called the
linear and quadratic loop equations (they are stated explicitly in Section~\ref{sec:RecollectionTR}). It is proved in \cite{DKPS-rspin} that along with the linear and quadratic loop equations the forms of topological recursion satisfy also a family of
the higher loop equations. We extend the notion of the higher loop equations to the case of mixed differentials~$\omega^{(g)}_{m,n}$ (see  Section \ref{sec:higherloop}).

\begin{definition} \label{def:FirstDefXi-space} We denote by $\Xi^x$ the space of meromorphic functions defined in a neighborhood of the zero locus of $dx$ on $S$ and spanned by the functions of the form
	$\partial_x^kf$ where $k=0,1,2,\dots$ and $f$ is holomorphic.  Similarly, we denote by $\Xi^y$ the space of meromorphic functions defined in a neighborhood of the zero locus
	of $dy$ and spanned by the functions of the form $\partial_y^kf$ where $k=0,1,2,\dots$ and $f$ is holomorphic.
\end{definition}

\begin{definition}	
	We say that the forms $\omega^{(g)}_{m,n}$ satisfy the \emph{$r$-loop equation} with respect to the variable~$z_i$ if the coefficient of $u^{r-1}$ in $\cW^{x,(g)}_{m,n}(u)$
	regarded as a meromorphic function in the $i$th argument $z_i$ belongs to $\Xi^x$ if $i\le m$ or to the space $\Xi^y$ if $i>m$, respectively. The loop equations impose
	certain restrictions on the principal parts of the poles of $\omega^{(g)}_{m,n}$. 
\end{definition}

This definition generalizes in a natural way the well-known linear and quadratic loop equations that we recall in Section~\ref{sec:RecollectionTR}. To our best knowledge, the way we define the $r$-loop equations here is not present in the literature, and we discuss this new concept and explain its origin in Section~\ref{sec:higherloop}.

\begin{theorem}\label{th:loopeq} (see Proposition \ref{prop:loopequations})
	
	Assume that the differentials $\omega^{(g)}_{m,n}$ of Theorem~\ref{th:formalidentities} are obtained from the initial forms $\omega^{(g)}_{m,0}$ satisfying the topological recursion with
	the spectral curve data $(S,x,y,B)$. Then
	\begin{enumerate}
		\item $\omega^{(g)}_{m,n}$ is holomorphic in $z_i$ at zeros of $dy$ for $i=1,\dots,m$ and is holomorphic at zeros of $dx$ for $i=m+1,\dots,m+n$;
		\item $\omega^{(g)}_{m,n}$ satisfy the $r$-loop equations with respect to all its arguments and for all $r\ge1$.
	\end{enumerate}
\end{theorem}
We derive the statement of Theorem \ref{th:loopeq} from identities~\eqref{eqintro:WxtoW} and~\eqref{eqintro:WytoW} and the linear and quadratic loop equations for the original differentials $\omega^{(g)}_{m,0}$  (the loop equations are recalled in Section~\ref{sec:RecollectionTR}). This gives a new independent proof of \cite[Proposition~3.3]{DKPS-rspin}.

We finish the proof of Theorem~\ref{thm:simpleformulaxyswap} by the following arguments. Consider the forms $\omega^{(g)}_{0,n}$ defined by the formula of summation over graphs. All possible poles of
the terms entering this formula are either at zeros of $dx_i$ (they are present for the initial forms $\omega^{(g)}_{m,0}$), or at the diagonals $z_i=z_j$ (due to presence of
the singular form $\omega^{(0)}_{2,0}$), or at zeros of $dy_i$ (appearing after application of the operator $\partial_{y_i}$). The poles on the diagonals cancel out by
Theorem~\ref{th:formalidentities}. The poles at zeros of $dx_i$ cancel out by Theorem~\ref{th:loopeq}. The only poles that survive are those at zeros of $dy_i$. Moreover,
the principal parts of the poles at zeros of $dy_i$ satisfy all loop equations, including the linear and quadratic ones. This implies the topological recursion for the forms
$\omega^{(g)}_{0,n}$.

\begin{remark} It is important to warn the reader that the mixed $(m,n)$-point differentials $\omega_{m,n}^{(g)}$ and mixed $(m,n)$-point functions $W_{m,n}^{(g)}$ of the present paper are quite different from the ones discussed in~\cite{BDKS-FullySimple} and \cite{BDKS-symplectic}. In this paper we are beyond formal Laurent expansions of rational functions and discuss relations between globally defined meromorphic differentials. However, in some steps of the proofs we do get down to Laurent series expansions (without any assumption that these Laurent series are expansions of rational functions) and there we do use techniques from \cite[Sections 4.7, 4.8]{BDKS-FullySimple}.
\end{remark}

\subsection{Organization of the paper}

Let us note that in this paper we work in three different setups simultaneously. Firstly, to discuss topological recursion we need to assume some conditions on functions $x$ and $y$, see Section~\ref{sec:toprecintro}. Secondly, while talking about VEV's, formal power series and Theorem~\ref{thm:mainconjholds} we need some additional algebraic assumptions which are given in Section~\ref{sec:FormalPowerSeries}. And thirdly, new recursions on mixed $(m,n)$-point differentials and their other properties do not need any assumptions, except for the ones given in Section~\ref{sec:wmndef}.

The proof of Conjecture~\ref{conj:Main} is generally based on the equivalence of two recursions for the differentials $\omega_{m,n}^{(g)}$, that is Theorem~\ref{thm:newformula}. So we have two types of expressions for the $\omega_{m,n}^{(g)}$. The first type 
comes from our previous papers, see equations \eqref{eq:MainFormula}, \eqref{eq:mainrecchange}. The second one is in fact the far generalization of Hock's formula for the genus $0$ case, and  we call them alternative, see Equations \eqref{eq:MainFormulaSimple}, \eqref{eq:Wxtoomega} (A.~Hock has independently proved the equivalence of these two types of expressions for general $g$ in~\cite{Hock-FullFormula} while the present paper was in preparation).


In Section~\ref{sec:MixedDefinition} we introduce the key objects of the paper: mixed disconnected and connected $(m,n)$-point differentials in the VEV formalism in Section~\ref{sec:MotivationFPS}, and in the general setting in Section~\ref{sec:wmndef}. We state here Proposition~\ref{prop:simple-rec-mixed} which claims that two types of 
expressions for the differentials $\omega_{m,n}^{(g)}$ coincide and which implies Theorem~\ref{thm:newformula}.

Section~\ref{sec:proof-of-algebraic-equivalence-of-relations} is totally devoted to the proving of Proposition~\ref{prop:MainSimplification}, which implies Proposition~\ref{prop:simple-rec-mixed}.

In Section~\ref{sec:proofs} we study mixed correlation differentials in the most general case. In Section~\ref{sec:diagonalsproofs} we provide the formulas for them using the summation over certain sets of graphs which use later in the proof of topological recursion for dual differentials $\omega_{0,n}^{(g)}$.

In Section~\ref{sec:proofs1} we reduce the proofs of Propositions~\ref{prop:regular-at-diagonals},~\ref{prop:XY-relations} and~\ref{prop:xy-relations} to the case of formal power series and VEV's setup, which we are dealing with in Section~\ref{sec:vertex-operators-standard-E}.

In Section~\ref{sec:TopologicalRecursion} we prove that if the differentials $\omega^{(g)}_{m,0}$ satisfy topological recursion then the differentials $\omega^{(g)}_{0,n}$ do so as well. As a consequence, we obtain proofs of Theorems~\ref{thm:mainconjholds} and~\ref{thm:simpleformulaxyswap}.

In Section~\ref{sec:examples} we provide explicit formulas for $\omega^{(g)}_{m,n}$ for a few small $g$, $m$, and $n$.

In Section~\ref{sec:applications} we discuss certain applications of Theorem~\ref{thm:mainconjholds}. In particular, in Theorem~\ref{cor:y=z} we give an explicit closed expression for the differentials $\omega^{(g)}_{m,0}$ produced by the topological recursion on the spectral curve $x=x(z),\; y=z$ for an arbitrary rational function $x(z)$.



\subsection{Acknowledgments}
P.~D.-B. is grateful to the University of Haifa for hospitality.

A.~A. was supported by the Institute for Basic Science (IBS-R003-D1).
Research of B.~B. was supported by the ISF grant 876/20. M.K. was supported by the International Laboratory of Cluster Geometry NRU HSE, RF Government grant, ag. № 075-15-2021-608 dated 08.06.2021. S.~S. was supported by the Netherlands Organization for Scientific Research.

We would like to thank the referees for useful remarks and suggestions.


\section{Mixed correlation differentials}

\label{sec:MixedDefinition}

\subsection{Definition} We define the mixed correlation differentials $\omega^{(g)}_{m,n}$ for $n\geq 1$ in terms of $\omega^{(g)}_{m,0}$'s. The definition is obtained in exactly the same way as Equation~\eqref{eq:MainFormula}: we consider some formal power series computation (under exactly the same assumptions as in Paragraph~\ref{sec:FormalPowerSeries}) and we extract from it a universal differential-algebraic formula that serves as a recursive definition of $\omega^{(g)}_{m,n}$'s, $n\geq 1$. There is also an equivalent formula of the same type as the one given in Equation~\eqref{eq:MainFormulaSimple} and it can be used as an alternative definition of $\omega^{(g)}_{m,n}$'s, $n\geq 1$.

We start with a motivation, with the actual definition following in Section~\ref{sec:wmndef}.

\subsubsection{Motivation from the formal power series} \label{sec:MotivationFPS} Recall the setup of  Paragraph~\ref{sec:FormalPowerSeries}. Consider a system of functions $W^{(g)}_{m,n}$ defined as
\begin{align} \label{eq:MixedFirstDef}
	W_{m,n}\coloneqq \sum_{g=0}^{\infty} \hbar^{2g} W^{(g)}_{m,n} \coloneqq \hbar^{2-m-n} \VEVc{\prod_{i=1}^m \sum_{k_i\in\Z} x_i^{k_i} J_{k_i}\;\cD^{-1}\;\prod_{i=m+1}^{m+n} \sum_{\ell_i\in\Z} y_i^{-\ell_i}   J_{\ell_i}\;\cD\;Z},
\end{align}
where $Z$ is given by~\eqref{def:Z}
(cf.~Equation~\eqref{eq:UExpr-FormalPowerSeries}; note however that now we take the summation over $k_i,\ell_i\in\Z$).

\begin{remark}
	It is important to comment here what we mean by the connected vacuum expectation values. Let
	\begin{align}\label{JJdef}
		J(z)& \coloneqq \sum_{\ell\in\Z} z^{\ell} J_{\ell}; \qquad \text{and} \qquad \J(\hbar,z)\coloneqq \cD^{-1} J(z)\cD.
	\end{align}
	Then the right hand side of Equation~\eqref{eq:MixedFirstDef} can be rewritten as
	\begin{align} 
		\hbar^{2-m-n}\prod_{i=1}^m [u_i^0] \prod_{i=m+1}^{m+n} \restr {u_i} 1 \VEVc{\prod_{i=1}^m
			\J(u_i\hbar ,x_i) \prod_{i=m+1}^{m+n}  \J(u_i\hbar,y_i^{-1}) \;Z},
	\end{align}
	where, as we have already defined above, $\restr u 1$ denotes the operator of a formal substitution $u\to 1$ and $\coeff u a$ denotes the operator of extracting the coefficient of $u^a$ in a Laurent series in $u$ next to it.
	
 	Now, the connected vacuum expectation value 
 	\begin{align}
 		\VEVc{\prod_{i=1}^m
 			\J(u_i\hbar ,x_i) \prod_{i=m+1}^{m+n}  \J(u_i \hbar,y_i^{-1}) \;Z}
 	\end{align} 
 is defined by the standard  inclusion-exclusion formula (see e.~g.~\cite[Equation~(71)]{BDKS-toporec-KP}):
	\begin{align}
		\VEVc{\prod_{i=1}^{m+n} \J_i Z} \coloneqq
		\sum_{\ell=1}^{m+n} \frac{(-1)^{\ell-1}}{\ell}
		\sum_{\substack{I_1\sqcup\cdots\sqcup I_\ell = \set{n+m} \\ \forall j \; I_j\not=\emptyset}} \prod_{j=1}^\ell \VEV{ \J_{I_j} Z},
	\end{align}
	or, equivalently
	\begin{align}
		\VEV{\prod_{i=1}^{m+n} \J_i Z} =
		\sum_{\ell=1}^{m+n} \frac{1}{\ell!}
		\sum_{\substack{I_1\sqcup\cdots\sqcup I_\ell = \set{n+m} \\ \forall j \; I_j\not=\emptyset}} \prod_{j=1}^\ell \VEVc{ \J_{I_j} Z},
	\end{align}
	where $\J_I$ for $I=\{i_1<\cdots < i_p\}\subseteq \set{m+n}$ denotes $\J_{i_1}\cdots \J_{i_p}$, and we use $\J_i = \J(u_i\hbar,x_i)$ for $i=1,\dots,m$ and $\J_i = \J(u_i\hbar,y_i^{-1})$ for $i=m+1,\dots,m+n$.
\end{remark}

There is a differential-algebraic expression that expresses $W^{(g)}_{m,n+1}$ in terms of $W^{(g)}_{m+1,n}$ and $W^{(g')}_{m',n'}$ with $g'\leq g$, $n'\leq n$, and $2g'+n'+m'\leq 2g+n+m$~\cite{BDKS-symplectic}. Let us recall it.

Let $M=\set{m}$ and $N=\set{m+n+1}\setminus \set{m+1}$. It is convenient to denote $x=x_{m+1}$ and $y=y_{m+1}$. Let
\begin{align} \label{eq:cTX}
	\cT^X_{m+1,n}(u;x_M;x;y_N) &\coloneqq
	\sum_{k=1}^\infty\frac{\hbar^{2(k-1)}}{k!}\left(\prod_{i=1}^k
	\big\lfloor_{x_{\bar i}\to x}u\,\cS(u\hbar x_{\bar i}\partial_{x_{\bar i}})\right)
	\\ \notag & \quad
	\left(
	W_{m+k,n}(x_{M},x_{\{\bar1,\dots,\bar k\}};y_{N})-
	\delta_{(m,k,n),(0,2,0)}\frac{x_{\bar1}x_{\bar2}}{(x_{\bar1}-x_{\bar2})^2}\right),
	\\ \label{eq:cWX}
	\cW^X_{m+1,n}(u;x) & \coloneqq \frac{e^{\cT^X_{1,0}(u;x)}}{u\cS(u\,\hbar)}
	\sum_{\substack{\sqcup_\alpha K_\alpha=M\cup N,~
			K_\alpha\ne\varnothing,\\~I_\alpha=K_\alpha\cap M,~J_\alpha=K_\alpha\cap N}}
	\prod_{\alpha}\cT^X_{|I_\alpha|+1,|J_\alpha|}(u;x_{I_\alpha};x;y_{J_\alpha}).
\end{align}
In order to shorten the notation, we omit dependencies on $x_{\set m},y_{\set{m+n+1}\setminus\set{m+1}}$, and $\hbar$ for $\cW^X_{m+1,n}$ and on $\hbar$ for $\cT^X_{m+1,n}$, respectively, in the lists of their arguments. With the same convention applied to $W_{m,n+1}$, we have the following proposition:

\begin{proposition} \label{prop:FirstFormOfrecursion} We have:
	\begin{align}\label{eq:mainrecchange}
		W_{m,n+1}(y) & =-\sum_{r,j=0}^\infty(-y\partial_y)^j\Bigl(
		[v^j]\Big\lfloor_{\theta \to xy}\frac{\partial_\theta^r e^{v\frac{\cS(v\hbar\partial_\theta)}{\cS(\hbar\partial_\theta)}\log\theta}}{\theta^v}
		\frac{dx}{x}\frac{y}{dy} \\ \notag & \qquad \qquad \qquad \qquad 
		\cdot [u^r]e^{-u\,W^{(0)}_{1,0}(x)}\cW^X_{m+1,n}(u;x)\Bigr)\\ \notag & \quad
		{}+\delta_{m+n,0}\biggl(W^{(0)}_{1,0}(x)+ \sum_{j=0}^\infty (-y\partial_{y})^j [v^{j+1}] \Big\lfloor_{\theta \to x\,y} \frac{ e^{v\frac{\cS(v\hbar \partial_{\theta})}{\cS(\hbar \partial_{\theta})} \log\theta}}{\theta^v}(-y\partial_{y}(xy))\biggr),
	\end{align}	 
	where $x$ and $y$ are related by
	\begin{equation}\label{eq:xychange}
		y=x^{-1}(1+W^{(0)}_{1,0}(x))
	\end{equation}
	(at $x=0$, $y=\infty$).
\end{proposition}

This proposition is a straightforward corollary of the technique developed in~\cite{BDKS-OrlovScherbin,BDKS-toporec-KP,BDKS-FullySimple,BDKS-symplectic}. For completeness, we recall the key steps of the proof and give precise references in Section~\ref{sec:vertex-operators-standard-E}.

\begin{remark}
	It is important to notice that on the right hand side of~\eqref{eq:mainrecchange} we have a formal power series in $\hbar^2$, whose coefficients are differential-algebraic expressions in $W^{(g')}_{m',n'}$ with finitely many nonzero terms. 	
\end{remark}

\begin{remark} \label{rem:mainrecexpressionomega} Let
	\begin{align}\label{eq:omega-W}
		\omega^{(g)}_{m,n}\coloneqq \Bigl(W^{(g)}_{m,n}+\delta_{(g,m+n),(0,1)}\Bigr) 	\prod_{i=1}^m \Bigl(-\frac{dx_i}{x_i}\Bigr)
		\prod_{i=m+1}^{m+n}  \Bigl(-\frac{dy_i}{y_i}\Bigr).
	\end{align}
	Then Equation~\eqref{eq:mainrecchange} can be rewritten as a differential-algebraic expression that expresses $\omega^{(g)}_{m,n+1}$ in terms of $\omega^{(g)}_{m+1,n}$ and $\omega^{(g')}_{m',n'}$ with $g'\leq g$, $n'\leq n$, and $2g'+n'+m'\leq 2g+n+m$ (see the next subsection).
\end{remark}

\begin{remark} Perhaps, it would be more natural to connect $W$'s and $\omega$'s by
	\begin{align}\label{eq:MoreNaturalConventionW-omega}
		\omega^{(g)}_{m,n}\coloneqq W^{(g)}_{m,n}	\prod_{i=1}^m \Bigl(-\frac{dx_i}{x_i}\Bigr)
		\prod_{i=m+1}^{m+n}  \Bigl(-\frac{dy_i}{y_i}\Bigr).
	\end{align}
	Since the functions $W^{(0)}_{1,0}(x)$ and $W^{(0)}_{0,1}(y)$ are identified through Equations~\eqref{eq:mainrecchange} and~\eqref{eq:xychange} as
	\begin{equation}\label{eq:W01=xy-1}
		W^{(0)}_{0,1}(y)=W^{(0)}_{1,0}(x)=x\,y-1,
	\end{equation}
	Convention~\eqref{eq:MoreNaturalConventionW-omega} would imply
	\begin{equation}
		\omega^{(0)}_{1,0}=-y\,dx+\frac{dx}{x},\qquad \omega^{(0)}_{0,1}=-x\,dy+\frac{dy}{y}.
	\end{equation}
	Thus this choice disagrees with the convention on the forms $\omega^{(0)}_{1,0}$ and $\omega^{(0)}_{0,1}$ chosen in the Introduction. Note, however, that the difference is not essential. Adding $dx/x$ to $\omega^{(0)}_{1,0}$, that is, adding a constant to  $W^{(0)}_{1,0}$ does not affect the derivatives of $W^{(0)}_{1,0}$. Hence, it does not change the term $e^{-u\,W^{(0)}_{1,0}(x)}\cW_{m+1,n}(u,x)$ in~\eqref{eq:mainrecchange}, and the whole Equation~\eqref{eq:mainrecchange} is preserved with this modification of $\omega^{(0)}_{1,0}$. Do note, however, that in what follows we do not use $W$'s anywhere except for Sections~\ref{sec:vertex-operators-standard-E} and~\ref{sec:Computations-vertex-new-form}, and there it is important that we use relation~\ref{eq:omega-W} rather than~\ref{eq:MoreNaturalConventionW-omega}, otherwise the expressions in these sections would be slightly different. 
\end{remark}

\begin{remark} Note that we don't need any correction for the terms $(g,m+n)=(0,2)$ (cf.~Equation~\eqref{eq:UExpr-Algebraic} where this correction was needed). The  double poles are added automatically by extending the summation over $k_i,\ell_i$ in Equation~\eqref{eq:MixedFirstDef} to $k_i,\ell_i\in\Z$.
\end{remark}

\subsubsection{Definition of \texorpdfstring{$\omega^{(g)}_{m,n}$}{mixed correlation differentials}}\label{sec:wmndef} In this section we want to take a look at Equation~\eqref{eq:mainrecchange} in an abstract way, using it as a basis for a new construction and completely ignoring its origin. So, from now on we consider Equation~\eqref{eq:mainrecchange} as differential-algebraic expression that expresses some abstract $\omega^{(g)}_{m,n+1}$ in terms of given $\omega^{(g)}_{m+1,n}$ and $\omega^{(g')}_{m',n'}$ with $g'\leq g$, $n'\leq n$, and $2g'+n'+m'\leq 2g+n+m$ (according to Remark~\ref{rem:mainrecexpressionomega}, Equation~\eqref{eq:omega-W}), with no further assumptions on the origin of $\omega^{(g)}_{m,n}$'s.

In other words, we forget all assumptions that we made in order to define and explain the origin of Equation~\eqref{eq:mainrecchange} and apply this expression directly to
\emph{an arbitrary system of globally defined symmetric meromorphic $m$-differentials $\omega^{(g)}_{m,0}(z_{\set m})$ on $S^m$, $g\geq 0$, $m\geq 1$, such that $\omega^{(0)}_{1,0}=-ydx$ for some meromorphic functions $x$ and $y$ on $S$, and for $m\geq 2$ the differentials}
\begin{align} \label{eq:assumption-diagonal}
	\omega^{(g)}_{m,0}(z_{\set{m}}) - \delta_{(g,m),(0,2)} \frac{dx(z_1)dx(z_2)}{(x(z_1)-x(z_2))^2}
	\text{ have no poles on the diagonals.}
\end{align}
As a result, we obtain a family of  $(m+n)$-differentials $\omega^{(g)}_{m,n}(z_{\set {m+n}})$ on $S^{m+n}$, $g\geq 0$, $m+n\geq 1$.

Let us write this recursive definition of $\omega^{(g)}_{m,n}$ explicitly. Let $dX = -dx/x$, $dY=-dy/y$, $\partial_X=-x\partial_x$, $\partial_Y=-y\partial_y$. Define
\begin{align}\label{eq:cTXom}
	\cT^X_{m+1,n}(u;z_{M};z;z_{N})& \coloneqq
	\sum_{k=1}^\infty\frac{\hbar^{2(k-1)}u^k}{k!}\left(\prod_{i=1}^k
	\big\lfloor_{z_{\bar i}\to z}\cS(u\hbar \partial_{X_{\bar i}})\right)\\ \notag & \quad
	\sum_{g=0}^\infty\hbar^{2g}\frac{
		\omega^{(g)}_{m+k,n}(z_{M},z_{\{\bar1,\dots,\bar k\}};z_{N})-
		\delta_{(g,m,k,n),(0,0,2,0)}\frac{dx_{\bar1}dx_{\bar2}}{(x_{\bar1}-x_{\bar2})^2}}
	{\prod_{i=1}^m dX_i \prod_{i=1}^k dX_{\bar i}\prod_{i=m+2}^{m+n+1} dY_i},
	\\
	\label{eq:cWexpXom}
	\cW^X_{m+1,n}(u;z)& \coloneqq \frac{~e^{\cT^X_{1,0}(u;z)}}{u\,\cS(u\,\hbar)}
	\sum_{\substack{M\cup N=\sqcup_\alpha K_\alpha~
			K_\alpha\ne\varnothing\\~I_\alpha=K_\alpha\cap M,~J_\alpha=K_\alpha\cap N}}
	\prod_{\alpha}\cT^X_{|I_\alpha|+1,|J_\alpha|}(u;z_{I_\alpha};z;z_{J_\alpha})
\end{align}
(these are the same expressions as given by Equations~\eqref{eq:cTX}  and~\eqref{eq:cWX} but now used in this new, more general situation). In order to shorten the notation, we omit dependencies on $x_{\set m},y_{\set{m+n+1}\setminus\set{m+1}}$, and $\hbar$ for $\cW^X_{m+1,n}$ and on $\hbar$ for $\cT^X_{m+1,n}$, respectively, in the lists of their arguments.
Denote
\begin{align} \label{eq:Theta}
	\Theta(z)&=\frac{\omega^{(0)}_{1,0}}{dX}=\frac{\omega^{(0)}_{0,1}}{dY}=x(z)\,y(z)
	,\\ \label{eq:L0}
	L_0(v,\theta)&=e^{v\bigl(\frac{\cS(v\hbar\partial_\theta)}{\cS(\hbar\partial_\theta)}-1\bigr)\log\theta},\\ \label{eq:Lr}
	L_r(v,\theta)&=\theta^{-v}\partial_\theta^r\theta^{v}L_0(v,\theta)=(\partial_\theta+\tfrac{v}{\theta})^rL_0(v,\theta).
\end{align}
	Then we have:
	\begin{align}\label{eq:WXtoomega}
		& \frac{\omega^{(g)}_{m,n+1}}{\prod_{i\in M} dX_i \cdot dY \cdot \prod_{i\in N} dY_i} =
		\\ \notag 
		& 
		[\hbar^{2g}]
		\sum_{j\ge0} \partial_Y^j
		[v^j]  \Bigl(-\sum_{r\ge0}L_r(v,\Theta)[u^r]\frac{dX}{dY} e^{-u\,\Theta}\cW^{X}_{m+1,n}(u;z)
		\\ \notag
		& \qquad \qquad \qquad \qquad
		+\delta_{m+n,0}\frac{L_0(v,\Theta)}{v}\frac{d\Theta}{dY}
		+\delta_{m+n,0} \Theta \Bigr),
	\end{align} 
	where $\Theta=\Theta(z)$.
\begin{remark}
	Formulas~\eqref{eq:cTXom}, \eqref{eq:cWexpXom}, and~\eqref{eq:WXtoomega} coincide with~\eqref{eq:cTX}, \eqref{eq:cWX}, and~\eqref{eq:mainrecchange}, respectively, if one expresses $\omega^{(g)}_{m,n}$'s via $W^{(g)}_{m,n}$'s according to~\eqref{eq:omega-W}.
\end{remark}

 \begin{definition} \label{def:mixedcorrdiff} Define $\omega^{(g)}_{m,n}$ to be the 
differentials which are recursively obtained via \eqref{eq:WXtoomega} from an \emph{arbitrary} set of globally defined symmetric meromorphic $m$-dif\-fe\-ren\-tials $\omega^{(g)}_{m,0}(z_{\set m})$ on $S^m$, $g\geq 0$, $m\geq 1$, and two meromorphic functions $x$ and $y$, such that $\omega^{(0)}_{1,0}=-ydx$ and such that~\eqref{eq:assumption-diagonal} holds.
\end{definition} 

Let us stress that
so far we do not connect Definition~\ref{def:mixedcorrdiff} to the theory of topological recursion.

\subsection{Alternative expression}\label{sec:altexpr}
Let $M=\set{m}$ and $N=\set{m+n+1}\setminus \set{m+1}$. We consider variables $z_{\set{m+n+1}}$ on $S^{m+n+1}$, and we denote $x_i\coloneqq x(z_i)$ and $y_i\coloneqq y(z_i)$. It is convenient to denote $z=z_{m+1}$, $x=x_{m+1}$ and $y=y_{m+1}$. Define
\begin{align}\label{eq:cT}
	\cT^x_{m+1,n}(w;z_{M};z;z_{N}) & \coloneqq
	\sum_{k=1}^\infty\frac{\hbar^{2(k-1)}w^k}{k!}\left(\prod_{i=1}^k
	\big\lfloor_{z_{\bar i}\to z}\cS(w\hbar \partial_{x_{\bar i}})\right)
	\\ \notag & \quad
	\sum_{g=0}^\infty\hbar^{2g}\frac{
		\omega^{(g)}_{m+k,n}(z_{M},z_{\{\bar1,\dots,\bar k\}};z_{N})-
		\delta_{(g,m,k,n),(0,0,2,0)}\frac{dx_{\bar1}dx_{\bar2}}{(x_{\bar1}-x_{\bar2})^2}}
	{\prod_{i=1}^m dx_i\prod_{i=1}^k dx_{\bar i}\prod_{i=m+2}^{m+n+1}dy_i},
	\\
	\label{eq:cW}
	\cW^x_{m+1,n}(w;z) & \coloneqq 
	\frac{1}{w}
	\sum_{k=0}^\infty\frac{1}{k!}
	\sum_{\substack{M=\sqcup_{i=1}^kI_i\\N=\sqcup_{i=1}^kJ_i}}
	\prod_{i=1}^k\cT^x_{|I_i|+1,|J_i|}(w;z_{I_i};z;z_{J_i})
	\\ \notag 
	& = \frac{1}{w}~e^{\cT^x_{1,0}(w;z)}\!\!\!\!\!\!
	\sum_{\substack{M\cup N=\sqcup_\alpha K_\alpha~
			K_\alpha\ne\varnothing\\~I_\alpha=K_\alpha\cap M,~J_\alpha=K_\alpha\cap N}}
	\prod_{\alpha}\cT^x_{|I_\alpha|+1,|J_\alpha|}(w;z_{I_\alpha};z;z_{J_\alpha}).
\end{align}
In order to shorten the notation, we omit dependencies on $x_M ,y_N$, and $\hbar$ for $\cW^x_{m+1,n}$ and on $\hbar$ for $\cT^x_{m+1,n}$, respectively, in the lists of their arguments.

Several remarks are in order.

\begin{remark} The regularizing correction in~\eqref{eq:cT} is used for the restriction to the diagonal $z_{\bar1}=z_{\bar2}=z$ being correctly defined: while $\omega^{(0)}_{2,0}$ is singular, its regularization
	\begin{equation}
		\omega^{(0)}_{2,0}(z_{\bar1},z_{\bar2})-\frac{dx_{\bar1}dx_{\bar2}}{(x_{\bar1}-x_{\bar2})^2}
	\end{equation}
	extends to the diagonal holomorphically.
\end{remark}

\begin{remark} By construction, 
\begin{equation}\label{Wgg}
\cW^{x,(g)}_{m+1,n}\coloneqq \coeff \hbar {2g} \cW^{x}_{m+1,n}
\end{equation}
 is a meromorphic function in $z$ on $S$. Since $\cT^x_{1,0}=w\frac{\omega^{(0)}_{1,0}}{dx}+O(\hbar^2)=-w y+O(\hbar^2)$, we conclude that the combination $e^{w\,y}\cW^{x,(g)}_{m+1,n}$ is polynomial in $w$ for any triple $(g,m,n)$ with the exception of the term $e^{w\,y}\cW^{x,(0)}_{1,0}(w;z)=\frac{1}{w}$.
\end{remark}

\begin{proposition}\label{prop:simple-rec-mixed} For any triple $(g,m,n)\ne(0,0,0)$, the mixed correlation differentials $\omega^{(g)}_{m,n}$ given by Definition~\ref{def:mixedcorrdiff} satisfy
	\begin{equation}\label{eq:Wxtoomega}
		\frac{\omega^{(g)}_{m,n+1}(z_M;z,z_N)}{\prod_{i\in M} dx_i \cdot dy\cdot \prod_{i\in N} dy_i} =-\sum_{r\ge0} \partial_y^r [w^r] \frac{dx}{dy} e^{w\,y}\cW^{x,(g)}_{m+1,n}
	\end{equation}
	(recall that functions $x$, $y$, and the differential $\omega^{(0)}_{1,0}$ are related by the condition $\omega^{(0)}_{1,0}=-ydx$).
\end{proposition}

\begin{remark} Equation~\eqref{eq:Wxtoomega} expresses $\omega^{(g)}_{m,n+1}$ in terms of $\omega^{(g)}_{m+1,n}$ and $\omega^{(g')}_{m',n'}$ with $g'\leq g$, $n'\leq n$, and $2g'+n'+m'\leq 2g+n+m$. A direct implication of Proposition~\ref{prop:simple-rec-mixed} is that  Equation~\eqref{eq:Wxtoomega}  can be used as an equivalent alternative to Definition~\ref{def:mixedcorrdiff}.
\end{remark}

We prove Proposition~\ref{prop:simple-rec-mixed} in Section~\ref{sec:proof-of-algebraic-equivalence-of-relations}.

\subsection{Proof of Theorem~\ref{thm:newformula}}
It is a direct corollary of Proposition~\ref{prop:simple-rec-mixed}. Indeed, Equation~\eqref{eq:MainFormula} (respectively, Equation~\eqref{eq:MainFormulaSimple}) that expresses $\omega^{(g)}_{0,n}$'s in terms of $\omega^{(g)}_{m,0}$'s is obtained by iterative application of Equation~\eqref{eq:mainrecchange} (respectively, Equation~\eqref{eq:Wxtoomega}). The iteration procedure works as follows (we describe it below for Equation~\eqref{eq:MainFormulaSimple}, the case of Equation~\eqref{eq:MainFormula} is fully analogous):

Abusing the language a little bit, we can say that Equation~\eqref{eq:MainFormulaSimple} expresses a  correlator with $n$ `$dy$-arguments' in terms of the correlators with only `$dx$-arguments'. On the other hand, Equation~\eqref{eq:Wxtoomega} expresses the correlators with $m$ `$dx$-arguments' and $(n+1)$ `$dy$-arguments' in terms of the correlators with 
\begin{itemize}
	\item either the only term with same Euler characteristic $2-2g-m-n-1$ but then with $(m+1)$ `$dx$-arguments' and $n$ `$dy$-arguments' (thus one special `$dx$-argument' on the right hand side is matched to the distinguished `$dy$-argument' specified on the left hand side of the formula),
	\item or a sum of product of correlators with strictly bigger Euler characteristic, with $(m+k)$ `$dx$-arguments', $k\geq 1$, and $n$ `$dy$-arguments' . Note also that the structure of Equation~\eqref{eq:Wxtoomega}  assumes that the first $m$ `$dx$-arguments' and all $n$ `$dy$-arguments' match the same arguments on the left hand side, and to the extra $k$ `$dx$-arguments' we first apply some differential operators and then specialize them to the diagonal, where they are matched to the the distinguished `$dy$-argument' specified on the left hand side of the formula. 
\end{itemize} 
From this description it is clear that iteratively applying Equation~\eqref{eq:Wxtoomega} to the left hand side of Equation~\eqref{eq:MainFormulaSimple}, we convert all `$dy$-arguments' into `$dx$-arguments' in a finite number of steps. Note also that once a`$dy$-argument' is replaces by a number of `$dx$-arguments' (restricted to the diagonal, with some differential operators applied), it remains intact at all the subsequent iterative applications of Equation~\eqref{eq:Wxtoomega}. The result is assembled into a closed formula on the right hand side of  Equation~\eqref{eq:omega-g-n-graphsXY}.

Now, since  Equations~\eqref{eq:mainrecchange} and~\eqref{eq:Wxtoomega} are equivalent (by Proposition~\ref{prop:simple-rec-mixed}), Theorem \ref{thm:newformula} follows.
 
\begin{remark} Note that the iterative application of Equation~\eqref{eq:mainrecchange}, or, alternatively, Equation~\eqref{eq:Wxtoomega}, gives also two alternative but equivalent expressions for all $\omega^{(g)}_{m,n}$, $g,m,n\geq 0$, $m+n\geq 1$, in terms of $\omega^{(g)}_{m,0}$'s. In particular, the second expression (the one that is based on the iterative application of Equation~\eqref{eq:Wxtoomega}) specializes for $g=0$ to the formula of Hock in~\cite[Theorem~4.9]{Hockx-x-y}.
\end{remark}


\section{Equivalence of differential-algebraic relations} \label{sec:proof-of-algebraic-equivalence-of-relations}

The main goal of this Section is to give a proof of Proposition~\ref{prop:simple-rec-mixed} claiming that the inductive definition of the mixed differentials $\omega^{(g)}_{m,n}$ given by~\eqref{eq:cTX}--\eqref{eq:omega-W} is equivalent to that one given by~\eqref{eq:cT}--\eqref{eq:Wxtoomega}. As a corollary, we obtain Theorem~\ref{thm:newformula} claiming that the two expressions~\eqref{eq:MainFormula} and~\eqref{eq:MainFormulaSimple} for the differentials $\omega^{(g)}_{0,n}$ of the $x-y$-dual topological recursion are equivalent.

It is convenient to represent both recursions in the following equivalent form. Define the \emph{$k$-connected, $(m,n)$-disconnected $(m+k,n)$-point functions} $W_{m,k,n}$ on $S^{m+k+n}$ as follows. Let $M,K,N$ be three nonintersecting sets of indices of cardinalities $|M|=m,|K|=k,|N|=n$. We set
\begin{equation}
W_{m,k,n}(z_M;z_K;z_N)=\sum_{\ell=1}^k\frac{1}{\ell!}
\!\!\!\sum_{\substack{\sqcup_{i=1}^\ell M_i=M,\\\sqcup_{i=1}^\ell N_i=N,\\
\sqcup_{i=1}^\ell K_i=K,\\
K_i\ne\emptyset,~(|M_i|,|K_i|,|N_i|)\ne (0,1,0)}}\!\!\!
\prod_{i=1}^\ell \tilde W_{|M_i|+|K_i|,|N_i|}(z_{M_i},z_{K_i};z_{N_i}),
\end{equation}
where for $(m,n)\ne(2,0)$ we set $\tilde W_{m,n}=W_{m,n}$. In the exceptional case $(m,n)=(2,0)$ we use the following convention: the meaning of $\tilde W_{2,0}$ depends on the range of indices of its arguments, namely, we set
\begin{align}
\tilde W_{2,0}(z_i,z_j)&=W_{2,0}(z_i,z_j),&\text{if~}&i\in M,~j\in K
\\\tilde W_{2,0}(z_i,z_j)&=W_{0,2}(z_i,z_j)-\frac{x_ix_j}{(x_i-x_j)^2},&\text{if~}&i\in K,~j\in K.
\end{align}
The condition $(|M_i|,|K_i|,|N_i|)\ne (0,1,0)$ means that we do not include the contribution of $W_{1,0}$ in $W_{m,k,n}$.
With this notation, the inductive definition~\eqref{eq:mainrecchange} suggests the following expression for $W_{m,n+1}$:
\begin{multline}\label{eq:equivproof2}
-\sum_{r,j=0}^\infty(-y\partial_y)^j\Bigl(
		[v^j]L_r(v,x\,y)
		\frac{dx}{x}\frac{y}{dy}
		[u^r]\frac{e^{u(\cS(u\,\hbar\,x\partial_x)-1)\,(x\,y)}}{u\,\cS(u\,\hbar)}\times\\
\sum_{k=1}^\infty\frac{\hbar^{2(k-1)}}{k!}\prod_{i=1}^k\bigl(\big\lfloor_{z_{\bar i}\to z}u\,\cS(u\hbar x_{\bar i}\partial_{x_{\bar i}})\bigr)
\;W_{m,k,n}(z_M;z_{\{\bar 1,\dots,\bar k\}};z_N)\Bigr),
\end{multline}
where
\begin{equation}\label{eq:equivproof1}
	L_r(v,\theta)\coloneqq\theta^{-v}\partial_\theta^re^{v\frac{\cS(v\,\hbar\,\partial_\theta)}{\cS(\hbar\,\partial_\theta)}\log\theta}.
\end{equation}
was earlier introduced in \eqref{eq:Lr}.
For simplicity, we restrict ourselves here to the case $m+n>0$; in the exceptional case $m=n=0$ the arguments are adjusted accordingly.

On the other hand, the recursion~\eqref{eq:Wxtoomega} reformulated similarly suggests the following alternative expression for the same function $W_{m,n+1}$:
\begin{multline}\label{eq:equivproof3}
-(-y)\sum_{r=0}^\infty \partial_y^r\Bigl(
		\frac{dx}{dy}
		[w^r]\frac{e^{w\,(\cS(w\hbar\partial_x)-1)(-y)}}{w}\times\\
\sum_{k=1}^\infty\frac{\hbar^{2(k-1)}}{k!}\prod_{i=1}^k\bigl(\big\lfloor_{z_{\bar i}\to z}w\,\cS(w\hbar \partial_{x_{\bar i}})\bigr)
\;\frac{W_{m,k,n}(z_M;z_{\{\bar 1,\dots,\bar k\}};z_N)}{\prod_{i=1}^k(-x_{\bar i})}\Bigr).
\end{multline}
We claim that the order of summation in~\eqref{eq:equivproof2} and~\eqref{eq:equivproof3} can be changed, so that for the resulting sums the terms with the same values of the index~$k$ coincide with each other. Actually, we have the following general relationship.

\begin{proposition}\label{prop:MainSimplification}
For any $k\ge1$ and any meromorphic function $H(z_1,\dots,z_k)$ on $S^k$ having no poles on the diagonals $z_i=z_j$ the following identity holds:
\begin{align}\label{eq:xXyY}
&
\sum_{r,j=0}^\infty(-y\partial_y)^j\Big(
		[v^j]L_r(v,x\,y)
		\frac{dx}{x}\frac{y}{dy}
\\ \notag & \quad
		[u^r]\frac{e^{u(\cS(u\,\hbar\,x\partial_x)-1)\,(x\,y)}}{u\,\cS(u\,\hbar)}
\prod_{i=1}^k\bigl(\big\lfloor_{z_{i}\to z}u\,\cS(u\hbar x_{i}\partial_{x_{i}})\bigr)
\;H(z_{\set{k}})\Big)
\\ \notag &
=
(-y)\sum_{r=0}^\infty \partial_y^r\left(
		\frac{dx}{dy}
		[w^r]\frac{e^{w\,(\cS(w\hbar\partial_x)-1)(-y)}}{w}
\prod_{i=1}^k\bigl( \big\lfloor_{z_{i}\to z}w\,\cS(w\hbar \partial_{x_{i}})\bigr)
\;\frac{H(z_{\set{k}})}{\prod_{i=1}^k(-x_{i})}\right).
\end{align} 
\end{proposition}

Equivalence of~\eqref{eq:equivproof2} and \eqref{eq:equivproof3} and thus Proposition~\ref{prop:simple-rec-mixed} is a direct corollary of this proposition.

\subsection{Proof of Proposition~\ref{prop:MainSimplification}}
\label{sec:ProofPropositionSimplificationformal}

We divide the computation leading to the proof of Proposition~\ref{prop:MainSimplification} into several steps. 
The structure of expressions on both sides of Equation~\eqref{eq:xXyY} implies that all sums involved contribute finitely many nonzero terms to the coefficient of any particular monomial in~$\hbar$ expansion. The proof is organized in such a way that this finiteness property is preserved in all intermediate steps of computations.

First of all, let us express the right hand side of~\eqref{eq:xXyY} in terms of the operators $x\partial_x$ and $y\partial_y$ instead of $\partial_x$ and $\partial_y$, respectively. Denote
\begin{align}\label{eq:uw}
u(w)&=\hbar^{-1}\log \left(\tfrac{1+w \,\hbar/2}{1-w\,\hbar/2}\right)=w+\frac{w^3 \hbar ^2}{12}+\frac{w^5 \hbar ^4}{80}+\dots,\\
\gamma(w)&=(1-\tfrac14w^2\hbar^2)^{1/2}.
\end{align}
For any monomial $x^d$, $d\in\Z$, we have
\begin{equation}
\begin{aligned}\label{eq:SpxSxpx}
w\,\cS(w\hbar\partial_x)\frac{x^d}{(-x)}
&=-\sum_{j=0}^\infty\frac{w^{2j+1}\hbar^{2j}x^{d-2j-1}}{2^{2j}(2j+1)!}\prod_{i=1}^{2j}(d-i)
\\&=- \frac{(1+\frac {\hbar\,w}{2\,x})^d-(1-\frac {\hbar\,w}{2\,x})^d}{\hbar\,d}x^d
\\&=(1-\tfrac{\hbar^2\,w^2}{4x^2})^{d/2}
\frac{\left(\frac{1-\frac {\hbar\,w}{2\,x}}{1+\frac {\hbar\,w}{2\,x}}\right)^{d/2}-
\left(\frac{1-\frac {\hbar\,w}{2\,x}}{1+\frac {\hbar\,w}{2\,x}}\right)^{-d/2}}{\hbar \,d}x^d
\\&=\gamma(\tfrac{w}{-x})^d\;u(\tfrac{w}{-x})\,\cS(u(\tfrac{w}{-x})\,\hbar\,d)\;x^d.
\end{aligned}
\end{equation}
This proves the following equality for operators acting on meromorphic functions on the spectral curve
\begin{equation}\label{eq:xXyY6}
w\,\cS(w\hbar\partial_x)\circ \frac{1}{-x}=\big\lfloor_{w\to\frac w{-x}}\gamma(w)^{x\partial_x}\;u(w)\,\cS(u(w)\,\hbar\,x\partial_x).
\end{equation}

Next, for any monomial $y^{d}$, $d\in\Z$, we compute similarly
\begin{equation}
y\,\partial^r_yy^{d+r-1}=\prod_{i=1}^r(d-i+r)\;y^{d}
\end{equation}
which proves the following equality for operators acting on meromorphic functions on the spectral curve
\begin{equation}\label{eq:xXyY6y}
(-y)\circ\partial^r_y\circ (-y)^{r-1}=\prod_{i=1}^r(-y\partial_y +i-r).
\end{equation}
Substituting~\eqref{eq:xXyY6} and~\eqref{eq:xXyY6y} to the right hand side of~\eqref{eq:xXyY} we bring it to the following form
\begin{equation}\label{eq:xXyY7}
\begin{aligned}
\sum_{r=0}^\infty \prod_{i=1}^r
&(-y\partial_y +i-r)\Bigl(\frac1{(-y)^{r-1}}\frac1{(-x)^{r+1}}
		\frac{dx}{dy}
		[w^r]\frac{e^{-w\,x\,y}}{w}\times
\\&e^{\gamma(w)^{x\partial_x}u(w)\,\cS(u(w)\hbar x\partial_x)(x\,y)}
\prod_{i=1}^k\bigl(\big\lfloor_{z_{i}\to z}\gamma(w)^{x_i\partial_{x_i}}u(w)\,\cS(u(w)\hbar x_i \partial_{x_{i}})\bigr)
\;H(z_{\set{k}})\Bigr)
\\=
\sum_{r,j\ge0}&(-y\partial_y)^j\Bigl([v^j]\frac{\partial_\theta^r\theta^v}{\theta^v}\Bigm|_{\theta=x\,y}
	\frac{dx}{x}\frac{y}{dy}
		[w^r]\frac{e^{-w\,x\,y}}{w}\times
\\&\gamma(w)^{x\partial_x}e^{u(w)\,\cS(u(w)\hbar x\partial_x)(x\,y)}
\prod_{i=1}^k\bigl(\big\lfloor_{z_{i}\to z}u(w)\,\cS(u(w)\hbar x_i \partial_{x_{i}})\bigr)
\;H(z_{\set{k}})\Bigr).
\end{aligned}
\end{equation}

Comparing the obtained expression with the left hand side of~\eqref{eq:xXyY} we reduce thereby the proof of Proposition~\ref{prop:MainSimplification}, in turn, to the following identity holding in  the space of meromorphic functions in just one variable.

\begin{lemma} \label{lem:xXyY8}
Let $G(u,z)$ be a polynomial in~$u$ whose coefficients are meromorphic functions on~$S$ (or a power series in~$\hbar$ whose coefficients are functions on~$S$ depending on $u$ in a polynomial way). Then, the following identity holds:
\begin{align}\label{eq:xXyY8}
&
\sum_{r,j\ge 0}(-y\partial_y)^j\Bigl(
		[v^ju^r]L_r(v,x\,y)
		\frac{dx}{x}\frac{y}{dy}G(u,z)\Bigr) 
\\ \notag &
=\sum_{r,j\ge0}(-y\partial_y)^j\Bigl([v^jw^r]\frac{\partial_\theta^r\theta^v}{\theta^v}\Bigm|_{\theta=xy}
	\frac{dx}{x}\frac{y}{dy} \tfrac{u(w)\cS(u(w)\,\hbar)}{w} \times
\\ \notag & \qquad \qquad \qquad \quad 
		e^{-w\,x\,y}\gamma(w)^{x\partial_x}e^{u(w)\,x\,y}G(u(w),z)\Bigr).
\end{align} 
\end{lemma} 

Proposition~\ref{prop:MainSimplification} follows from this lemma, since the equality of the left hand side of~\eqref{eq:xXyY} and the obtained expression~\eqref{eq:xXyY7} for the right hand side is just the equality of the last Lemma for the case
\begin{equation}
G(u,z)=\frac{e^{u(\cS(u\,\hbar\,x\partial_x)-1)\,(x\,y)}}{u\,\cS(u\,\hbar)}
\prod_{i=1}^k\bigl(\big\lfloor_{z_{i}\to z}u\,\cS(u\hbar x_{i}\partial_{x_{i}})\bigr)
\;H(z_{\set{k}}).
\end{equation}

The next step of our computations is to simplify the right hand side of~\eqref{eq:xXyY8} using the following identity

\begin{lemma}\label{lem:xXyY9}
	For any function $F(w,v,z)$ on the spectral curve depending on parameters~$w$ and~$v$ such that the dependence of $e^{w\,x\,y}F(w,v,z)$ in~$u$ and~$v$ is polynomial we have
	\begin{multline}\label{eq:xXyY9}
		\sum_{r,j\ge0}(-y\partial_y)^j\Bigl([w^rv^j]
		\frac{\partial_\theta^r\theta^v}{\theta^v}\Bigm|_{\theta=x\,y}
		\frac{dx}{x}\frac{y}{dy}e^{-w\,x\,y}\;x\partial_x\; F(w,v,z)\Bigr)\\
		= \sum_{r,j\ge0}(-y\partial_y)^j\Bigl([w^rv^j]
		\frac{\partial_\theta^r\theta^v}{\theta^v}\Bigm|_{\theta=x\,y}
		\frac{dx}{x}\frac{y}{dy}e^{-w\,x\,y}\;v\; F(w,v,z)\Bigr).
	\end{multline}
\end{lemma}
\begin{proof}
First, assume that there is a point $p$ on the spectral curve such that $x$ is a local coordinate at this point, and~$y$ has a simple pole, $y=\frac{a}{x}+O(x^0)$, $a\ne0$. Then, expanding $F$ as a Laurent series, we compute the left hand side of~\eqref{eq:xXyY9} using Lagrange inversion and Taylor expansion formulas (cf.~\cite[Sect. 4.3]{BDKS-OrlovScherbin})
\begin{equation}
\begin{aligned}
		\sum_{r,j\ge0}&(-y\partial_y)^j\Bigl([w^rv^j]
		\frac{\partial_\theta^r\theta^v}{\theta^v}\Bigm|_{\theta=x\,y}
		\frac{dx}{x}\frac{y}{dy}e^{-w\,x\,y}\;x\partial_x\; F(w,v,z)\Bigr)
\\=&-		\sum_{\ell\in\Z}\sum_{r\ge0}y^{-\ell}\Bigl([w^rx^\ell]
		\partial_\theta^r\theta^\ell\Bigm|_{\theta=x\,y}
		e^{-w\,x\,y}\;x\partial_x\; F(w,\ell,z)\Bigr)
\\=&-		\sum_{\ell\in\Z}\sum_{r\ge0}y^{-\ell}\Bigl([w^rx^\ell]
		\partial_\theta^r\theta^\ell\Bigm|_{\theta=a}e^{-a\,w}
		\;x\partial_x\; F(w,\ell,z)\Bigr)
\\=&-		\sum_{\ell\in\Z}\sum_{r\ge0}y^{-\ell}\Bigl([w^rx^\ell]
		\partial_\theta^r\theta^\ell\Bigm|_{\theta=a}e^{-a\,w}
		\;\ell\;F(w,\ell,z)\Bigr).
\end{aligned}
\end{equation}
The last equality holds because $\partial_\theta^r\theta^\ell\Bigm|_{\theta=a}e^{-a\,w}$ is a constant. A similar computation of the right hand side of~\eqref{eq:xXyY9} leads to the same answer.  This proves the Lemma for the case of formal Laurent series. We would like to reproduce this computation in a form that does not require Laurent expansions and is applicable to arbitrary meromorphic functions.

The following equality is obtained by a straightforward expanding of all brackets
\begin{equation}\label{eq:xXyY10}
\frac{y}{dy}\frac{dx}{x}\frac{e^{-w\,x\,y}}{(x\,y)^v}(x\partial_x-v)F(w,v,z)
=(v+y\partial_y+ (y\partial_y(x\,y))\,w)\frac{e^{-w\,x\,y}}{(x\,y)^v}F(w,v,z).
\end{equation}
Applying the operator $\sum_{r,j\ge0}(-y\partial_y)^j\circ\bigl([w^rv^j]\partial_\theta^r\theta^v\big|_{\theta=x\,y}\bigr)$ to the left hand side of~\eqref{eq:xXyY10} we obtain the difference of the two sides of~\eqref{eq:xXyY9}. On the other hand, application of the same operator to the right hand side of~\eqref{eq:xXyY10} gives zero. Indeed, the insertion of the factor $v$ is equivalent to the shift of the index~$j$ by one, that is, to the application of $-y\partial_y$. While computing this derivative we should take into account the contribution of $-y\partial_y(\partial_\theta^r\theta^v|_{\theta=x\,y})=(-y\partial_y(x\,y))(\partial_\theta^{r+1}\theta^v|_{\theta=x\,y})$. This shift of the index~$r$ by one is accounted by insertion of the factor~$w$. This proves Lemma~\ref{lem:xXyY9} in full generality.
\end{proof}

Let us return to the proof of Lemma~\ref{lem:xXyY8}. The operator $\gamma(w)^{x\partial_x}$ entering the right hand side of~\eqref{eq:xXyY8} is a series in $\hbar^2$ whose coefficients are polynomials in $x\partial_x$. Applying repeatedly the equality of Lemma~\ref{lem:xXyY9} we conclude that the right hand side of~\eqref{eq:xXyY8} is equal to
\begin{equation}
\sum_{r,j\ge0}(-y\partial_y)^j\Bigl([v^jw^r]\frac{\partial_\theta^r\theta^v}{\theta^v}\Bigm|_{\theta=x\,y}
	\frac{dx}{x}\frac{y}{dy}
		\tfrac{u(w)\cS(u(w)\,\hbar)}{w}e^{-w\,x\,y}\gamma(w)^{v}e^{u(w)\,x\,y}G(u(w),z)\Bigr).
\end{equation}
Comparing this with the left hand side of~\eqref{eq:xXyY8} and using that
\begin{equation}
\tfrac{u(w)\,\cS(u(w)\,\hbar)}{w}=(1-w^2\hbar^2/4)^{-1/2}=\gamma(w)^{-1}.
\end{equation}
we reduce the proof of Lemma~\ref{lem:xXyY8} to the following identity
\begin{equation}
L_r(v,x\,y)=
\sum_{k\ge0}\frac{\partial_\theta^k\theta^v}{\theta^v}\bigm|_{\theta=x\,y}[w^k]
		e^{(u(w)-w)\,x\,y}\gamma(w)^{v-1}u(w)^r
\end{equation}
that holds for all $r\ge0$ or, more explicitly, to the identity
\begin{equation}\label{eq:xXyY4}
\partial_\theta^r e^{v\frac{\cS(v\,\hbar\,\partial_\theta)}{\cS(\hbar\,\partial_\theta)}\log\theta}=
\sum_{k\ge0}\partial_\theta^k\theta^v[w^k]
		e^{(u(w)-w)\,\theta}\gamma(w)^{v-1}u(w)^{r}
\end{equation}
which is understood as an equality of functions in $v$ and $\theta$. The application of $\partial_\theta$ to the left hand side results in the multiplication of the expression to the right of $[w^k]$ by $w+(u(w)-w)=u(w)$, that is, to the increment of the index $r$ by one. Therefore, Equation~\eqref{eq:xXyY4} for any $r\ge 0$ follows from its special case for $r=0$. Let us write~\eqref{eq:xXyY4} for $r=0$ more explicitly. Denote
\begin{equation}\label{eq:etadef}
\eta(w)=e^{u(w)-w}=e^{\hbar^{-1}\log\left(\frac{1+w\hbar/2}{1-w\hbar/2}\right)-w}.
\end{equation}

\begin{lemma} \label{lem:phiidentity} The following identity holds true
	\begin{equation} \label{eq:phiidentity1}
		e^{v\frac{\cS(v\,\hbar\,\partial_\theta)}{\cS(\hbar\,\partial_\theta)}\log\theta}
		=\sum_{k=0}^\infty\bigl(\partial_\theta^k\theta^v\bigr)
		[w^k] (1-w^2\hbar^2/4)^{\frac{v-1}2}\eta(w)^\theta.
	\end{equation}
\end{lemma}

\begin{proof}
It would be more convenient to apply $e^{-\hbar\frac{v-1}{2}\partial_\theta}$ to both sides of~\eqref{eq:phiidentity1} and to represent the identity to be proved in the following equivalent form
\begin{equation}\label{eq:phishifted}
\begin{aligned}
e^{-\hbar\frac{v-1}{2}\partial_\theta}e^{v\frac{\cS(v \hbar\partial_\theta)}{\cS(\hbar\partial_\theta)}\log(\theta)}
&=\sum_{r\ge0}\bigl(\partial_\theta^r\theta^v\bigr)[w^r]e^{-\hbar\frac{v-1}{2}w}(1-w^2\hbar^2/4)^{\frac{v-1}{2}}
\eta(w)^{\theta-\hbar\frac{v-1}{2}}
\\&=\sum_{r\ge0}\bigl(\partial_\theta^r\theta^v\bigr)[w^r](1-\hbar w/2)^{v-1}\bigl(\tfrac{1+\hbar w/2}{1-\hbar w/2}\bigr)^{\theta/\hbar}e^{-w\theta}.
\end{aligned}
\end{equation}

Both sides of~\eqref{eq:phishifted} divided by $\theta^v$ are infinite power series in $\hbar^2$ whose coefficients are polynomials in $v$ and $\theta^{\pm1}$. Every polynomial of degree $d$ is uniquely determined by its values at any non-coincident $d+1$ points. It follows that it is sufficient to prove~\eqref{eq:phishifted} in the case when~$v$ is a nonnegative integer. So we assume that $v\in\Z_{\ge0}$. In that case the left hand side specializes to
\begin{equation}\label{eq:phishiftedlhs}
\exp\Bigl(\tfrac{1-e^{-v\hbar\partial\theta}}
{1-e^{-\hbar\partial\theta}}
\log(\theta)\Bigr)=
\theta(\theta-\hbar)(\theta-2\hbar)\dots(\theta-(v-1)\hbar).
\end{equation}
Now we compute the right hand side. Applying the Taylor expansion formula we transform the right hand side to
\begin{equation}\label{eq:phishiftedrhs}
\sum_{r\ge0}\bigl(\partial_\theta^r\theta^v\bigm|_{\theta=0}\bigr)
[w^r](1-\hbar w/2)^{v-1}\bigl(\tfrac{1+\hbar w/2}{1-\hbar w/2}\bigr)^{\theta/\hbar}
=v![w^v](1-\hbar w/2)^{v-1}\bigl(\tfrac{1+\hbar w/2}{1-\hbar w/2}\bigr)^{\theta/\hbar}.
\end{equation}
Remark that all these manipulations make sense if $v$ is a nonnegative integer only.

The obtained expression can be further simplified by means of Lagrange inversion formula. Introduce a change of the $w$ parameter defined by
\begin{equation}
s(w)=\frac{w}{1-w\hbar/2}.
\end{equation}
Then, for this change we have
\begin{equation}
\begin{gathered}
\frac1{1-w\hbar/2}=\frac{w}{dw}\frac{ds(w)}{s(w)}=\frac{s(w)}{w},
\\\frac{1+w\hbar/2}{1-w\hbar/2}=1+s(w)\hbar.
\end{gathered}
\end{equation}
Applying Lagrange inversion formula, we represent~\eqref{eq:phishiftedrhs} as
\begin{equation}
v![w^v]\Bigl(\frac{w}{s(w)}\Bigr)^v\frac{w}{dw}\frac{ds(w)}{s(w)}(1+s(w)\,\hbar)^{\theta/\hbar}
=v![s^v](1+s\,\hbar)^{\theta/\hbar}=v!\hbar^v\binom{\theta/\hbar}{v}
\end{equation}
which agrees with~\eqref{eq:phishiftedlhs}. This completes the proof of Lemma~\ref{lem:phiidentity}, Lemma~\ref{lem:xXyY8}, Proposition~\ref{prop:MainSimplification}, Proposition~\ref{prop:simple-rec-mixed}, and Theorem~\ref{thm:newformula}.
\end{proof}


\section{Algebraic properties of mixed correlation differentials}
\label{sec:proofs}

In this section we consider the set of mixed correlation differentials $\omega^{(g)}_{m,n}$ associated by Definition~\ref{def:mixedcorrdiff} to an arbitrary system of globally defined symmetric meromorphic $m$-dif\-ferentials $\omega^{(g)}_{m,0}$ satisfying Assumption~\eqref{eq:assumption-diagonal}. The goal of this section is to analyze some formal consequences of this definition (which are then used in the proof of the main theorem in Section~\ref{sec:TopologicalRecursion}). It is important to stress again that we have so far no further assumptions on the initial $m$-differentials $\omega^{(g)}_{m,0}$, in particular, we don't relate them to the topological recursion.

\subsection{Closed  formulas and diagonals}
\label{sec:diagonalsproofs}

 We provide two closed differential-algebraic formulas for all $\omega^{(g)}_{m,n}$.

\begin{proposition}\label{prop:omega-g-n-graphsXY} We have:
	\begin{align} \label{eq:omega-g-n-graphsXY}
		&\omega^{(g)}_{m,n}(z_{\set{m+n}}) \prod_{j=1}^m\Big(-\frac{x_{j}}{dx_{j}}\Big) \prod_{i=m+1}^{m+n}\Big(-\frac{y_i}{dy_i}\Big)
		=
		\\ \notag &
		[\hbar^{2g}] \sum_{\Gamma} \frac{
			\hbar^{2g(\Gamma)}}{|\mathrm{Aut}(\Gamma)|} \prod_{i=m+1}^{m+n}
		\sum_{k_i=0}^\infty ( -y_i \partial_{y_i})^{k_i} [v_i^{k_i}]
		\\ \notag &
		\prod_{i=m+1}^{m+n}
		 \Big(-\frac{y_i}{dy_i}\frac{dx_i}{x_i} \Big)
		\sum_{r_i=0}^\infty \Big(\partial_{\theta_i} + \frac{v_i}{\theta_i}\Big)^{r_i}
		e^{v_i\frac{\cS(\hbar v_i \partial_{\theta_i})}{\cS(\hbar \partial_{\theta_i})} \log\theta_i - v_i\log\theta_i}\Big|_{\theta_i = - \frac{x_i}{dx_i}\omega^{(0)}_{1,0} (z_i)}[u_i^{r_i}]
		\\ \notag &
		\prod_{i=m+1}^{m+n}
		\frac{1}{ u_i \cS(\hbar u_i)} e^{u_i \cS(\hbar u_i x_i \partial_{x_i}) \sum_{\tilde g=0}^\infty \hbar^{2\tilde g} \big(-\frac{x_i}{dx_i}\omega^{(\tilde g)}_{1,0} (z_i)\big)-u_i \big(-\frac{x_i}{dx_i}\omega^{(0)}_{1,0} (z_i)\big)}
		\\ \notag &
		\prod_{e\in E(\Gamma)} \prod_{\substack{j\in\llbracket |e| \rrbracket\\ e(j)\notin \llbracket m \rrbracket}}  \tilde u_j \cS(\hbar \tilde u_j \tilde x_j \partial_{\tilde x_j}) \sum_{\tilde g=0}^\infty \hbar^{2\tilde g}\tilde \omega^{(\tilde g)}_{|e|,0}(\tilde z_{\llbracket |e|\rrbracket})  \prod_{j=1}^{|e|} \Big(-\frac{\tilde x_j}{d\tilde x_j}\Big)\Big|_{(\tilde u_j, \tilde x_j) \to (u_{e(j)},x_{e(j)})}
		\\ \notag &
		+\delta_{(m,n),(0,1)}[\hbar^{2g}] \sum_{k=0}^\infty (-y_1\partial_{y_1})^k \Big( [v^{k+1}] e^{v\frac{\cS(\hbar v \partial_{\theta})}{\cS(\hbar \partial_{\theta})} \log\theta - v\log\theta}\Big|_{\theta= - \frac{x_1}{dx_1}\omega^{(0)}_{1,0} (z_1)} 
		\\ \notag & \qquad \qquad \qquad \qquad \qquad \qquad \qquad 
		(-y_1\partial_{y_1})\big(-\frac{x_1}{dx_1}\omega^{(0)}_{1,0} (z_1)\big) \Big)
		\\ \notag &
		+ \delta_{(g,m,n),(0,0,1)} \Big(-\frac{x_1}{dx_1}\omega^{(0)}_{1,0} (z_1)\Big),
	\end{align}  
	Here
	\begin{itemize}
		
		\item The sum is taken over all connected graphs $\Gamma$ with $n$ labeled vertices (with labels from $m+1$ to $m+n$), $m$ labeled leaves (i.e. special vertices of valence strictly equal to one; their labels are from $1$ to $m$) and multiedges of index $\geq 2$; cf. Section~\ref{sec:ExplFormula}. 
		\item For convenience, for a given such graph, we also label all legs of every given multiedge $e$ from $1$ to $|e|$ in an arbitrary way.
		\item For a multiedge $e$ with index $|e|$  we control its attachment to the vertices by the associated map $e\colon \llbracket |e| \rrbracket \to 
		\llbracket n \rrbracket
		$ that we denote also by $e$, abusing notation (so $e(j)$ is the label of the vertex to which the $j$-th leg of the multiedge $e$ is attached). Unlike the case of Equation~\eqref{eq:MainFormula}, these maps are no longer arbitrary here, due to the condition on the valence of leaves; but a given regular vertex can still be attached to multiple legs of a given multiedge. 
		\item For a given multiedge $e$ with $|e|=2$ we define $\tilde \omega^{(0)}_{2,0}(\tilde x_1,\tilde x_2) :=  \omega^{(0)}_{2,0}(\tilde x_1,\tilde x_2) - \frac{d\tilde x_1d\tilde x_2}{(\tilde x_1-\tilde x_2)^2}$ if $e(1)=e(2)$, and $\tilde \omega^{(0)}_{2,0}(\tilde x_1,\tilde x_2) :=  \omega^{(0)}_{2,0}(\tilde x_1,\tilde x_2)$ otherwise. For all $(g,n)\not=(0,2)$ we simply have $\tilde \omega^{(g)}_{n,0} :=  \omega^{(g)}_{n,0}$.
		\item By $g(\Gamma)$ we denote the first Betti number of $\Gamma$.
		\item $|\mathrm{Aut}(\Gamma)|$ stands for the number of automorphisms of $\Gamma$.
	\end{itemize}
\end{proposition}

\begin{proposition}\label{prop:omega-g-n-graphsxy} We have:
	\begin{align} \label{eq:omega-g-n-graphsxy}
		&\omega^{(g)}_{m,n}(z_{\set{m+n}}) \prod_{j=1}^m \frac{1}{dx_{j}} \prod_{i=m+1}^{m+n} \frac{1}{dy_i} 
		=
		(-1)^n
		\coeff \hbar {2g} \sum_{\Gamma} \frac{\hbar^{2g(\Gamma)}}{|\mathrm{Aut}(\Gamma)|} \prod_{i=m+1}^{m+n}
		\sum_{k_i=0}^\infty  \partial_{y_i}^{k_i} [w_i^{k_i}] \frac{dx_i}{dy_i}
		\\ \notag &
		\frac{1}{w_i} e^{ w_i \cS(\hbar w_i \partial_{x_i}) \sum_{\tilde g=0}^\infty \hbar^{2\tilde g} \frac{1}{dx_i}\omega^{(\tilde g)}_{1,0} (z_i)-w_i \frac{1}{dx_i}\omega^{(0)}_{1,0} (z_i)}
		\\ \notag &
		\prod_{e\in E(\Gamma)} \prod_{\substack{j\in\llbracket |e| \rrbracket\\ e(j)\notin \llbracket m \rrbracket}} \tilde w_j \cS(\hbar \tilde w_j  \partial_{\tilde x_j}) \sum_{\tilde g=0}^\infty \hbar^{2\tilde g}\tilde \omega^{(\tilde g)}_{|e|,0}(\tilde z_{\llbracket |e|\rrbracket})  \prod_{j=1}^{|e|} \frac{1}{d\tilde x_j}\Big|_{(\tilde w_j, \tilde x_j) \to (w_{e(j)},x_{e(j)})}
		\\ \notag &
		+\delta_{(g,m,n),(0,0,1)} (-x_1).		
	\end{align}
\end{proposition}

\begin{proof}[Proof of Propositions~\ref{prop:omega-g-n-graphsXY} and~\ref{prop:omega-g-n-graphsxy}]
	They both follow directly via iterative application of Equations~\eqref{eq:mainrecchange} and~\eqref{eq:Wxtoomega}. Let us explain the way it works for Proposition~\ref{prop:omega-g-n-graphsXY}. The explanation is essentially the same we give in the proof of Theorem~\ref{thm:newformula}, but we repeat it here for completeness.  
	
	Abusing the language a little bit, we can say that Equation~\eqref{eq:omega-g-n-graphsXY} expresses a mixed correlator with $m$ `$dx$-arguments' and $n$ `$dy$-arguments' in terms of the correlators with only `$dx$-arguments'. On the other hand, Equation~\eqref{eq:mainrecchange} expresses the correlators with $m$ `$dx$-arguments' and $(n+1)$ `$dy$-arguments' in terms of the correlators with 
	\begin{itemize}
		\item either the only term with same Euler characteristic $2-2g-m-n-1$ but then with $(m+1)$ `$dx$-arguments' and $n$ `$dy$-arguments' (thus one special `$dx$-argument' on the right hand side is matched to the distinguished `$dy$-argument' specified on the left hand side of the formula),
		\item or a sum of product of correlators with strictly bigger Euler characteristic, with $(m+k)$ `$dx$-arguments', $k\geq 1$, and $n$ `$dy$-arguments' . Note also that the structure of Equation~\eqref{eq:mainrecchange} assumes that the first $m$ `$dx$-arguments' and all $n$ `$dy$-arguments' match the same arguments on the left hand side, and to the extra $k$ `$dx$-arguments' we first apply some differential operators and then specialize them to the diagonal, where they are matched to the the distinguished `$dy$-argument' specified on the left hand side of the formula. 
	\end{itemize} 
	From this description it is clear that iteratively applying Equation~\eqref{eq:mainrecchange} to the left hand side of Equation~\eqref{eq:omega-g-n-graphsXY}, we convert all `$dy$-arguments' into `$dx$-arguments' in a finite number of steps. Note also that once a`$dy$-argument' is replaces by a number of `$dx$-arguments' (restricted to the diagonal, with some differential operators applied), it remains intact at all the subsequent iterative applications of Equation~\eqref{eq:mainrecchange}. The result is assembled into a closed formula on the right hand side of  Equation~\eqref{eq:omega-g-n-graphsXY}.
\end{proof}

\begin{remark} \label{rem:equiv-mixed} Note that this also implies that Equations~\eqref{eq:omega-g-n-graphsXY} and~\eqref{eq:omega-g-n-graphsxy} are equivalent by Proposition~\ref{prop:simple-rec-mixed}.
\end{remark}

\begin{remark}
	Below, in Sections~\ref{sec:vertex-operators-standard-E} and~\ref{sec:Computations-vertex-new-form}, we clarify the source of Equations~\eqref{eq:omega-g-n-graphsXY} and~\eqref{eq:omega-g-n-graphsxy}, respectively, in terms of some computations with vacuum expectation values, which are quite similar to the presentation in Section~\ref{sec:MotivationFPS} of the motivation of Definition~\ref{def:mixedcorrdiff}.
\end{remark}

\begin{remark} \label{rem:Symmetry}
	Note that one of the direct corollaries of Propositions~\ref{prop:omega-g-n-graphsXY} and~\ref{prop:omega-g-n-graphsxy} is that $\omega^{(g)}_{m,n}$ is symmetric under the action of $\mathfrak{S}_m\times \mathfrak{S}_n$ permuting the first $m$ and the last $n$ variables, respectively.
\end{remark}

Let us consider one more property of $\omega^{(g)}_{m,n}$ that in full generality is not obvious neither from Definition~\ref{def:mixedcorrdiff} nor from its simplified version given by Equation~\eqref{eq:Wxtoomega}, nor from the explicit formulas~\eqref{eq:omega-g-n-graphsXY},~\eqref{eq:omega-g-n-graphsxy}.

\begin{proposition} \label{prop:regular-at-diagonals} The differentials $\omega^{(g)}_{m,n}(z_{\set{m+n}})$ are regular at the diagonals $z_i=z_j$ for $i,j\in \set{m}$ and for $i,j\in \set{m+n}\setminus \set{m}$ for $(g,m,n)\not=(0,2,0)$ and $(g,m,n)\not=(0,0,2)$. In the latter two cases
	\begin{align}
		\omega^{(0)}_{2,0}(z_{\set{2}})-\frac{dx(z_1)dx(z_2)}{(x(z_1)-x(z_2))^2} \qquad \text{and} \qquad \omega^{(0)}_{0,2}(z_{\set{2}})-\frac{dy(z_1)dy(z_2)}{(y(z_1)-y(z_2))^2}
	\end{align}
	are regular at the diagonal.
\end{proposition}

See Section \ref{S.Proofs} for the proof.

\begin{remark} The regularity at the diagonals $z_i=z_j$ for $i,j\in\set{m}$ is a part of the assumptions for $\omega^{(g)}_{m,0}$ in the case $n=0$, and for general $n$ it follows directly from Equations~\eqref{eq:omega-g-n-graphsXY}, \eqref{eq:omega-g-n-graphsxy}. Then the non-trivial part of this Proposition is that there are no poles at the diagonals $z_i=z_j$  for $i,j\in \set{m+n}\setminus \set{m}$.
\end{remark} 


\subsection{Differential-algebraic identities for mixed correlation differentials}
Here we formulate a number of formal consequences of the defining relations for mixed correlation differentials. We use them subsequently for the analysis of the mixed correlation differentials.

We use the same conventions / notation as in Section~\ref{sec:MixedDefinition}: $M=\set{m}$ and $N=\set{m+n+1}\setminus \set{m+1}$; $z_{\set{m+n+1}}$ are variables on $S^{m+n+1}$, $x_i\coloneqq x(z_i)$, $y_i\coloneqq y(z_i)$, and for $i=m+1$ we denote $z=z_{m+1}$, $x=x_{m+1}$ and $y=y_{m+1}$. Recall~\eqref{eq:cT}--\eqref{eq:cW} for comparison. We use the symmetry property explained in Remark~\ref{rem:Symmetry} and regularity at the diagonals stated in Proposition~\ref{prop:regular-at-diagonals} to define
\begin{align}\label{eq:cTy}
	\cT^y_{m,n+1}(w;z_{M};z;z_{N}) & \coloneqq
	\sum_{k=1}^\infty\frac{\hbar^{2(k-1)}w^k}{k!}\left(\prod_{i=1}^k
	\big\lfloor_{z_{\bar i}\to z}\cS(w\hbar \partial_{y_{\bar i}})\right)\\ \notag & \quad
	\sum_{g=0}^\infty\hbar^{2g}\frac{
		\omega^{(g)}_{m,k+n}(z_{M};z_{\{\bar1,\dots,\bar k\}},z_{N})-
		\delta_{(g,m,k,n),(0,0,2,0)}\frac{dy_{\bar1}dy_{\bar2}}{(y_{\bar1}-y_{\bar2})^2}}
	{\prod_{i\in M} dx_i \prod_{i=1}^k dy_{\bar i}\prod_{i\in N} dy_i},
	\\ \label{eq:cWy}
	\cW^y_{m,n+1}(w;z)&=\sum_{g=0}^\infty\hbar^{2g}\cW^{y,(g)}_{m,n+1}(w,z)\coloneqq\\ \nonumber
	 &  \frac{1}{w}~e^{\cT^y_{0,1}(w;z)}\!\!\!\!\!\!
	\sum_{\substack{M\cup N=\sqcup_\alpha K_\alpha~
			K_\alpha\ne\varnothing\\~I_\alpha=K_\alpha\cap M,~J_\alpha=K_\alpha\cap N}}
	\prod_{\alpha}\cT^y_{|I_\alpha|,|J_\alpha|+1}(w;z_{I_\alpha};z;z_{J_\alpha}).
\end{align}

\begin{proposition} \label{prop:xy-relations} For $(g,m,n)\neq(0,0,0)$ we have:
	\begin{equation}\label{eq:Wytoomega}
		\frac{\omega^{(g)}_{m+1,n}}{\prod_{i \in M} dx_i \cdot dx \cdot \prod_{i\in N} dy_i}=-\sum_{r\ge0} \partial_x^r
		[w^r] \frac{dy}{dx} e^{w\,x}\cW^{y,(g)}_{m,n+1}
		.
	\end{equation}
	Moreover, relations~\eqref{eq:Wxtoomega} and~\eqref{eq:Wytoomega} possess the following ``parametric'' generalization:
	\begin{align}\label{eq:WxtoW}
		\cW^{y,(g)}_{m,n+1}(\tilde w;z)&=
		-\sum_{r\ge0} \partial_y^r\Bigl(e^{-\tilde w\,x}[w^r]\frac{dx}{dy} e^{w\,y}\cW^{x,(g)}_{m+1,n}(w;z)\Bigr);\\
		\label{eq:WytoW}
		\cW^{x,(g)}_{m+1,n}(w;z)&=
		-\sum_{r\ge0} \partial_x^r\Bigl(e^{-w\,y}[{\tilde w}^r]\frac{dy}{dx} e^{\tilde w\,x}\cW^{y,(g)}_{m,n+1}(\tilde w;z)\Bigr).
	\end{align}
\end{proposition}
See Section \ref{S.Proofs} for the proof.
A few remarks are in order.

\begin{remark}
	Recall that for $(g,m,n)=(0,0,0)$ we, \emph{by definition}, have $\omega^{(0)}_{1,0}=-ydx$ (as well as $\omega^{(0)}_{0,1}=-xdy$).
\end{remark}

\begin{remark}
	Equation~\eqref{eq:Wytoomega} expresses $\omega^{(g)}_{m+1,n}$ in terms of $\omega^{(g')}_{m',n'}$ with $m'<m+1$. This expression is $x-y$ swap dual to the one given by Equation~\eqref{eq:Wxtoomega}, which is a manifestation of the symmetry between $x$ and $y$ that we are chasing in this paper.
\end{remark}

\begin{remark} Equations~\eqref{eq:WxtoW} and~\eqref{eq:WytoW} are also $x-y$ swap dual to each other. They will be used later in this paper to produce the so-called loop equations (see Section~\ref{sec:RecollectionTR}).
\end{remark}

\begin{remark} Equations~\eqref{eq:WxtoW} and~\eqref{eq:WytoW} are understood as equalities of the coefficients of each particular power of the parameters $\tilde w$ and $w$, respectively. In particular,~\eqref{eq:Wxtoomega} and~\eqref{eq:Wytoomega} can be recovered as the coefficients of $\tilde w^0$ and $w^0$ in~\eqref{eq:WxtoW} and~\eqref{eq:WytoW}, respectively.
	
	Note that~\eqref{eq:WxtoW} and~\eqref{eq:WytoW} can be represented in the following more symmetric form:
	\begin{align}
		e^{\tilde w\,x}\cW^{y,(g)}_{m,n+1}(\tilde w;z)&
		=-\sum_{r\ge0} \bigl(\partial_y-\tilde w\tfrac{dx}{dy}\bigr)^r\Bigl([w^r]\frac{dx}{dy}e^{w\,y}\cW^{x,(g)}_{m+1,n}(w;z)\Bigr);\\
		e^{w\,y}\cW^{x,(g)}_{m+1,n}(w;z)&
		=-\sum_{r\ge0} \bigl(\partial_x-w\tfrac{dy}{dx}\bigr)^r\Bigl([{\tilde w}^r]\frac{dy}{dx}e^{\tilde w\,x}\cW^{y,(g)}_{m,n+1}(\tilde w;z)\Bigr).
	\end{align}
	The advantage of this form of equations is that both sides of the equality are now polynomial in the parameters $w$ and $\tilde w$, respectively. However, for applications to loop equations discussed below (see Section~\ref{sec:TopologicalRecursion}) the original  form~\eqref{eq:WxtoW} and~\eqref{eq:WytoW} is preferable.
\end{remark}

Recall the notation $dX = -dx/x$, $dY=-dy/y$, $\partial_X=-x\partial_x$, $\partial_Y=-y\partial_y$. Recall~\eqref{eq:cTXom}--\eqref{eq:cWexpXom} 
and let
\begin{align}\label{eq:cTY}
	\cT^Y_{m,n+1}(u;z_{M};z;z_{N})& \coloneqq
	\sum_{k=1}^\infty\frac{\hbar^{2(k-1)}u^k}{k!}\left(\prod_{i=1}^k
	\big\lfloor_{z_{\bar i}\to z}\cS(u\hbar \partial_{Y_{\bar i}})\right)\\ \notag & \quad
	\sum_{g=0}^\infty\hbar^{2g}\frac{
		\omega^{(g)}_{m,k+n}(z_{M};z_{\{\bar1,\dots,\bar k\}},z_{N})-
		\delta_{(g,m,k,n),(0,0,2,0)}\frac{dy_{\bar1}dy_{\bar2}}{(y_{\bar1}-y_{\bar2})^2}}
	{\prod_{i=1}^m dX_i \prod_{i=1}^k dY_{\bar i}\prod_{i=m+2}^{m+n+1} dY_i},
	\\
	\label{eq:cWexpY}
	\cW^Y_{m,n+1}(u;z)& \coloneqq \frac{~e^{\cT^Y_{0,1}(u;z)}}{u\,\cS(u\,\hbar)}
	\sum_{\substack{M\cup N=\sqcup_\alpha K_\alpha~
			K_\alpha\ne\varnothing\\~I_\alpha=K_\alpha\cap M,~J_\alpha=K_\alpha\cap N}}
	\prod_{\alpha}\cT^Y_{|I_\alpha|,|J_\alpha|+1}(u;z_{I_\alpha};z;z_{J_\alpha}).
\end{align}
Recall also~\eqref{eq:Theta}--\eqref{eq:Lr}.
Those definitions imply that the coefficient of $\hbar^{2g}$ in $L_r$ is a polynomial in~$v$ divided by $\theta^{2g+r}$. We also need a parametric generalization of the latter function:
\begin{equation}
	\tilde L_r(v,\theta,u)=(\partial_\theta+\tfrac{v}{\theta})^re^{u\,\theta}L_0(v,\theta).
\end{equation}

\begin{proposition} \label{prop:XY-relations} We have:
	\begin{align}\label{eq:WYtoomega}
		& \frac{\omega^{(g)}_{m+1,n}}{\prod_{i\in M} dX_i \cdot dX \cdot \prod_{i\in N} dY_i} =
		\\ \notag & [\hbar^{2g}]
		\sum_{j\ge0} \partial_X^j[v^j]\Bigl(-\sum_{r\ge0}L_r(v,\Theta)[u^r]\frac{dY}{dX}e^{-u\,\Theta}\cW^{Y}_{m,n+1}(u;z)
		\\ \notag & \qquad \qquad \qquad \qquad
		+\delta_{m+n,0}\frac{L_0(v,\Theta)}{v}\frac{d\Theta}{dX}
		+\delta_{m+n,0} \Theta \Bigr).
	\end{align} 
	Moreover, relations~\eqref{eq:WXtoomega} and~\eqref{eq:WYtoomega} possess the following  ``parametric'' generalization:
	\begin{align}\label{eq:WXtoW}
		\cW^{Y}_{m,n+1}(\tilde u;z)& =
		\sum_{j\ge0} \partial_Y^j [v^j]\Bigl(-\sum_{r\ge0}\tilde L_r(v,\Theta,\tilde u)[u^r]\frac{dX}{dY}e^{-u\,\Theta}\cW^{X}_{m+1,n}(u;z)
		\\ \notag & \qquad \qquad \qquad
		+\delta_{m+n,0}\frac{\tilde L_0(v,\Theta,\tilde u)}{v}\frac{d\Theta}{dY}+\delta_{m+n,0}\frac{e^{\tilde u \Theta}}{\tilde u}\Bigr),
		\\ \label{eq:WYtoW}
		\cW^{X}_{m+1,n}(u;z)& =
		\sum_{j\ge0} \partial_X^j[v^j]\Bigl(-\sum_{r\ge0}\tilde L_r(v,\Theta,u)[\tilde u^r]\frac{dY}{dX}e^{-\tilde u\,\Theta}\cW^{Y}_{m,n+1}(\tilde u;z)
		\\ \notag & \qquad \qquad \qquad
		+\delta_{m+n,0}\frac{\tilde L_0(v,\Theta,u)}{v}\frac{d\Theta}{dX}+\delta_{m+n,0}\frac{e^{u \Theta}}{u}\Bigr).
	\end{align}
\end{proposition}
See Section \ref{S.Proofs} for the proof.
The same remarks as the ones listed after Proposition~\ref{prop:xy-relations} are applied here as well.

\subsection{Equivalences and proofs}\label{sec:proofs1} The general strategy of all proofs of Propositions~\ref{prop:regular-at-diagonals},~\ref{prop:xy-relations}, and~\ref{prop:XY-relations} is to reduce them via a sequence of steps to some identities for formal power series that can be explained in terms of the vertex operator formalism and certain vacuum expectation values.

\subsubsection{Equivalences} Before we start explaining the proof, we note that a straightforward generalization of Proposition~\ref{prop:simple-rec-mixed} is the following statement:

\begin{proposition} \label{prop:equiv} Equations~\eqref{eq:Wytoomega}, \eqref{eq:WxtoW}, \eqref{eq:WytoW} are equivalent to Equations~\eqref{eq:WYtoomega}, \eqref{eq:WXtoW}, \eqref{eq:WYtoW}, respectively.
\end{proposition}
The precise statement on the equivalence of~\eqref{eq:WxtoW} and~\eqref{eq:WytoW} to \eqref{eq:WXtoW} and \eqref{eq:WYtoW}, respectively, uses the following Lemma:

\begin{lemma} \label{lem:Wx-WX}
	The functions $\cW^{x}_{m+1,n}$ and $\cW^{X}_{m+1,n}$ are related to one another by the following relation
	\begin{align}\label{eq:WxWX}
		& \prod_{i\in M}dx_i \cdot \prod_{i\in N} dy_i\cdot ({-xw})\,\cW^{x}_{m+1,n}(-x\,w;z)
		\\ \notag & =\prod_{i\in M}dX_i \cdot \prod_{i\in N} dY_i \cdot e^{-\frac12\log(1-\frac14w^2\hbar^2)\partial_X}{u(w)\,\cS(u(w)\,\hbar)}\cW^{X}_{m+1,n}(u(w);z),
	\end{align}
	where $u(w)$ is given by~\eqref{eq:uw}.
\end{lemma}

\begin{proof}
	By definition, $u\cS(u\,\hbar)\cW^{X}_{m+1,n}(u;z)$ is a linear combination of the terms of the form
	\begin{equation}\label{eq:WXS}
		\prod_{i=1}^k\Bigr(\big\lfloor_{z_i\to z}u\cS(u\hbar\partial_{X_{\bar i}})\frac{1}{dX_{\bar i}}\Bigr)\omega(z_{\bar1},\dots,z_{\bar k})
	\end{equation}
	for different $k$ and different multidifferentials $\omega$. On the other hand, ${w}\, \cW^{x}_{m+1,n}(w;z)$ is a similar linear combination of the terms of the form
	\begin{equation}\label{eq:WxS}
		\prod_{i=1}^k\Bigr(\big\lfloor_{z_i\to z}w\cS(w\hbar\partial_{x_{\bar i}})\frac{1}{dx_{\bar i}}\Bigr)\omega(z_{\bar1},\dots,z_{\bar k})
	\end{equation}
	with the same forms~$\omega$. 
	Note that from~\eqref{eq:xXyY6} we have
	\begin{equation}\label{eq:SwxSuX}
		\lfloor_{w\to -x\,w}w\cS(w\hbar\partial_{x})\circ\frac{1}{dx}=
		e^{-\frac12\log(1-\frac14w^2\hbar^2)\partial_X}u(w)\cS(u(w)\hbar\partial_{X})\circ\frac{1}{dX},
	\end{equation}
which then implies~\eqref{eq:WxWX}.
\end{proof}

\begin{proof}[Proof of Proposition~\ref{prop:equiv}] Recall the proofs of Proposition~\ref{prop:simple-rec-mixed} and Proposition~\ref{prop:MainSimplification}. The same argument works here \emph{mutatis mutandis}.
\end{proof}

So, it is sufficient to prove one of the two Propositions~\ref{prop:xy-relations} and~\ref{prop:XY-relations},
and then according to Proposition~\ref{prop:equiv} the other one is a direct corollary. However, we would like to give two parallel arguments for both Propositions~\ref{prop:xy-relations} and~\ref{prop:XY-relations}. The idea of the proof is the same for both propositions: we reduce the proof of a general case to the case of a particular formal power series expansion. The arguments for the identities for formal power series expansions are completely independent, even though the reduction to the formal power series case can be directly performed only for Proposition~\ref{prop:xy-relations}, and would use Proposition~\ref{prop:equiv} in the proof of Proposition~\ref{prop:XY-relations}.


\subsubsection{Proofs}\label{S.Proofs}

\begin{proof}[Proof of Proposition~\ref{prop:XY-relations}] Here we present the logical structure of the proof of Proposition~\ref{prop:XY-relations}, leaving the technical details of computations with vertex operators for Section~\ref{sec:vertex-operators-standard-E}.

	\begin{enumerate}
		
		\item[Step (a)]
		Recall the origin of the defining Equation~\eqref{eq:WXtoomega}, which is explained in Section~\ref{sec:MotivationFPS}. It implies that if there exists a point $p\in S$ such that $x|_p=0$ and $x$ gives a local coordinate near $p$, and $y$ has a simple pole at $p$ such that $xy = 1 + O(x)$, then the local expansions of the $\omega^{(g)}_{m,n}$ near $(p,\dots,p)\in S^{m+n}$ in the coordinates $(x_{\set{m}},y^{-1}_{\set{m+n}\setminus \set{m}})$ are given by some connected vacuum expectation values $\covac\cdots Z\vac^\circ$, where $Z$ comes from the expansions of $\omega^{(g)}_{m,0}$.
		
		\item[Step (b)]
		Since all identities that we want to prove are identities for the meromorphic functions on $S^{m+n+1}$ (defined as explicit finite differential-algebraic expressions in the initial meromorphic differentials  $\{\omega^{(g)}_{m,0}\}$), it is sufficient to prove them for the formal powers series expansions near $\mathbf{p}\coloneqq (p,\dots,p)$ (that define analytic functions on some multidisk $D^{m+n+1}\subset S^{m+n+1}$, $D\ni p$). To this end, we observe that the formal power series expansion near $\mathbf{p}$ of both sides of each of the equations~\eqref{eq:WYtoomega},~\eqref{eq:WXtoW}, and~\eqref{eq:WYtoW} have an interpretation as connected vacuum expectation values $\covac\cdots Z\vac^\circ$, and the subsequent identification of the two sides of each of three equations reduces to a computation of the involved operators in the Fock space formalism.
		
	\end{enumerate}
	Now we discuss several variations of this argument. Basically, the idea is that we can still reduce the statement of Proposition~\ref{prop:XY-relations} to a comparison of formal power series expansions given in terms of the Fock space formalism.
	\begin{enumerate}
		
		\item[Step (c)] Assume we have a point $p\in S$ such that $x|_p=0$ and $x$ gives a local coordinate near $p$, and $y$ has a simple pole at $p$, but $xy = c + o(1)$ for some non-vanishing $c\not=1$. In this case we can rescale $y\to c^{-1}y$ and $\omega^{(g)}_{m,0} \to c^{2-2g-m}\omega^{(g)}_{m,0}$. 
		Analyzing the structure of Equation~\eqref{eq:WXtoomega} we see that this rescaling implies the rescaling $\omega^{(g)}_{m,n} \to c^{2-2g-m-n}\omega^{(g)}_{m,n}$ for all $\omega^{(g)}_{m,n}$ and this rescaling preserves Equations~\eqref{eq:WYtoomega}, \eqref{eq:WXtoW}, \eqref{eq:WYtoW}. Thus we reduce the argument in this case to the one discussed in Step~(a).
		
		\item[Step (d)] Assume we have a point $p\in S$ such that $x|_p=d$ and $x-d$ gives a local coordinate near $p$, and $y$ has a simple pole at $p$ such that $(x-d)y = c + o(1)$, $c\not=0$. In fact, the shift $x\to x-d$ affects only $\omega^{(0)}_{0,1}/dy$ shifting it by a constant. It is not visible directly from Equation~\eqref{eq:WXtoomega}, but we have already established in Proposition~\ref{prop:simple-rec-mixed} that Equation~\eqref{eq:WXtoomega} is equivalent to Equation~\eqref{eq:Wxtoomega}, which is manifestly invariant under the shifts of $x$. Note that Equations~\eqref{eq:Wytoomega},~\eqref{eq:WxtoW}, and~\eqref{eq:WytoW} are also invariant under the shifts of $x$, and they are equivalent to Equations~\eqref{eq:WYtoomega},~\eqref{eq:WXtoW},~\eqref{eq:WYtoW} by Proposition~\ref{prop:equiv}. Thus we reduce the argument in this case to the one discussed in Step~(c).
		
		\item[Step (e)] Take a point $p\in S$ such that $y$ has a not necessarily simple pole at $p$, and we make no assumptions on $x$ at $p$. Let $y=z^{-\ell}$ in some local coordinate $z$ at $p$, $z|_p=0$. Restrict $\omega^{(g)}_{m,0}$'s to $D^m$, $p\in D\subset S$ for some disk $D$ and deform $y$ by changing it to $y(z,\epsilon) = (z^\ell-\epsilon)^{-1}$, $\epsilon\in U$, where $U$ is a small disk around zero. We get a new system of differentials $\omega^{(g)}_{m,n}(\epsilon)$ constructed using Equation~\eqref{eq:WXtoomega} from $\omega^{(g)}_{m,0}$'s that all are constant in $\epsilon$ except for $\omega^{(0)}_{1,0} = -y(\epsilon)dx$.
		
		For each of Equations~\eqref{eq:WYtoomega},~\eqref{eq:WXtoW},~\eqref{eq:WYtoW} the difference of the left hand side and the right hand side is a function on $D^{m+n+1}\times U$. It depends on the equation and in the case of~\eqref{eq:WXtoW}, respectively,~\eqref{eq:WYtoW} one has to consider the coefficients of the expansion of the corresponding equation in $\hbar, \tilde u$, respectively, in $\hbar, u$. Denote this function by $F=F(z_{\set{m+n+1}},\epsilon)$. It is given as a rational expression in terms of  $\omega^{(g)}_{m,0}$, $2g-2+m\geq 0$, $x(z)$, $y(z,\epsilon) = (z^\ell-\epsilon)^{-1}$, regularizing terms for $\omega^{(0)}_{(2,0)}$ and $\omega^{(0)}_{(0,2)}$, and the $\partial_z$-derivatives of all listed functions and differentials. By Step (d) $F$ is constant zero in $z_{\set{m+n+1}}$ for all values of $\epsilon$ such that $x$ and $y$ satisfy the conditions of Step (d) at $z=\epsilon^{1/\ell}$ for at least one choice of the $\ell$-th root of $\epsilon$. Since it is an open condition, we obtain that $F$ vanishes as a rational expression. Hence Equations~\eqref{eq:WYtoomega},~\eqref{eq:WXtoW},~\eqref{eq:WYtoW} are satisfied.
	\end{enumerate}
	The last argument completes the proof of Proposition~\ref{prop:XY-relations} in the most general case.
\end{proof}

Let us also present the proof (or rather a guide to the proof) of Proposition~\ref{prop:xy-relations}.

\begin{proof}[Proof of Proposition~\ref{prop:xy-relations}] The proof follows the same scheme as the proof of Proposition~\ref{prop:XY-relations}, so we mainly stress the differences. The technical details of computations with vertex operators for this case are given in Section~\ref{sec:Computations-vertex-new-form}.
	
	\begin{enumerate}
		
		\item[Step (a)]
		In fact, Equation~\eqref{eq:Wxtoomega} also has an origin in some computation of connected vacuum expectation values $\covac\cdots Z\vac^\circ$ in the bosonic Fock space (it involves a principally new form for the operator $\J$ defined in Section~\ref{sec:MotivationFPS} --- we present this interpretation below). Exactly as in the proof of Proposition~\ref{prop:XY-relations}, this implies that if there exists a point $p\in S$ such that $x|_p=0$ and $x$ gives a local coordinate near $p$, and $y$ has a simple pole at $p$ such that $xy = 1 + o(1)$, then the local expansions of the $\omega^{(g)}_{m,n}$ near $(p,\dots,p)\in S^{m+n}$ in the coordinates $(x_{\set{m}},y^{-1}_{\set{m+n}\setminus \set{m}})$ are given by some connected vacuum expectation values $\covac\cdots Z\vac^\circ$, where $Z$ comes from the expansions of $\omega^{(g)}_{m,0}$.
		
		\item[Step (b)]  Just repeats the corresponding step in the proof of Proposition~\ref{prop:XY-relations}.
		
		\item[Step (c)] Just repeats the corresponding step in the proof of Proposition~\ref{prop:XY-relations}. Note that the rescaling $y\to c^{-1}y$ and $\omega^{(g)}_{m,0} \to c^{2-2g-m}\omega^{(g)}_{m,0}$ combined with the structure of Equation~\eqref{eq:Wxtoomega} also implies the rescaling $\omega^{(g)}_{m,n} \to c^{2-2g-m-n}\omega^{(g)}_{m,n}$ for all $\omega^{(g)}_{m,n}$ and this rescaling preserves Equations~\eqref{eq:Wytoomega},~\eqref{eq:WxtoW},~\eqref{eq:WytoW}.
		
		\item[Step (d)] Just repeats the corresponding step in the proof of Proposition~\ref{prop:XY-relations}. Note that since Equations,~\eqref{eq:Wytoomega},~\eqref{eq:WxtoW}, and~\eqref{eq:WytoW} are invariant under the shifts of $x$, we don't need
		Propositions~\ref{prop:equiv} for this argument.
		
		\item[Step (e)] Just repeats the corresponding step in the proof of Proposition~\ref{prop:XY-relations}.
	\end{enumerate}
\end{proof}

\begin{proof}[Proof of Proposition~\ref{prop:regular-at-diagonals}] Here we present the logical structure of the proof of Proposition~\ref{prop:regular-at-diagonals}, leaving the technical details of computations with vertex operators for Sections~\ref{sec:vertex-operators-standard-E} and~\ref{sec:Computations-vertex-new-form}. There are two parallel arguments since we can base ourselves on one of the two equivalent formulas, Equations~\eqref{eq:omega-g-n-graphsXY} and/or~\eqref{eq:omega-g-n-graphsxy} (cf. the proofs of Propositions~\ref{prop:XY-relations} and~\ref{prop:xy-relations} above).
	
	\begin{enumerate}
		\item[Step (a)] If there exists a point $p\in S$ such that $x|_p=0$ and $x$ gives a local coordinate near $p$, and $y$ has a simple pole at $p$ such that $xy = 1 + o(1)$, then the local expansions of the $\omega^{(g)}_{m,n}$ near $(p,\dots,p)\in S^{m+n}$ in the coordinates $(x_{\set{m}},y^{-1}_{\set{m+n}\setminus \set{m}})$ are given by some connected vacuum expectation values $\covac\cdots Z\vac^\circ$, where $Z$ comes from the expansions of $\omega^{(g)}_{m,0}$. This vacuum expectation values manifestly have no poles at the respective diagonals, except for the special cases $(g,m,n)=(0,2,0)$ and $(g,m,n)=(0,0,2)$.
		There are two different formulas for the vertex operators entering $Z$ that give Equations~\eqref{eq:omega-g-n-graphsXY} and~\eqref{eq:omega-g-n-graphsxy}, respectively, but as we see below absence of the diagonal forms is a consequence of the general shape of the corresponding vacuum expectation values.
		
		\item[Step (b)]  Just repeats the corresponding step in the proof of Proposition~\ref{prop:XY-relations}.
		
		\item[Step (c)] Just repeats the corresponding step in the proof of Proposition~\ref{prop:XY-relations}. Note that the rescaling $y\to c^{-1}y$ and $\omega^{(g)}_{m,0} \to c^{2-2g-m}\omega^{(g)}_{m,0}$ combined with the structure of Equation~\eqref{eq:Wxtoomega} also implies the rescaling $\omega^{(g)}_{m,n} \to c^{2-2g-m-n}\omega^{(g)}_{m,n}$ for all $\omega^{(g)}_{m,n}$ and this rescaling preserves Equations~\eqref{eq:Wytoomega},~\eqref{eq:WxtoW},~\eqref{eq:WytoW}.
		
		\item[Step (d)] Just repeats the corresponding step in the proof of Proposition~\ref{prop:XY-relations}. If we base our proof on Equation~\eqref{eq:omega-g-n-graphsXY}, we additionally need at this point the equivalence with Equation~\eqref{eq:omega-g-n-graphsxy} (see Remark~\ref{rem:equiv-mixed}).
		
		\item[Step (e)] This step is similar to the corresponding step in the proof of Proposition~\ref{prop:XY-relations}, though we would like to expand it here, since it is not a question about a formal algebraic relation. We apply the same idea as before, that is,
		we take a point $p\in S$ such that $y$ has a not necessarily simple pole at $p$, let $y=z^{-\ell}$ in some local coordinate $z$ at $p$, $z|_p=0$, restrict $\omega^{(g)}_{m,n}$'s to $D^{m+n}$, $p\in D\subset S$ for some disk $D$ and deform $y$ by changing it to $y(z,\epsilon) = (z^\ell-\epsilon)^{-1}$, $\epsilon\in U$, where $U$ is a small disk around zero.
		
		The new system of differentials $\omega^{(g)}_{m,n}(\epsilon)$ are still meromorphic differentials with coefficients holomorphic in $\epsilon$. If $\omega^{(g)}_{m,n}(\epsilon)|_{\epsilon=0}$ has a pole at the diagonal $z_i=z_j$ for $i,j\in\set{m+n}\setminus\set{m}$, then it has it along the same diagonal in $(D')^{m+n}\subset D^{m+n}$, where $D'\subset D$ is an open disk far from the origin (further then $\epsilon^{1/\ell}$ for all $\epsilon\in U$). Then for a small enough $\epsilon\not=0$ $\omega^{(g)}_{m,n}(\epsilon)$ must retain a pole in $(D')^{m+n}\subset D^{m+n}$, which is impossible since it is an analytic continuation of the restrictions of $\omega^{(g)}_{m,n}(\epsilon)$ to a neighborhood of a simple pole of $y(z,\epsilon)$ that we treat by Steps (a)-(d) of the proof.
	\end{enumerate}

\end{proof}

All these proofs reduce the corresponding propositions to the special case when the corresponding meromorphic differentials can be expanded near a point $(p,\dots,p)$ and these expansions can be assembled into vacuum expectation values. So, the missing part consists of technical computations with the vertex operators in these vacuum expectations. The necessary computations are presented in the next two subsections.

\subsection{Computations with vertex operators: the standard form}  \label{sec:vertex-operators-standard-E}

In order to complete the proof of Proposition~\ref{prop:XY-relations}, we have to prove Proposition~\ref{prop:FirstFormOfrecursion} (which is the key statement to switch from Equation~\eqref{eq:WXtoomega} considered as the definition of the mixed correlation differentials to vacuum expectation values, once formal power series expansion of this type is possible).  Then we have to show that Equations~\eqref{eq:WYtoomega},~\eqref{eq:WXtoW}, and~\eqref{eq:WYtoW} also follow from the defining Equation~\eqref{eq:MixedFirstDef}.

\begin{proof}[Proof of Proposition~\ref{prop:FirstFormOfrecursion}. ]
	The operator  $\mathbb{J}(\hbar,z)$, defined by \eqref{JJdef}, can be presented as
	\begin{align}\label{eq:bJfromE}		
		\mathbb{J}(\hbar,z)&=
		\sum_{\ell =-\infty}^\infty z^\ell  \sum_{r=0}^\infty\partial_\theta^r e^{\ell\frac{\cS(\ell \hbar\partial_\theta)}{\cS(\hbar\partial_\theta)}\log\theta}|_{\theta=1}
		[u^r x^{\ell }]\cE(u\hbar,x),
	\end{align}
	where
	\begin{align}
		\cE(u,x)
		=\frac{1}{u\,\cS(u)}\Bigl(e^{\sum_{i=1}^\infty u\,\cS(i\,u){x^{-i}J_{-i}}}e^{\sum_{i=1}^\infty u\,\cS(i\,u){x^{i}J_{i}}}-1\Bigr) \label{eq:Ebosonic}
	\end{align}
	(see e.g.~\cite[Equation~(24)]{BDKS-OrlovScherbin}).
	Recall Equation~\eqref{eq:MixedFirstDef}
	\begin{align} \label{eq:MixedFirstDef-1}
		W_{m,n+1}=\hbar^{1-m-n} \VEVc{\prod_{i=1}^m
			J(x_i) \prod_{i=m+1}^{m+n+1}  \J(\hbar,y_i^{-1}) \;Z}
	\end{align}
	and substitute $\J(\hbar,y_{m+1}^{-1})$ by the expression given by Equation~\eqref{eq:bJfromE}.
	To describe the result of this substitution, denote by $U$ the transformation taking a Laurent series $f$ in~$u$ and~$x$ to the following series in $y^{-1}$:
	\begin{equation}\label{eq:Uoper}
		Uf=\sum_{\ell=-\infty}^\infty y^{-\ell}\sum_{r=0}^\infty
		\partial_\theta^r e^{\ell \frac{\cS(\ell \hbar\partial_\theta)}{\cS(\hbar\partial_\theta)}\log\theta}|_{\theta=1}
		[u^rx^{\ell }]f(u,x).
	\end{equation}
	Let $x=x_{m+1}$, $y=y_{m+1}$. We have:
	\begin{equation}
		 W_{m,n+1} =
		\hbar^{1-m-n} U
		\covac \prod_{i=1}^{m}
		J(x_i)%
		\cE(u\hbar ,x)
		\prod_{i=m+2}^{m+n+1} \J(\hbar,y_i^{-1})
		\;Z\vac^\circ.
	\end{equation}
	
	It is proved in~\cite[Section 2.4]{BDKS-toporec-KP} and~\cite[Section 2.4]{BDKS-symplectic}  that $\cW^X_{m+1,n}(u;x)$ defined by Equation~\eqref{eq:cWX} is given by
	\begin{equation}
		\cW^X_{m+1,n}(u;x) = \hbar^{1-m-n}
		\covac \prod_{i=1}^{m}
		J(x_i)%
		\cE(u\hbar ,x)
		\prod_{i=m+2}^{m+n+1} \J(\hbar,y_i^{-1})
		\;Z\vac^\circ
	\end{equation}
	(as usual, we omit the arguments $\hbar, x_{\set{m}}, y_{\set{m+n+1}\setminus \set{m+1}}$ in the notation).
	Note that we thus have
	\begin{equation}\label{eq:WUW}
		W_{m,n+1} = U\, \cW^X_{m+1,n}(u;x).
	\end{equation}
	Now, using the \emph{principal identity} from~\cite[Section 4]{BDKS-OrlovScherbin}, see also~\cite[Section 3]{BDKS-symplectic}, one can rewrite Equation~\eqref{eq:Uoper} as 
	%
	\begin{align}
			Uf=-\sum_{r,j\geq 0}(-y\partial_y)^j\Bigl(
		[v^j]\frac{\partial_\theta^r e^{v \frac{\cS(v\hbar\partial_\theta)}{\cS(\hbar\partial_\theta)}\log\theta}}{\theta^v}|_{\theta=x\,y}
		\frac{dx}{x}\frac{y}{dy}
		[u^r]e^{-u\,W^{(0)}_{1,0}(x)}f(u;x)
		\Bigr)
	\end{align}
	in the case $m+n\not=0$, with $y$ and $x$ related by $y=x^{-1}(1+W^{(0)}_{1,0}(x))$ near $y=\infty$, $x=0$. Taking into account~\eqref{eq:WUW}, we obtain the first line of Equation~\eqref{eq:mainrecchange}. If $m+n=0$, then we also get the second line of Equation~\eqref{eq:mainrecchange} that comes from regularization of the negative powers of $u$ as in~\cite[Section 6]{BDKS-OrlovScherbin}, see also~\cite[Section 3]{BDKS-symplectic}.
\end{proof}

 Now, to complete the proof of Proposition~\ref{prop:XY-relations}, we have to derive  Equations~\eqref{eq:WYtoomega}, \eqref{eq:WXtoW}, and~\eqref{eq:WYtoW}. Let us start with
Equation~\eqref{eq:WYtoomega}. From the definition $\J(z) \coloneqq \cD^{-1} J(z) \cD$ we have $J(z)= \cD \J(z) \cD^{-1}$, and under substitution of~\eqref{eq:bJfromE} one immediately concludes that
\begin{equation}\label{eq:DbJDi-fromE}		
	{J}(z)=
	\sum_{\ell =-\infty}^\infty z^\ell  \sum_{r=0}^\infty\partial_\theta^r e^{-\ell \frac{\cS(\ell \hbar\partial_\theta)}{\cS(\hbar\partial_\theta)}\log\theta}|_{\theta=1}
	[u^r y^{-\ell }] \cD^{-1} \cE(u\hbar ,y^{-1}) \cD.
\end{equation}

We substitute expression~\eqref{eq:DbJDi-fromE} into
\begin{align} \label{eq:MixedFirstDef-2}
	W_{m+1,n}=\hbar^{1-m-n} \VEVc{\prod_{i=1}^{m+1}
		J(x_i) \prod_{i=m+2}^{m+n+1}  \J(\hbar,y_i^{-1}) \;Z}
\end{align}
and describe the result of this substitution using the transformation $V$ that takes a Laurent series $f$ in~$u$ and~$y^{-1}$ to the following series in $x$:
\begin{equation}
	Vf=\sum_{\ell=-\infty}^\infty x^{\ell}\sum_{r=0}^\infty
	\partial_\theta^r e^{-\ell \frac{\cS(\ell \hbar\partial_\theta)}{\cS(\hbar\partial_\theta)}\log\theta}|_{\theta=1}
	[u^ry^{-\ell }]f(u,y^{-1}).
\end{equation}
We have:
\begin{align}
	W_{m+1,n} & =
	V \hbar^{1-m-n}
	\covac \prod_{i=1}^{m}
	J(x_i) 
	\cD^{-1}
        \cE(u\hbar ,y^{-1})
	\cD
	\prod_{i=m+2}^{m+n+1} \J(\hbar,y_i^{-1})
	\;Z\vac^\circ
	\\ \notag
	& = V \cW^{Y}_{m,n+1}(u; y^{-1}),
\end{align}
and proceeding the same way as in the proof of Proposition~\ref{prop:FirstFormOfrecursion} we convert the latter equation with operator $V$ to Equation~\eqref{eq:WYtoomega}. Note two changes of sign (in the degree of the resulting variable, which is $x$ as opposed to $y^{-1}$, and $\ell \to -\ell$ in the exponent in the definition of $V$ as opposed to the definition of $U$) --- these two adjustments of signs compensate each other, so the resulting Equation~\eqref{eq:WYtoomega} looks very similar to Equation~\eqref{eq:WXtoomega}.

The proof of Equations~\eqref{eq:WXtoW} and~\eqref{eq:WYtoW} goes along exactly the same lines, and they are based on the following two formulas (they are again special cases of~\cite[Equation~(24)]{BDKS-OrlovScherbin}):
\begin{align}
	\cD^{-1}\cE(\tilde u \hbar;y^{-1})\cD
	&=\sum_{\ell =-\infty}^\infty y^{-\ell}  \sum_{r=0}^\infty\partial_\theta^r e^{\tilde u \theta+\ell \frac{\cS(\ell \hbar\partial_\theta)}{\cS(\hbar\partial_\theta)}\log\theta}|_{\theta=1}
	[u^r x^{\ell }]\cE(\hbar u,x); \\
	\cD\cE( u \hbar;x)\cD^{-1}
	&=\sum_{\ell =-\infty}^\infty x^{\ell}  \sum_{r=0}^\infty\partial_\theta^r e^{ u \theta-\ell \frac{\cS(\ell \hbar\partial_\theta)}{\cS(\hbar\partial_\theta)}\log\theta}|_{\theta=1}
	[\tilde u^r y^{-\ell }]\cE(\hbar \tilde u,y^{-1}).
\end{align}

In order to complete the proof of Proposition~\ref{prop:regular-at-diagonals}, we rewrite Equation~\eqref{eq:MixedFirstDef-1} as
\begin{equation}
	W_{m,n}=\hbar^{2-m-n} \VEVc{\Big(\prod_{i=1}^{m}
		\cD J(x_i)\cD^{-1} \Big) \prod_{i=m+1}^{m+n}  J(y_i^{-1})\; \Big(\cD Z\Big)}
\end{equation}
and note that by the inclusion-exclusion formula the only pole at the diagonal $z_i=z_j$ for $i,j\in\set{m+n}\setminus \set{m}$ is possible at $[\hbar^0]$ in $W_{0,2}$, where it is coming from the term $\VEVc{J(y_1^{-1})J(y_2^{-1})} = y_1^{-1} y_2^{-1}/(y_1^{-1}-y_2^{-1})^2$, cf.~\cite[Corollary 4.10]{BDKS-OrlovScherbin}.

This completes the proofs of Proposition~\ref{prop:regular-at-diagonals} and thus Proposition~\ref{prop:XY-relations}, and hence, by Proposition~\ref{prop:equiv}, the proof of Proposition~\ref{prop:xy-relations}.
\begin{remark} In the spirit of the above computations, Equation~\eqref{eq:omega-g-n-graphsXY} gets the following explanation: we just use the graphical formula given in~\cite[Theorem 4.14]{BDKS-FullySimple} to compute
\begin{align}
	\hbar^{2-m-n} \VEVc{\prod_{i=1}^m
		\J(u_i\hbar ,x_i) \prod_{i=m+1}^{m+n}  \J(u_i\hbar,y_i^{-1}) \;Z},
\end{align}
and then further restrict it with $\prod_{i=1}^m [u_i^0] \prod_{i=m+1}^{m+n} \restr {u_i} 1$ to obtain $W_{m,n}$.
\end{remark}

\subsection{Computations with vertex operators: a new form}

\label{sec:Computations-vertex-new-form}
While we already have a proof of Propositions~\ref{prop:xy-relations} and~\ref{prop:simple-rec-mixed}, we present here an alternative proof (in the spirit of Section~\ref{sec:vertex-operators-standard-E}) for the case of formal series (in the setup of Section~\ref{sec:FormalPowerSeries}).

Let us make a change of arguments in the operator $\cE(u,z)$ implied by the equalities $z e^{u/2}=x(1+w/2)$, $z e^{-u/2}=x(1-w/2)$, that is $z=(1-w^2/4)^{1/2}x$, $u=\log\bigl(\frac{1+w/2}{1-w/2}\bigr)$ (cf. Section~\ref{sec:ProofPropositionSimplificationformal}). Having this substitution in mind, we define
\begin{align}
	\tilde\cS(w,\delta)&=\frac{(1+w/2)^{\delta}-(1-w/2)^{\delta}}{\delta},
	\\
	\tilde\cE(w,x)&=(1-w^2/4)^{-1/2}\cE\bigl(\log\bigl(\tfrac{1+w/2}{1-w/2}\bigr),(1-w^2/4)^{1/2}x\bigr)\label{eq:tcS}
	\\ \notag &
	=\frac{1}{w}\Bigl(e^{\sum_{i=1}^\infty-\tilde\cS(w,-i){x^{-i}J_{-i}}}e^{\sum_{i=1}^\infty\tilde\cS(w,i){x^{i}J_{i}}}-1\Bigr),
	\\\eta(w)&
	=e^{(\hbar^{-1}\tilde\cS(w\hbar,x\partial_x)-w)1}
	=e^{\hbar^{-1}\log\left(\frac{1+w\hbar/2}{1-w\hbar/2}\right)-w}.
\end{align}
Here $\eta(w)$ is the same as in \eqref{eq:etadef}.

Recall Lemma~\ref{lem:phiidentity}. It claims that
\begin{equation}\label{eq:phiidentity}
	e^{v\frac{\cS(v \hbar\partial_\theta)}{\cS(\hbar\partial_\theta)}\log\theta}
	=\sum_{r\ge0}\Bigl(\partial_\theta^r\theta^v\Bigr)[w^r](1-w^2\hbar^2/4)^{\frac{v-1}{2}}
	\eta(w)^\theta.
\end{equation}
Applying the Taylor expansion formula (see~\cite[Lemma~4.5]{BDKS-OrlovScherbin}) to this (and replacing the argument $\theta$ with $\sigma$), we get 
\begin{align}
	e^{v\frac{\cS(v \hbar\partial_\sigma)}{\cS(\hbar\partial_\sigma)}\log\sigma}
	&=\sum_{r\ge0}\Bigl(\partial_\theta^r\theta^v\Bigr)\Big|_{\theta=1}[w^r]e^{w(\sigma-1)}(1-w^2\hbar^2/4)^{\frac{v-1}{2}}
	\eta(w)^\sigma \\ \nonumber
	&=\sum_{r\ge0}\Bigl(\partial_\theta^r\theta^v\Bigr)\Big|_{\theta=1}[w^r]e^{-w}(1-w^2\hbar^2/4)^{\frac{v-1}{2}}
	\left(\frac{1+w\hbar/2}{1-w\hbar/2}\right)^{\sigma/\hbar}.
\end{align}
%
Substituting this into~\eqref{eq:bJfromE}, we get
\begin{align}\label{eq:bJfromEnew}
	&\mathbb{J}(\hbar,z)=\sum_{\ell =-\infty}^\infty z^\ell  \sum_{j=0}^\infty\left.\left(\partial_{\sigma}^j \sum_{r\ge0}\partial_\theta^r\theta^\ell\bigm|_{\theta=1}[w^r]e^{-w}(1-w^2\hbar^2/4)^{\frac{\ell-1}{2}}
	\left(\frac{1+w\hbar/2}{1-w\hbar/2}\right)^{\sigma/\hbar}\right)\right|_{\sigma=1}\\ \nonumber
	&\phantom{=}\times[u^j x^{\ell }]\cE(u\hbar,x)\\ \nonumber
	&=\sum_{\ell =-\infty}^\infty z^\ell  \sum_{r=0}^\infty\Bigl(\partial_\theta^r\theta^\ell\bigm|_{\theta=1}\Bigr)
	[w^r x^{\ell }]e^{-w}(1-w^2\hbar^2/4)^{\frac{\ell-1}{2}} \sum_{j=0}^\infty \left.\left(\partial_{\sigma}^j 
	\left(\frac{1+w\hbar/2}{1-w\hbar/2}\right)^{\sigma/\hbar}\right)\right|_{\sigma=1}\\ \nonumber
	&\phantom{=}\times[u^j]\cE(u\hbar,x)\\ \nonumber
	&=\sum_{\ell =-\infty}^\infty z^\ell  \sum_{r=0}^\infty\Bigl(\partial_\theta^r\theta^\ell\bigm|_{\theta=1}\Bigr)
	[w^r x^{\ell }]\eta(w)(1-w^2\hbar^2/4)^{\frac{\ell-1}{2}} \sum_{j=0}^\infty \left(\hbar^{-1}\log \left(
	\frac{1+w\hbar/2}{1-w\hbar/2}\right)\right)^j\\ \nonumber
	&\phantom{=}\times[u^j]\cE(u\hbar,x)\\ \nonumber
	&=\sum_{\ell =-\infty}^\infty z^\ell  \sum_{r=0}^\infty\Bigl(\partial_\theta^r\theta^\ell\bigm|_{\theta=1}\Bigr)
	[w^r x^{\ell }]\eta(w)(1-w^2\hbar^2/4)^{-\frac{1}{2}}\cE\bigl(\log\bigl(\tfrac{1+w\hbar/2}{1-w\hbar/2}\bigr),(1-w^2\hbar^2/4)^{1/2}x\bigr)\\ \nonumber
	&=\sum_{\ell =-\infty}^\infty z^\ell  \sum_{r=0}^\infty\Bigl(\partial_\theta^r\theta^\ell\bigm|_{\theta=1}\Bigr)
	[w^r x^{\ell }]\eta(w)\tilde\cE(\hbar w,x).
\end{align}


First, we give a proof of Proposition~\ref{prop:simple-rec-mixed} for the case of formal power series (in the setup of Section~\ref{sec:FormalPowerSeries}) similar to the proof of Proposition~\ref{prop:FirstFormOfrecursion} given in Section~\ref{sec:vertex-operators-standard-E}. The computation provided below actually served as a motivation for the proof of Proposition~\ref{prop:simple-rec-mixed} for the general case given in Section~\ref{sec:proof-of-algebraic-equivalence-of-relations}.
\begin{proof}[Proof of Proposition~\ref{prop:simple-rec-mixed} in the setup of Section~\ref{sec:FormalPowerSeries}.]
	Recall Equation~\eqref{eq:MixedFirstDef-1}
	\begin{align} \label{eq:MixedFirstDef-3}
		W_{m,n+1}=\hbar^{1-m-n} \VEVc{\prod_{i=1}^m
			J(x_i) \prod_{i=m+1}^{m+n+1}  \J(\hbar,y_i^{-1}) \;Z}
	\end{align}
	and substitute $\J(\hbar,y_{m+1}^{-1})$ by the expression given by Equation~\eqref{eq:bJfromEnew}. Let $x=x_{m+1}$, $y=y_{m+1}$.
	We get
	\begin{align}\label{eq:Wnewcor}
		 W_{m,n+1} &=
		\hbar^{1-m-n} \sum_{\ell =-\infty}^\infty y^{-\ell } \sum_{r=0}^\infty\Bigl(\partial_\theta^r\theta^\ell\bigm|_{\theta=1}\Bigr)
		[w^r x^{\ell }]\eta(w)
		\\ \notag &
		\times \covac \prod_{i=1}^{m}
		J(x_i)%
		\tilde\cE(\hbar w,x)
		\prod_{i=m+2}^{m+n+1} \J(\hbar,y_i^{-1})
		\;Z\vac^\circ\\ \notag
		&=
		\hbar^{1-m-n} \sum_{\ell =-\infty}^\infty y^{-\ell } \sum_{r=0}^\infty\Bigl(\partial_\theta^r\theta^\ell\bigm|_{\theta=1}\Bigr)
		[w^r x^{\ell }](-x)^{r} \eta(-w/x)
		\\ \notag &
		\times \covac \prod_{i=1}^{m}
		J(x_i)%
		\tilde\cE(-\hbar w/x,x)
		\prod_{i=m+2}^{m+n+1} \J(\hbar,y_i^{-1})
		\;Z\vac^\circ\\ \notag
		&=
		\hbar^{1-m-n} \sum_{\ell =-\infty}^\infty y^{-\ell } \sum_{r=0}^\infty\Bigl(\partial_\theta^r\theta^\ell\bigm|_{\theta=1}\Bigr)
		[w^r x^{\ell }](-x)^{r} e^{w/x} e^{\hbar^{-1}\log\left(\frac{1-w\hbar/2x}{1+w\hbar/2x}\right)}
		\\ \notag &
		\times \covac \prod_{i=1}^{m}
		J(x_i)%
		\tilde\cE(-\hbar w/x,x)
		\prod_{i=m+2}^{m+n+1} \J(\hbar,y_i^{-1})
		\;Z\vac^\circ.
	\end{align}
	Recall from \eqref{eq:omega-W} that
	\begin{equation}\label{eq:omeganewW}
		\dfrac{\omega^{(g)}_{m,n+1}}{\prod_{j=1}^m dx_j\prod_{i=m+1}^{m+n+1} dy_i} = \left(W^{(g)}_{m,n+1}+\delta_{(g,m,n),(0,0,0)}\right)\prod_{j=1}^m \left(-\dfrac{1}{x_i}\right) \prod_{i=m+1}^{m+n+1} \left(-\dfrac{1}{y_i}\right).
	\end{equation}
Denote
\begin{equation}\label{eq:Utilde}
	\tilde Uf:=-\dfrac{1}{y}\sum_{\ell =-\infty}^\infty y^{-\ell}  \sum_{r=0}^\infty\Bigl(\partial_\theta^r\theta^\ell\bigm|_{\theta=1}\Bigr)
	[w^r x^{\ell }](-x)^{r+1} e^{w/x} f(w,x).
\end{equation}
	Combining \eqref{eq:Wnewcor},  \eqref{eq:omeganewW}, and \eqref{eq:Utilde}, we have, for $(g,m,n)\neq (0,0,0)$:
	\begin{align}\label{eq:omegaUcor}
		& \dfrac{\omega^{(g)}_{m,n+1}}{\prod_{j=1}^m dx_j\prod_{i=m+1}^{m+n+1} dy_i} =\tilde U\;[\hbar^{2g}]\,e^{\hbar^{-1}\log\left(\frac{1-w\hbar/2x}{1+w\hbar/2x}\right)} \prod_{j=1}^m \left(-\dfrac{1}{x_i}\right) \prod_{i=m+2}^{m+n+1} \left(-\dfrac{1}{y_i}\right)
		\\ \notag &
		\times\left(-\frac{1}{x}\right) \hbar^{1-m-n}
		\covac \prod_{i=1}^{m}
		J(x_i)%
		\tilde\cE(-\hbar w/x,x)
		\prod_{i=m+2}^{m+n+1} \J(\hbar,y_i^{-1})
		\;Z\vac^\circ.
	\end{align}
	
	Analogously to~\cite[Section 2.4]{BDKS-toporec-KP} and~\cite[Section 2.4]{BDKS-symplectic}, one can show that $\cW^x_{m+1,n}(w;x)$ defined by Equation~\eqref{eq:cW} is given by
	\begin{align}\label{eq:cWcor}
		&\cW^x_{m+1,n}(w;x) = \left(e^{-\mathcal{S}(w\hbar \partial_x)\frac{w}{x}}\right)  \prod_{j=1}^m \left(-\dfrac{1}{x_i}\right) \prod_{i=m+2}^{m+n+1} \left(-\dfrac{1}{y_i}\right)
		\\ \notag & \times \left(-\frac{1}{x}\right) \hbar^{1-m-n}
		\covac \prod_{i=1}^{m}
		J(x_i)%
		\tilde\cE(-\hbar w/x,x)
		\prod_{i=m+2}^{m+n+1} \J(\hbar,y_i^{-1})
		\;Z\vac^\circ.
	\end{align}
	This follows from the relation	
	\begin{align}
		\sum_{i=1}^\infty \tilde S(-\hbar w/x,i) x^i J_i =  w\hbar S(w\hbar\partial_x)\; \left(-\dfrac{1}{x}\right) \sum_{i=1}^\infty x^i J_i. 
	\end{align} 
	The factor  $e^{-\mathcal{S}(w\hbar \partial_x)\frac{w}{x}}$ comes from the term $\delta_{(g,m,n),(0,0,0)}$ in \eqref{eq:omeganewW}. Indeed, Equation \eqref{eq:cW} is written purely in terms of $\omega$, while the argument from~\cite[Section 2]{BDKS-toporec-KP} produces a very similar formula, but in terms of $W$, and we need to account for the $\delta_{(g,m,n),(0,0,0)}$ term, which arises in the $e^{\cT^x_{1,0}(w;z)}$ factor, and the difference is precisely $e^{-\mathcal{S}(w\hbar \partial_x)\frac{w}{x}}$. Now, using~\eqref{eq:SpxSxpx} for $d=0$, one can write
	\begin{equation}
		e^{-\mathcal{S}(w\hbar \partial_x)\frac{w}{x}} = e^{u(-w/x)} = e^{\hbar^{-1}\log \left(\tfrac{1-w \,\hbar/2x}{1+w\,\hbar/2x}\right)}
	\end{equation}
(here $u(-w/x)$ is the function $u$ defined in~\eqref{eq:uw} with the argument $-w/x$).

	Thus, we get
	\begin{equation}\label{eq:omuw}
		\dfrac{\omega^{(g)}_{m,n+1}}{\prod_{j=1}^m dx_j\prod_{i=m+1}^{m+n+1} dy_i} = \tilde U\, \cW^{x,(g)}_{m+1,n}(w;x).
	\end{equation}
	
	Now let us prove that the operator $\tilde U$ can be rewritten as follows, for $f(w,x)$ being polynomial in $w$ (this condition always holds in \eqref{eq:omuw} for $(g,m,n)\neq(0,0,0)$):
	\begin{align}
		\tilde Uf=-\sum_{r\ge0} \partial_y^r
		[w^r] \frac{dx}{dy} e^{w\,y}f(w,x),
	\end{align}
	for dependent variables $y$ and $x$ near $y=\infty$, $x=0$. 	
	First, we have
	\begin{align}
		&\tilde Uf=-\dfrac{1}{y}\sum_{\ell =-\infty}^\infty y^{-\ell}  \sum_{r=0}^\infty\Bigl(\partial_\theta^r\theta^\ell\bigm|_{\theta=1}\Bigr)
		[w^r x^{\ell }](-x)^{r+1}e^{w/x} f(w,x)\\ \nonumber
		&=-\dfrac{1}{y}\,\sum_{\ell =-\infty}^\infty y^{-\ell}[x^\ell]  \sum_{r=0}^\infty\Bigl(\partial_\theta^r\theta^\ell\bigm|_{\theta=1}\Bigr)		
		[w^r]e^{w (xy-1)}e^{-w (xy-1)} (-x) e^{-w} f(-wx,x)\\ \nonumber
		&=-\dfrac{1}{y}\,\sum_{\ell =-\infty}^\infty y^{-\ell}[x^\ell]  \sum_{r=0}^\infty\Bigl(\partial_\theta^r\theta^\ell\bigm|_{\theta=xy}\Bigr)		
		[w^r]e^{-w xy} (-x) f(-wx,x),
	\end{align}
	where in the last line we have used the Taylor expansion formula (see~\cite[Lemma~4.5]{BDKS-OrlovScherbin}). Then we take the $\partial_\theta$ derivatives and rewrite this as
	\begin{align}
		&-\dfrac{1}{y}\,\sum_{\ell =-\infty}^\infty y^{-\ell}[x^\ell]  \sum_{r=0}^\infty\left(\prod_{i=0}^{r-1}(\ell-i)\right) (xy)^{\ell-r}
		[w^r]e^{-w xy} (-x) f(-wx,x)\\ \nonumber		
		&=\sum_{r=0}^\infty(-1)^{r+1}\,\sum_{\ell =-\infty}^\infty \partial_y^r y^{-\ell+r-1} [x^\ell]   (xy)^{\ell-r}
		[w^r]e^{w y} (-x)^{r+1} f(w,x)\\ \nonumber		
		&=\sum_{r=0}^\infty\sum_{\ell =-\infty}^\infty \partial_y^r y^{-\ell+r-1} [x^\ell]   y^{\ell-r}
		[w^r]e^{w y} x^{\ell+1} f(w,x).
	\end{align}
Shifting the summation index $\ell$ by $r-1$, we rewrite this as
\begin{align}
		& \sum_{r=0}^\infty \sum_{\ell =-\infty}^\infty \partial_y^r y^{-\ell} [x^{\ell+r-1}] y^{\ell-1}
		[w^r]e^{w y} x^{\ell+r} f(w,x)\\ \nonumber		
		&=\sum_{r=0}^\infty[w^r]\partial_y^r\sum_{\ell =-\infty}^\infty  y^{-\ell} [x^{\ell}]   (xy)^{\ell}
		\,\dfrac{x}{y}\, e^{w y}  f(w,x)\\ \nonumber		
		&=-\sum_{r=0}^\infty[w^r]\partial_y^r \,\dfrac{y}{x}\, \dfrac{dx}{dy}
		\,\dfrac{x}{y}\, e^{w y}  f(w,x)\\ \nonumber		
		&=-\sum_{r=0}^\infty[w^r]\partial_y^r \, \dfrac{dx}{dy}
		\, e^{w y}  f(w,x),
	\end{align}
where we have used~\cite[Lemma~4.7]{BDKS-OrlovScherbin}, a variant of the Lagrange--B\"urmann formula, in the third line.
	
	This concludes the proof of Equation~\eqref{eq:Wxtoomega}.
\end{proof}


Now, to complete the alternative proof of Proposition~\ref{prop:xy-relations} in this setup, we have to derive  Equations~\eqref{eq:Wytoomega},~\eqref{eq:WxtoW}, and~\eqref{eq:WytoW}. Let us start with proving
Equation~\eqref{eq:Wytoomega}.

Note that
\begin{align}\label{eq:DbJDi-fromEtilde}	
	{J}(z)&=\sum_{\ell =-\infty}^\infty z^\ell  \sum_{r=0}^\infty\Bigl(\partial_\theta^r\theta^{-\ell}\bigm|_{\theta=1}\Bigr)
	[w^r y^{-\ell }]\eta(w)\cD^{-1}  \tilde\cE(\hbar w,y^{-1}) \cD,	
\end{align}
which can be obtained similarly to~\eqref{eq:bJfromEnew}, but starting with~\eqref{eq:DbJDi-fromE} instead of~\eqref{eq:bJfromE}.

We substitute~\eqref{eq:DbJDi-fromEtilde} into
\begin{align} \label{eq:MixedFirstDef-5}
	W_{m+1,n}=\hbar^{1-m-n} \VEVc{\prod_{i=1}^{m+1}
		J(x_i) \prod_{i=m+2}^{m+n}  \J(\hbar,y_i^{-1}) \;Z}
\end{align}
in place of $J(x_{m+1})$ 
and describe the result of this substitution using the transformation $\tilde V$ that takes a Laurent series $f$ in~$w$ and~$y^{-1}$ to the following series in $x$:
\begin{equation}\label{eq:tildeV}
	\tilde Vf=-\dfrac{1}{x}\sum_{\ell =-\infty}^\infty x^{\ell}  \sum_{r=0}^\infty\Bigl(\partial_\theta^r\theta^{-\ell}\bigm|_{\theta=1}\Bigr)
	[w^r y^{-\ell }](-y)^{r+1} e^{w/y} f(w,y^{-1}).
\end{equation}
Similarly to~\eqref{eq:omegaUcor}--\eqref{eq:cWcor}, one can conclude that
\begin{align}\label{eq:omegatVcW}
	 \dfrac{\omega^{(g)}_{m+1,n}}{\prod_{j=1}^{m+1} dx_j\prod_{i=m+2}^{m+n+1} dy_i} = \tilde V \cW^{y,(g)}_{m,n+1}(w; y^{-1}),
\end{align}
where
\begin{align}\label{eq:cWycor}
		&\cW^y_{m,n+1}(w;y^{-1}) = \left(e^{-\mathcal{S}(w\hbar \partial_y)\frac{w}{y}}\right)  \prod_{j=1}^m \left(-\dfrac{1}{x_i}\right) \prod_{i=m+2}^{m+n+1} \left(-\dfrac{1}{y_i}\right)
		\\ \notag & \times \left(-\frac{1}{y}\right) \hbar^{1-m-n}
		\covac \prod_{i=1}^{m}
		J(x_i)%
		\cD^{-1}\tilde\cE(-\hbar w/y,y^{-1})\cD
		\prod_{i=m+2}^{m+n+1} \J(\hbar,y_i^{-1})
		\;Z\vac^\circ.
\end{align}
The proof that the RHS of~\eqref{eq:cWycor} coincides with $\cW^y_{m,n+1}(w;y^{-1})$ as defined in~\eqref{eq:cTy}--\eqref{eq:cWy} is analogous to the proof of~\eqref{eq:cWcor} .

Proceeding the same way as in the proof of Proposition~\ref{prop:simple-rec-mixed} in the present section, one can convert Equation~\eqref{eq:omegatVcW} with operator $\tilde V$ given by~\eqref{eq:tildeV} to Equation~\eqref{eq:Wytoomega}. 

The proofs of Equations~\eqref{eq:WxtoW} and~\eqref{eq:WytoW} go along similar lines. 


\section{Topological recursion}

\label{sec:TopologicalRecursion}

The goal of this Section is to prove Theorems~\ref{thm:mainconjholds} and~\ref{thm:simpleformulaxyswap}. To this end we extensively discuss the structure of poles of the differentials $\omega^{(g)}_{m,n}$ and the loop equations that they satisfy.

\subsection{Recollections on topological recursion} \label{sec:RecollectionTR}

Let $p_1,\dots,p_a\in S$ be the zeros of $dx$. We assume that all zeros of $dx$ are simple. Let $\sigma_i$ be the deck transformations of $x$ near $p_i$. It is proved in~\cite{BEO-loop,BS-blobbed} that the differentials $\omega^{(g)}_{m,0}$ satisfy the topological recursion for the input data $(S,x,y,B)$ if and only if they satisfy the following properties:

\begin{itemize}
	\item The \emph{linear loop equations}: for any $g,m\geq 0$, $i=1,\dots,a$
	\begin{equation} \label{eq:LLE-original}
		\omega^{(g)}_{m+1,0}(z_{\llbracket m \rrbracket},z) + \omega^{(g)}_{m+1,0}(z_{\llbracket m \rrbracket},\sigma_i(z))
	\end{equation}
	is holomorphic at $z\to p_i$ and has at least a simple zero in $z$ at $z=p_i$.
	\item The \emph{quadratic loop equations}: for any $g,m\geq 0$, $i=1,\dots,a$, the quadratic differential in $z$
	\begin{equation} \label{eq:QLE-original}
		\omega^{(g-1)}_{m+2,0}(z_{\llbracket m \rrbracket},z,\sigma_i(z)) + \sum_{\substack{g_1+g_2=g\\ I_1\sqcup I_2 = \llbracket m \rrbracket
		}} \omega^{(g_1)}_{|I_1|+1,0} (z_{I_1},z)\omega^{(g_2)}_{|I_2|+1,0} (z_{I_2},\sigma_i(z))
	\end{equation}
	is holomorphic at $z\to p_i$ and has at least a double zero at $z=p_i$.
	f
	\item  The \emph{projection property}: for any $g\geq 0$, $m\geq 1$, $2g-2+m>0$
	\begin{equation}
		\omega^{(g)}_{m,0}(z_{\set{m}}) = \sum_{i_1,\dots,i_m=1}^a \Bigg(\prod_{j=1}^m \res_{z_j'= p_{i_j}} \int^{z_j'}_{p_{i_j}} B(z_j',z_j)\Bigg) \omega^{(g)}_{m,0}(z'_{\set{m}}). 
	\end{equation}
	(by product in the parenthesis we mean the product of operators applied to $\omega^{(g)}_{m,0}(z'_{\set{m}})$).
\end{itemize}

We would like to make yet another convenient reformulation of this criterion.

\begin{definition} (cf.~Definition~\ref{def:FirstDefXi-space}) We denote by $\Xi^x$ the space of functions defined in a neighborhood of the zero locus of $dx$ on $S$ and having odd polar part with respect to the involution~$\sigma$,
\begin{equation}
f\in\Xi^x\quad\Leftrightarrow\quad f(z)+f(\sigma_i(z))\text{ is holomorphic at $z\to p_i$.}
\end{equation}
\end{definition}

\begin{remark} \label{rem:Xi-explained} Any function with poles of order at most one belongs to $\Xi^x$. Besides, $\Xi^x$ is preserved by the operator $\partial_x$. It follows that any function of the form $\partial_x^kf$, $k\ge0$, belongs to $\Xi^x$ if $f$ has at most simple poles at $p_i$. Moreover, its is easy to see by local computations that $\Xi^x$ is actually spanned by the functions of this form. 
	
Note also that any function defined in a neighborhood of the zero locus of $dx$ with at most simple poles at the zero locus of $dx$ can be represented as a linear combination of functions of the form $\partial_x^kf$, $k=0,1$, where $f$ is holomorphic. This observation allows to defined $\Xi^x$ as it was done in Definition~\ref{def:FirstDefXi-space}. 
\end{remark}

Then the linear and quadratic loop equations on the differentials $\omega^{(g)}_{m,0}=\oW^{(g)}_{m,0}\prod_{i=1}^mdx_i$ can be represented in the following equivalent form:
\begin{gather}\label{eq:LLE-W}
		\oW^{(g)}_{m+1,0}(z_{\llbracket m \rrbracket},z)\in\Xi^x\\
\label{eq:QLE-W}
		\oW^{(g-1)}_{m+2,0}(z_{\llbracket m \rrbracket},z,z) + \sum_{\substack{g_1+g_2=g\\ I_1\sqcup I_2 = \llbracket m \rrbracket
		}} \oW^{(g_1)}_{|I_1|+1,0} (z_{I_1},z)\oW^{(g_2)}_{|I_2|+1,0} (z_{I_2},z)\in\Xi^x.
\end{gather}

The first equation is a reformulation of the linear loop equation. Indeed, the form~\eqref{eq:LLE-original} being symmetric with respect to~$\sigma_i$ has zero of order at least one at~$p_i$ so that after division by~$dx$ it is still holomorphic and this is what~\eqref{eq:LLE-W} expresses. The second equation is not exactly~\eqref{eq:QLE-original} but is equivalent to it modulo the linear loop equation. In order to see that, define
	\begin{equation}
		Q(z',z'')\coloneqq \oW^{(g-1)}_{m+2,0}(z_{\llbracket m \rrbracket},z',z'') + \sum_{\substack{g_1+g_2=g\\ I_1\sqcup I_2 = \llbracket m \rrbracket
		}} \oW^{(g_1)}_{|I_1|+1,0} (z_{I_1},z')\oW^{(g_2)}_{|I_2|+1,0} (z_{I_2},z'')	
	\end{equation}
and denote by $\sigma'$ and $\sigma''$ the action of the involution $\sigma=\sigma_i$ on the argument $z'$ and~$z''$, respectively. Then, we have
\begin{multline}
Q(z,z)+Q(\sigma(z),\sigma(z))=\big\lfloor_{z'=z''=z}(1+\sigma'\sigma'')Q(z',z'')\\
=\big\lfloor_{z'=z''=z}(1+\sigma')(1+\sigma'')Q(z',z'')-\big\lfloor_{z'=z''=z}(\sigma'+\sigma'')Q(z',z'').
\end{multline}
The first summand is holomorphic by the linear loop equation, and the second summand coincides with~\eqref{eq:QLE-original}, up to the factor $-\frac2{dx^2\prod_{i=1}^m dx_i}$. Thus the quadratic loop equation is equivalent to the condition~\eqref{eq:QLE-W} if the linear loop equation is satisfied.

\medskip
Recall that the form~$B$ entering the spectral curve data is defined uniquely if~$S=\Cf P^1$. If the genus of~$S$ is greater than zero then the form $B$ is defined up to a biholomorphic summand that can be fixed uniquely by choosing a symplectic basis in $H_1(S)$ formed of a system of $(\mathfrak{A},\mathfrak{B})$-cycles and by imposing an additional requirement that all $\mathfrak{A}$-periods of~$B$ are vanishing with respect to both its arguments.

\begin{proposition}\label{prop:trueTRformulation} The collection of symmetric meromorphic differentials $\omega^{(g)}_{m,0}$ satisfies the topological recursion for the spectral curve data $(S,x,y,B)$ if and only if these differentials satisfy the linear and quadratic loop equations~\eqref{eq:LLE-W},~\eqref{eq:QLE-W}, and every differential $\omega^{(g)}_{m,0}$ regarded as a meromorphic $1$-form in its last argument has no poles other than at zeros of $dx$ and all $\mathfrak{A}$-periods of this form are vanishing.
\end{proposition}

\begin{proof} The projection property implies that the only poles of $\omega^{(g)}_{m,0}(z_{\set{m}})$, $2g-2+m>0$, are at the points $z_i=p_j$, $i\in \set{m}$, $j=1,\dots,a$, that is, at the zero locus of $dx$. The linear and quadratic loop equations determine uniquely the principal parts of these poles. A meromorphic differential on a compact Riemann surface~$S$ is uniquely determined by its principal parts near its poles up to a holomorphic add-on that is uniquely fixed by its $\mathfrak{A}$-periods.
\end{proof}

\subsection{Higher loop equations}\label{sec:higherloop} Consider the function $\cW^{x,(g)}_{m+1,0}(w;z)$ defined by~\eqref{Wgg} for the case $n=0$, where $z=z_{m+1}$. It is an infinite series in $w$ whose coefficients are expressed as polynomial combinations of the functions $W^{(g')}_{m',0}$ and their derivatives.

\begin{definition} \label{def:higherloop}
	The $r$-loop equation is the requirement that (for a given $r\geq 0$ and for given $g,m\geq 0$)
\begin{equation}
	[w^{r-1}]\cW^{x,(g)}_{m+1,0}(w;z) \in \Xi^x.
\end{equation}
	
\end{definition}

It is straightforward to see from~\eqref{eq:cW} that the $r$-loop equation for $r=1$ and~$2$ are exactly~\eqref{eq:LLE-W} and~\eqref{eq:QLE-W}, respectively. The higher loop equations involve not only the functions $\oW^{(g)}_{m,0}$ themselves but also their derivatives.

\begin{proposition}\label{lem:higherloop}
If the differentials $\{\omega^{(g)}_{m,0}\}$ satisfy the first $2$ loop equations, then they satisfy all loop equations for arbitrary $(r,g,m)$. As a corollary, the differentials obtained by topological recursion satisfy all higher $r$-loop equations.
\end{proposition}

This proposition is proved in~\cite[Proposition 3.3]{DKPS-rspin} in a slightly different but equivalent form. In fact, it is a special case of Proposition~\ref{prop:loopequations} whose independent proof is presented below.

We extend Definition~\ref{def:higherloop} to the full set of mixed correlation differentials $\{\omega^{(g)}_{m,n}\}$.  To this end, let us define the space $\Xi^y$ first.

\begin{definition} (cf.~Definition~\ref{def:FirstDefXi-space}) We denote by $\Xi^y$ the space of meromorphic functions defined in a neighborhood of the zero locus of $dy$ on $S$ and having odd polar part with respect to the deck transformation of $y$. Equivalently, it is spanned by the functions of the form $\partial_y^kf$ where $k=0,1,2,\dots$, and $f$ is holomorphic (or has a 
	simple pole) in a neighborhood of the corresponding zero of~$dy$.
\end{definition}

The structure of $\Xi^y$ and the equivalence of this definition and Definition~\ref{def:FirstDefXi-space} can be clarified in exactly the same way as it was done for $\Xi^x$ in Remark~\ref{rem:Xi-explained}.

\begin{definition}\label{def:all-loop-xy} We say that the differentials $\{\omega^{(g)}_{m,n}\}$ satisfy the $r$-loop equations (for given $r$ and given $(g,m,n)$) if the coefficient of $w^{r-1}$ in $\cW^{x,(g)}_{m+1,n}(w;z)$ belongs to $\Xi^x$ and the coefficient of $\tilde w^{r-1}$ in $\cW^{y,(g)}_{m,n+1}(\tilde w;z)$ belongs to $\Xi^y$.
\end{definition}

\begin{proposition}\label{prop:loopequations} Assume that all zeros of $dx$ and $dy$ are simple and pairwise distinct, the differentials $\{\omega^{(g)}_{m,0}\}$ satisfy the linear and quadratic loop equations~\eqref{eq:LLE-W} and~\eqref{eq:QLE-W} for all $(g,m)$ and have no poles at the zero locus of~$dy$. Then the differentials $\{\omega^{(g)}_{m,n}\}$ given by Definition~\ref{def:mixedcorrdiff} satisfy all loop equations in the sense of Definition~\ref{def:all-loop-xy}. Moreover, $\omega^{(g)}_{m,n}$ is holomorphic in $z_i$ at any zero of $dy$ for $i\in\set{m}$, and it is holomorphic in $z_j$ at any zero of $dx$ if $j\in \set{m+n}\setminus\set{m}$.
\end{proposition}

\begin{proof} The fact that $\omega^{(g)}_{m,n}$ is holomorphic in $z_i$ at any zero of $dy$ for $i\in\set{m}$ is obvious: these poles are absent for the initial forms~$\omega^{(g)}_{m,0}$ and they do not appear in the course of recursion~\eqref{eq:Wxtoomega}. Then, the very form of Equation~\eqref{eq:WxtoW} implies that the loop equations 
\begin{equation}
	[\tilde w^{r-1}]\cW^{y,(g)}_{m,n+1}(\tilde w;z)\in\Xi^y
\end{equation}
hold for all $r$ and all $(g,m,n)$.

If we want to use~\eqref{eq:WytoW} in a similar way in order to prove the loop equations
\begin{equation}
[w^{r-1}]\cW^{x,(g)}_{m+1,n}(w;z)\in\Xi^x,
\end{equation}
we need to show that the functions $\oW^{(g')}_{m',n'}$ entering the right hand side are holomorphic in $z_j$ at zeros of~$dx$ for $j\in \set{m'+n'}\setminus\set{m'}$. This is not obvious at all: these poles are present for the initial functions~$\omega^{(g)}_{m,0}$ and we need to show that they cancel out in the course of recursion~\eqref{eq:Wxtoomega}.

Let us look more attentively at the identity~\eqref{eq:WytoW} for the case $n=0$. The summand with $r=0$ corresponds to the coefficient $[\tilde w^0]\cW^{y,(g)}_{m,1}(\tilde w;z)=\oW^{(g)}_{m,1}$. So we can represent the coefficients of $w^0$ and $w^1$ of this identity as follows
\begin{align}
		[w^0]\cW^{x,(g)}_{m+1,0}(w;z)&=-\frac{dy}{dx}\oW^{(g)}_{m,1}(z_{\set{m}};z)
		-\sum_{r\ge1} \partial_x^r\Bigl([{\tilde w}^r]\frac{dy}{dx} e^{\tilde w\,x}\cW^{y,(g)}_{m,1}(\tilde w;z)\Bigr),\\
		-[w^1]\cW^{x,(g)}_{m+1,0}(w;z)&=-y\,\frac{dy}{dx} \oW^{(g)}_{m,1}(z_{\set{m}};z)
		-\sum_{r\ge1} \partial_x^r\Bigl(y\,[{\tilde w}^r]\frac{dy}{dx} e^{\tilde w\,x}\cW^{y,(g)}_{m,1}(\tilde w;z)\Bigr).
\end{align}
The left hand sides are exactly the expressions~\eqref{eq:LLE-W} and~\eqref{eq:QLE-W} entering the linear and quadratic loop equations, so that by assumption they belong to~$\Xi^x$. The second summands on the right hand sides involve functions $\oW^{(g')}_{m',n'}$ with $2g'-2+m'+n'\leq2g-2+m$. So, arguing by induction, we may assume that the regularity of these functions at zeros of~$dx$ with respect to the variables $z_j$, $j\in \set{m'+n'}\setminus\set{m'}$, is already proved. It follows that these terms also belong to $\Xi^x$. We conclude that $\oW^{(g)}_{m,1}$ regarded as a function of the last argument $z_{m+1}=z$ satisfies
\begin{equation} \label{eq:LLE-QLE-holo}
\frac{dy}{dx}\oW^{(g)}_{m,1}\in\Xi^x,\qquad
y\frac{dy}{dx}\oW^{(g)}_{m,1}\in\Xi^x.
\end{equation}
Let as show that this implies that $\oW^{(g)}_{m,1}$ is holomorphic at the corresponding zero point~$p$ of $dx$. Indeed, the first condition implies that if $\frac{dy}{dx}\oW^{(g)}_{m,1}$ has a pole at~$p$ then this pole should be of odd order. On the other hand, multiplying it by $y-y(p)$ we again obtain an element of~$\Xi^x$, but the order of pole decreases by~$1$ changing the parity of its order. This is only possible if $\frac{dy}{dx}\oW^{(g)}_{m,1}$ has at most simple pole, that is, $\oW^{(g)}_{m,1}$ is regular.

 Thus, we established regularity of  $\oW^{(g)}_{m,1}$ in $z_{m+1}$ at any zero of~$dx$. Applying the recursion~\eqref{eq:Wxtoomega} we obtain that $\oW^{(g)}_{m+1-n,n}$ is regular in $z_{m+1}$  at zeros of $dx$ for any $n=1,2,\dots,m+1$. On the other hand, since $\oW^{(g)}_{m,n}$ is symmetric with respect to $z_{m+1},\dots,z_{m+n}$ (see Remark \ref{rem:Symmetry}), we thereby established  the regularity of $\oW^{(g)}_{m,n}$ with respect to each of the last~$n$ arguments at zeros of~$dx$ by induction on $2g-2+m+n$. 

As it was mentioned already, this implies the loop equations~$[w^{r-1}]\cW^{x,(g)}_{m+1,n}(w;z)\in\Xi^x$ for all $(g,m,n)$ and all $r$ by the very form of~\eqref{eq:WytoW}.
\end{proof}

\begin{remark} Strictly speaking, under the assumption that $dx$ and $dy$ have simple zeros it is not necessary to use the higher loop equations at all in this argument --- indeed, in the key step of the proof we show that $\oW^{(g)}_{m,1}$ is holomorphic using only linear and quadratic loop equations, cf.~Equation~\eqref{eq:LLE-QLE-holo}.
	However, we still included the discussion of higher loop equations for the following two reasons:

\begin{itemize}
	\item Firstly, we obtain a conceptually new proof that the higher loop equations follow from the linear and quadratic ones (thus reproving the core technical result of~\cite{DKPS-rspin}).
	\item Secondly, though we restricted ourselves in this paper to the classical case of topological recursion, there is a generalization developed in~\cite{BouchardEynard} that allows $dx$ (and, for the purpose of the study of $x-y$ swap, $dy$) to have zeros of higher order. Including the higher loop equations in the argument above we actually pave the way to extend it to this more general situation. In this case, it is enough to observe that if $p$ is a zero point of $dx$ of order $r$, then it is enough to check that
	\begin{equation} \label{eq:LLE-QLE-holo-r}
		\frac{dy}{dx}\oW^{(g)}_{m,1}\in\Xi^x, \quad
		y\frac{dy}{dx}\oW^{(g)}_{m,1}\in\Xi^x, \quad \ldots, \quad  {y^{r}}\frac{dy}{dx}\oW^{(g)}_{m,1}\in\Xi^x,
	\end{equation}
	which is reduced to the first $(r+1)$ loop equations for $\{\omega^{(g)}_{m,0}\}$. The full argument for the general Bouchard--Eynard situation would require to match the higher loop equations that we discuss here and the more traditional ones used in~\cite{BE-2},~\cite[Section 7.6]{KramerPhD}, and~\cite{BKS}.
\end{itemize}

\end{remark}

\subsection{Proofs of Theorems~\ref{thm:mainconjholds} and~\ref{thm:simpleformulaxyswap}}
\label{sec:proofofmainthms}

In this section we prove Theorems~\ref{thm:mainconjholds} and~\ref{thm:simpleformulaxyswap}. Firstly, we prove that under our assumptions the differentials $\omega^{(g)}_{0,n}$ given by~\eqref{eq:MainFormulaSimple} satisfy the topological recursion for the input data $(S,y,x,B)$. Then by Theorem~\ref{thm:newformula} we derive Theorem~\ref{thm:mainconjholds} as a direct corollary.

We assume that the zero loci of $dx$ and $dy$ are disjoint and that $\omega^{(g)}_{m,0}$'s satisfy the topological recursion for the input data $(S,x,y,B)$. Then we note that by construction we have $\omega^{(0)}_{0,1} = -xdy$ and $\omega^{(0)}_{0,2}=B$. Consider the differentials $\omega^{(g)}_{m,n}$ obtained recursively by~\eqref{eq:Wxtoomega} or, equivalently, by the explicit formula~\eqref{eq:omega-g-n-graphsxy}. For each $(g,m,n)$ it represents $\omega^{(g)}_{m,n}$ as a finite combination of meromorphic differentials $\tilde\omega^{(g')}_{m',0}$'s on $S$ with some differential operators applied to them.
	Let us list the possible poles of the ingredients of these expressions:
	\begin{enumerate}
		\item The poles of $\tilde\omega^{(g')}_{m',0}$ at the zero locus of $dx$ (in the decorations of multiedges, including the ones coming from the applications of $\partial_x$). These poles remain for the variables $z_{\set{m}}$ and they cancel for the variables $z_{\set{m+n}\setminus\set{m}}$ by Proposition~\ref{prop:loopequations} once we consider the full expression for $\omega^{(g)}_{m,n}$ rather than isolated terms.
		\item The diagonal poles in $\tilde\omega^{(0)}_{2,0}$'s on multiedges. In the exceptional case of $(g,m,n)=(0,2,0)$ this pole just remains, it is outside the stable range that we discuss in this proposition.
		Otherwise, the two variables $\tilde\omega^{(0)}_{2,0}$ on multiedges turn either into $z_i,z_i$ for $i\in \set{m+n}\setminus\set{m}$ --- then this term is regularized and has no pole at the diagonal (that is exactly the difference between $\tilde\omega^{(0)}_{2,0}$ and $\omega^{(0)}_{2,0}$); or into $z_i,z_j$ for $i\in\set{m}$ and $j\in \set{m+n}\setminus\set{m}$ --- these poles remain;
		or into $z_i,z_j$ for $i,j\in \set{m+n}\setminus\set{m}$ --- these poles are canceled in the full expression for $\omega^{(g)}_{m,n}$, $2g-2+m+n>0$, by Proposition~\ref{prop:regular-at-diagonals}.
		\item The poles at the zero locus of $dy$ coming from the factors $(dy_i)^{-1}$ and the application of the operators $\partial_{y_i}$. These poles remain as the poles in variables $z_i$, $i\in \set{m+n}\setminus \set{m}$.
		\item Poles at the poles of $y$ coming from the action of the derivatives $w_i\partial_{x_i}$ on the expression $\frac{1}{dx_i}\omega^{(0)}_{1,0} (z_i)$ in the exponent. Let us introduce a local coordinate $z$ near a pole of $y$ such that $y=z^{-s}$ and $x=z$ in the neighborhood of that pole, for some  $s\in \mathbb{Z}_{\geq 1}$ (the function $x$ might also be nonzero at that point or even have a pole of some degree, but it will only improve the cancellation of the poles, so for brevity we assume that it has a zero of degree $1$; it cannot have a zero of higher degree due to the condition that $y$ is regular at the critical points of $x$). Consider a term in the expansion of the exponential coming from the product of the $r_1$-th, \dots, $r_l$-th terms in the expansion of $\mathcal{S}$. For the $r_i$-th term $expr_{r_i}$ coming from the expansion of $\mathcal{S}$ in the exponent we have
		\begin{equation}
			expr_{r_i} \sim w^{2r_i+1} (\partial_z)^{2r_i} \dfrac{1}{z^s} \sim w^{2r_i+1} z^{-s-2r_i}.
		\end{equation}
		For the product of these terms in the expansion of the exponential (which is also divided by $w$) we have
		\begin{align}
			1/w \prod_{i=1}^l expr_{r_i} \sim w^{l-1+ 2\sum r_i } z^{-sl-2\sum r_i}.
		\end{align}
		For the factor $dx/dy$ we have $dx/dy\sim z^{s+1}$. After that note that $w$ gets replaced by $\partial_y\sim z^s z\partial_z$, which means that the whole expression is
		\begin{align}
			\sim z^{s(l-1+ 2\sum r_i )} z^{-sl-2\sum r_i} z^{s+1} = z^{1+2(s-1)\sum r_i},
		\end{align}
		which is regular for $s\geq 1$ (and $y$ cannot have logarithmic poles due to the condition that it is meromorphic).
	\end{enumerate}

It follows that the only poles of $\omega^{(g)}_{0,n}$, $2g-2+n>0$ are at the zero locus of $dy$. Moreover, by Proposition~\ref{prop:loopequations}, all loop equations for these poles are satisfied, in particular, the linear and the quadratic ones.
Thus, by Proposition~\ref{prop:trueTRformulation}, the only extra thing that we still have to prove is that all $\mathfrak{A}$-periods vanish. In fact, a more general statement holds:

\begin{lemma} Assume that $\mathfrak{A}$-periods of $\omega^{(g)}_{m,0}$ in each of its variables vanish for all $(g,m)$ such that $2g-2+m\geq 0$. Then the $\mathfrak{A}$-periods of $\omega^{(g)}_{m,n}$ in each of its variables vanish for $2g-2+m+n\geq 0$.
\end{lemma}

\begin{proof} Indeed, Equation~\eqref{eq:Wxtoomega} implies that the sum
	\begin{equation}
		\omega^{(g)}_{m+1,n}+\omega^{(g)}_{m,n+1}
	\end{equation}
regarded as a $1$-form in $z=z_{m+1}$	is the differential of a meromorphic function on $S$. Hence the $\mathfrak{A}$-periods of  $\omega^{(g)}_{m,n+1}$ in $z_{m+1}$ vanish if and only if  the $\mathfrak{A}$-periods of  $\omega^{(g)}_{m+1,n}$ in $z_{m+1}$ vanish. Therefore, the $\mathfrak{A}$-periods of $\omega^{(g)}_{m,n}$ in each of its variables vanish for $2g-2+m+n\geq 0$.
\end{proof}

Since we assume that $\omega^{(g)}_{m,0}$ satisfy the topological recursion for $(S,x,y,B)$, then all their $\mathfrak{A}$-periods vanish for $2g-2+m\geq 0$ (by the projection property), hence the $\mathfrak{A}$-periods of $\omega^{(g)}_{0,n}$ vanish for $2g-2+n\geq 0$. This completes the proof of Theorems~\ref{thm:mainconjholds} and~\ref{thm:simpleformulaxyswap}.

\subsection{Splitting of poles and recursive computation of \texorpdfstring{$\omega^{(g)}_{m,n}$}{mixed correlation differentials} }
\label{sec:splitingofpoles}
Recall the definition of the differentials $\omega^{(g)}_{m+1,n}$'s, in particular the two equivalent formulas defining them given in Equations~\eqref{eq:mainrecchange} (rewritten as Equation~\eqref{eq:WXtoomega}) and~\eqref{eq:Wxtoomega}.

The form $\cW^{x,(g)}_{m+1,n}$ participating in~\eqref{eq:Wxtoomega} can be written as
\begin{equation}\label{eq:omega+omega}
	\cW^{x,(g)}_{m+1,n}=\omega^{(g)}_{m+1,n}+O(u),
\end{equation}
where the terms entering the summand $O(u)$ are expressed in terms of the forms $\omega^{(g')}_{m',n'}$ with either $g'<g$ or $g'=g$ but $m'+n'<m+n$. The same remark is applied to Equation~\eqref{eq:Wytoomega}. Therefore, both~\eqref{eq:Wxtoomega} and~\eqref{eq:Wytoomega} can be represented as equations where on the left hand side we have $\omega^{(g)}_{m+1,n}+\omega^{(g)}_{m,n+1}$
and by induction hypothesis the terms entering the right hand side are already computed in the previous steps of computations. Note that Equations~\eqref{eq:Wxtoomega} and~\eqref{eq:Wytoomega} provide different right hand sides (as symbolic expressions in $\omega^{(g')}_{m',n'}$), so we get two different equations of this shape.

We observe now that this relation allows one not only to relate $\omega^{(g)}_{m+1,n}$ and $\omega^{(g)}_{m,n+1}$ but actually to compute both of them. Indeed, the right hand side is a meromorphic form in $z=z_{m+1}$ with the poles at the zeros of $dx$, $dy$, and the points $z_i$, $i\ne m+1$, and with the vanishing $\mathfrak{A}$-periods. Thus, the inductive relation determines the principal parts of all these poles, and we can identify the two summands on the left hand side as the contributions of the principal parts of the poles at zeros of $dx$ and $z_i$, $i>m+1$, for the summand $\omega^{(g)}_{m+1,n}$, and the principal parts of the poles at zeros of $dy$ and the points $z_i$, $i<m+1$, respectively.

 In fact, for the collection of forms $\omega^{(g)}_{m+n}$ with the fixed values of $g$ and $m+n$ it is sufficient to realize this procedure of splitting of poles just once for two arbitrary neighboring values of pairs of bottom indices, say, in order to compute the forms $\omega^{(g)}_{m+n,0}$ and $\omega^{(g)}_{m+n-1,1}$. Then, for all other pairs of bottom indices we can apply inductional relations~\eqref{eq:Wxtoomega} and \eqref{eq:Wytoomega} directly. 

Remark that even for the case $n=0$ the procedure of splitting of poles provides a new and eventually more efficient method of computation the differentials $\omega^{(g)}_{m,0}$ of the original topological recursion. Applying inductional form of
~\eqref{eq:Wxtoomega} and~\eqref{eq:Wytoomega} discussed above
for the case $n=0$ we find, in particular, $\omega^{(g)}_{m+1,0}$ inductively without finding explicitly positions of zeros of~$dx$ and the expansions for the deck transformations at these points.

This method of splitting the poles  is especially useful for explicit computations in the case of a rational spectral curve.


\section{Examples}\label{sec:examples}

In this section we provide examples of computation of differentials $\omega^{(g)}_{m,n}$ for small $g, m, n$.

\subsection{Notation}
In order to reduce the size of the formulas we will use (locally in this section) reduced notations of the following sort
\begin{equation}
	\omega^{(0)}_{(1,\bar *,\bar 3,4,5)}=\omega^{(0)}_{3,2}(z_1,z_4,z_5;z_*,z_3).
\end{equation}
Namely, we use symbols of the kind $\omega^{(g)}_{(k_1,\dots,k_{m+n})}$ where each of the indices from the tuple $K=(k_1,\dots,k_{m+n})$ is either an integer or an overlined integer. In order to denote $\omega^{(g)}_{m,n}(z_I;z_J)$ where $m+n$ is the cardinality of~$K$, $n$ is the number of overlined indices, $I\subset K$ is the subcollection of integer indices and $J$ is the collection of integers $j$ such that $\bar j\in K$. Besides, we denote by $*$ the index of the variable $z=z_*$ used in the operators involved in the formulas. Namely, we set

	\begin{align}
		dx=dx(z_*)&=\frac{\partial x(z_*)}{\partial z_*}dz_*,\qquad D_x(\omega)=d\frac{\omega}{dx}=\frac{\partial\frac{\omega}{dx(z_*)}}{\partial z_*}dz_*,\\ \notag
		dy=dy(z_*)&=\frac{\partial y(z_*)}{\partial z_*}dz_*,\qquad D_y(\omega)=d\frac{\omega}{dy}=\frac{\partial\frac{\omega}{dy(z_*)}}{\partial z_*}dz_*.
	\end{align}

Remark that the form $\omega^{(g)}_{(k_1,\dots,k_{m+n})}$ is invariant with respect to permutations of the indices~$k_i$. Using this notation we can represent the definition of unstable differentials \eqref{eq:unstable} as follows  
	\begin{align} \label{eq:unstable*}
		\omega^{(0)}_{(*)}=-y(z_*)\,dx(z_*),\quad \omega^{(0)}_{(\bar *)}=-x(z_*)\,dy(z_*), 
		\\ \notag 
		\omega^{(0)}_{(1,2)}=-\omega^{(0)}_{(1,\bar 2)}=\omega^{(0)}_{(\bar 1,\bar 2)}=\tfrac{dz_1dz_2}{(z_1-z_2)^2}.
	\end{align}
We also denote
\begin{align}\label{eq:02reg}
		\omega^{(0),{\mathrm{reg}}}_{(1,2)}&=\tfrac{dz_1dz_2}{(z_1-z_2)^2}-\tfrac{dx_1dx_2}{(x_1-x_2)^2},\qquad x_i=x(z_i),\\ \notag
		\omega^{(0),{\mathrm{reg}}}_{(\bar1,\bar2)}&=\tfrac{dz_1dz_2}{(z_1-z_2)^2}-\tfrac{dy_1dy_2}{(y_1-y_2)^2},\qquad y_i=y(z_i).
\end{align}

\subsection{Genus zero}
Applying~\eqref{eq:Wxtoomega} for $(g,m+n+1)=(0,3)$ we get
	\begin{align}
		\omega _{(1,2,\bar{\ast})}^{(0)}+\omega _{(1,2,\ast)}^{(0)}&+D_y\frac{\omega_{(1,\ast)}^{(0)}\omega_{(2,\ast)}^{(0)}}{dx}=0,\\ \notag
		\omega _{(1,\bar{\ast},\bar 3)}^{(0)}+\omega _{(1,\ast,\bar 3)}^{(0)}&+D_y\frac{\omega_{(1,\ast)}^{(0)}\omega_{(\ast,\bar 3)}^{(0)}}{dx}=0,\\ \notag
		\omega _{(\bar{\ast},\bar 1,\bar 3)}^{(0)}+\omega _{(\ast,\bar 1,\bar 3)}^{(0)}&+D_y\frac{\omega_{(\ast,\bar 1)}^{(0)}\omega_{(\ast,\bar 3)}^{(0)}}{dx}=0.
	\end{align}
Similarly, for $(g,m+n+1)=(0,4)$ we get
	\begin{align}
		\omega _{(1,2,3,\bar{\ast})}^{(0)}+\omega _{(1,2,3,\ast)}^{(0)}&
		+D_y\tfrac{\omega_{(1,\ast)}^{(0)}\omega_{(2,3,\ast)}^{(0)}+\omega_{(2,\ast)}^{(0)}\omega_{(1,3,\ast)}^{(0)}+\omega_{(3,\ast)}^{(0)}\omega_{(1,2,\ast)}^{(0)}}{dx}
		+D_y^2\tfrac{\omega_{(1,\ast)}^{(0)}\omega_{(2,\ast)}^{(0)}\omega_{(3,\ast)}^{(0)}}{dx^2}=0,\\ \notag
		\omega _{(1,2,\bar{\ast},\bar4)}^{(0)}+\omega _{(1,2,\ast,\bar4)}^{(0)}&
		+D_y\tfrac{\omega_{(1,\ast)}^{(0)}\omega_{(2,\ast,\bar4)}^{(0)}+\omega_{(2,\ast)}^{(0)}\omega_{(1,\ast,\bar4)}^{(0)}+\omega_{(\ast,\bar4)}^{(0)}\omega_{(1,2,\ast)}^{(0)}}{dx}
		+D_y^2\tfrac{\omega_{(1,\ast)}^{(0)}\omega_{(2,\ast)}^{(0)}\omega_{(\ast,\bar4)}^{(0)}}{dx^2}=0,\\ \notag
		\omega _{(1,\bar{\ast},\bar3,\bar4)}^{(0)}+\omega _{(1,\ast,\bar3,\bar4)}^{(0)}&
		+D_y\tfrac{\omega_{(1,\ast)}^{(0)}\omega_{(\ast,\bar3,\bar4)}^{(0)}+\omega_{(\ast,\bar3)}^{(0)}\omega_{(1,\ast,\bar4)}^{(0)}+\omega_{(\ast,\bar4)}^{(0)}\omega_{(1,\ast,\bar3)}^{(0)}}{dx}
		+D_y^2\tfrac{\omega_{(1,\ast)}^{(0)}\omega_{(\ast,\bar3)}^{(0)}\omega_{(\ast,\bar4)}^{(0)}}{dx^2}=0,\\ \notag
		\omega _{(\bar{\ast},\bar2,\bar3,\bar4)}^{(0)}+\omega_{(\ast,\bar 2,\bar3,\bar4)}^{(0)}&
		+D_y\tfrac{\omega_{(\ast,\bar 2)}^{(0)}\omega_{(\ast,\bar3,\bar4)}^{(0)}+\omega_{(\ast,\bar3)}^{(0)}\omega_{(\ast,\bar 2,\bar4)}^{(0)}+\omega_{(\ast,\bar4)}^{(0)}\omega_{(\ast,\bar 2,\bar3)}^{(0)}}{dx}
		+D_y^2\tfrac{\omega_{(\ast,\bar 2)}^{(0)}\omega_{(\ast,\bar3)}^{(0)}\omega_{(\ast,\bar4)}^{(0)}}{dx^2}=0.
	\end{align}

In general, for the case $g=0$ Equations~\eqref{eq:Wxtoomega} and~\eqref{eq:Wytoomega} specialize to
	\begin{align}
		\omega^{(0)}_{(\bar\ast,K)}&=-\sum_{K=\sqcup_{\alpha\in\cA}J_\alpha,~J_\alpha\ne\emptyset}
		D_y^{|\cA|-1}\frac{\prod_{\alpha\in\cA}\omega^{(0)}_{(\ast,J_\alpha)}}{dx^{|\cA|-1}},\\ \notag
		\omega^{(0)}_{(\ast,K)}&=-\sum_{K=\sqcup_{\alpha\in\cA}J_\alpha,~J_\alpha\ne\emptyset}
		D_x^{|\cA|-1}\frac{\prod_{\alpha\in\cA}\omega^{(0)}_{(\bar\ast,J_\alpha)}}{dy^{|\cA|-1}},
	\end{align}
respectively, where $K$ is any nonempty set of indices, overlined or not, the summation carries over all partitions of $K$ into an unordered collection of nonempty subsets, and $|\cA|$ is the number of parts. This relation is in fact equivalent to a formula suggested by Hock \cite{Hockx-x-y}.

\subsection{Higher genera}
In the case $(g,m+n+1)=(1,1)$ the two dual relations are
	\begin{align}
		\omega^{(1)}_{(\bar \ast)}+\omega^{(1)}_{(\ast)}+D_y\frac{\omega^{(0),{\mathrm{reg}}}_{(\ast,\ast)}}{2\,dx}
		+D_y^2\frac{D_x^2\omega^{(0)}_{(\ast)}}{24}=0,\\
		\omega^{(1)}_{(\bar\ast)}+\omega^{(1)}_{(\ast)}+D_x\frac{\omega^{(0),{\mathrm {reg}}}_{(\bar\ast,\bar\ast)}}{2\,dy}
		+D_x^2\frac{D_y^2\omega^{(0)}_{(\bar\ast)}}{24}=0
	\end{align}
(cf.~\cite{EynardOrantin-toporec},~\cite{BCGFLS-Free},~\cite{Hockx-x-y},~\cite{Hock-FullFormula}).

For greater values of $(g,m,n)$ the equation becomes more complicated but it is still quite explicit. For example, for $(g,m+n+1)=(1,2)$ we get
	\begin{align}
		& \omega^{(1)}_{(1,\bar\ast)}+\omega^{(1)}_{(1,\ast)}
		+D_y\tfrac{\tfrac12\omega^{(0)}_{(1,\ast,\ast)}+\omega^{(0)}_{(1,\ast)}\omega^{(1)}_{(\ast)}}{dx} 
		\\ \notag &
		+D_y^2\Bigl(\tfrac{D_x^2\omega^{(0)}_{(\ast)}}{24}+
		\tfrac{\omega^{(0)}_{(1,\ast)}\omega^{(0),{\mathrm{reg}}}_{(\ast,\ast)}}{2\;dx^2}\Bigr)
		+D_y^3\tfrac{\omega^{(0)}_{(1,\ast)}\;D_x^2\omega^{(0)}_{(\ast)}}{24\;dx}
		=0,
		\\
		& 
		\omega^{(1)}_{(\bar\ast,\bar2)}+\omega^{(1)}_{(\ast,\bar2)}
		+D_y\tfrac{\tfrac12\omega^{(0)}_{(\ast,\ast,\bar2)}+\omega^{(0)}_{(\ast,\bar2)}\omega^{(1)}_{(\ast)}}{dx}
		\\ \notag &
		+D_y^2\Bigl(\tfrac{D_x^2\omega^{(0)}_{(\ast)}}{24}+
		\tfrac{\omega^{(0)}_{(\ast,\bar2)}\omega^{(0),{\mathrm{reg}}}_{(\ast,\ast)}}{2\;dx^2}\Bigr)
		+D_y^3\tfrac{\omega^{(0)}_{(\ast,\bar2)}\;D_x^2\omega^{(0)}_{(\ast)}}{24\;dx}
		=0.
	\end{align} 
To conclude, we present also an explicit relation for the case $(g,m+n+1)=(2,1)$.
	\begin{align}
		\omega^{(2)}_{(\bar\ast)}+\omega^{(2)}_{(\ast)}&
		=-\sum_{r\ge0}D_y^r\biggl([u^r]\biggl(
		u^2\tfrac{D_x^2\omega^{(1)}_{(\ast)}}{24}+u\,\tfrac{\omega^{(1)}_{(\ast,\ast)}}{2\;dx}
		+u^4\tfrac{D_x^4\omega^{(0)}_{(\ast)}}{1920}
		+u^2\tfrac{\omega^{(0)}_{(\ast,\ast,\ast)}}{6\;dx^2}
		\\ \notag &\hskip2cm
		+u^3\tfrac{\big\lfloor_{z_{1}\to z_*}\!\!\!D_x^2\omega^{(0),{\mathrm{reg}}}_{(\ast,1)}}{24\;dx}
		+\frac{u}{2\;dx}\Bigl(
		\omega^{(1)}_{(\ast)}+u^2\tfrac{D_x^2\omega^{(0)}_{(\ast)}}{24}+u\,\tfrac{\omega^{(0),{\mathrm{reg}}}_{(\ast,\ast)}}{2\;dx}
		\Bigr)^2\biggr)\biggr)
		\\ \notag &
		=-\sum_{r\ge0}D_x^r\biggl([u^r]\biggl(
		u^2\tfrac{D_y^2\omega^{(1)}_{(\bar\ast)}}{24}+u\,\tfrac{\omega^{(1)}_{(\bar\ast,\bar\ast)}}{2\;dy}
		+u^4\tfrac{D_y^4\omega^{(0)}_{(\bar\ast)}}{1920}
		+u^2\tfrac{\omega^{(0)}_{(\bar\ast,\bar\ast,\bar\ast)}}{6\;dy^2}
		\\ \notag &\hskip2cm
		+u^3\tfrac{\big\lfloor_{z_{1}\to z_*}\!\!\!D_y^2\omega^{(0),{\mathrm{reg}}}_{(\bar1,\bar\ast)}}{24\;dy}
		+\frac{u}{2\;dy}\Bigl(
		\omega^{(1)}_{(\bar\ast)}+u^2\tfrac{D_y^2\omega^{(0)}_{(\bar\ast)}}{24}+u\,\tfrac{\omega^{(0),{\mathrm{reg}}}_{(\bar\ast,\bar\ast)}}{2\;dy}
		\Bigr)^2\biggr)\biggr).
	\end{align}
All these examples are also worked out explicitly in~\cite[Section 2.3]{Hock-FullFormula}.


\section{Applications}\label{sec:applications}

The universal formulas for $x-y$ swap give \emph{explicit} closed formulas for the output of topological recursion once the correlation differentials on one of the sides are known explicitly. The simplest formulas are obtained in the case when the topological recursion on one of the sides of $x-y$ duality is trivial. This happens, for instance for $(S,y,x,B)$ for $S=\Cf P^1$, $y=z$, arbitrary $x$, and the unique possible $B=dz_1dz_2/(z_1-z_2)^2$. 

Examples of possible functions $x$ that are discussed in the literature and for which the topological recursion for the input $(\Cf P^1, x, y=z, B=dz_1dz_2/(z_1-z_2)^2) $ have interesting enumerative meaning are:
\begin{align}\label{eq:xWitten}
	x & =z^{r}- r\epsilon z & \text {(deformed Witten's $r$-spin class);} \\ \label{eq:xhyper}
	x & = z^{r-1}+ z^{-1} & \text {(enumeration of hypermaps);} \\ \label{eq:xTheta}
	x & = z^{-r} - r\lambda^{r-1} z^{-1} & \text {(deformed $\Theta$ class).}
\end{align}

The most basic and still the most important example, in a slightly different normalization, is:
\begin{align}
		x & =\frac{z^2}{2} & \text{(the Witten--Kontsevich case).} 	
\end{align}
(see~e.~g.~\cite{EynardOrantin-toporec,ZhouWK,Eynard-Book,DOPS,DNOPS,ShiftedWitten,ShiftedTheta} for the discussion of these spectral curves data and their enumerative meaning).

For all these cases, and in fact in general for all possible rational $x$ such that $dx$ has only simple zeros, we give an explicit closed differential-algebraic formula for the corresponding correlation differentials $\omega^{(g)}_{m,0}$.

Note that the theory of $m$-point functions is very well developed for these cases, especially for the Witten--Kontsevich case and the (deformed) Witten's $r$-spin case: there are many competing formulas, and their identification is usually quite non-trivial and gives new insights on the underlying enumerative geometry and related integrable systems. The survey of the existing literature on the $m$-point functions and the comparison of the existing formulas to the one that we obtain using the $x-y$ swap is out of scope of the present paper, we leave it for future research.

\subsection{General closed formula for \texorpdfstring{$\omega^{(g)}_{m,0}$}{mixed correlation differentials}  for the case \texorpdfstring{$y(z)=z$}{y(z)=z}} We assume that the spectral curve is rational and the function $y$ of the spectral curve data is given by $y(z)=z$. Then the differential $dy=dz$ has no zeros and the differentials $\omega^{(g)}_{0,n}$ are known explicitly: they are all trivial in the stable cases $2g-2+n>0$ while for the exceptional unstable cases they are
\begin{equation}
	\omega^{(0)}_{0,1}(z)=-x(z)\; dz,\qquad \omega^{(0)}_{0,2}(z_1,z_2)=\frac{dz_1dz_2}{(z_1-z_2)^2}.
\end{equation}

Using Equation~\eqref{eq:MainFormulaSimple} in the opposite direction (as the one expressing $\omega^{(g)}_{m,0}$'s in terms of $\omega^{(g)}_{0,n}$'s) we obtain the following explicit closed formula for the differentials $\omega^{(g)}_{m,0}$. Consider the set of connected simple (i.e.\ without loops and multiple edges) graphs on~$m$ numbered vertices. For a graph $\Gamma$ in this set, we denote by $V(\Gamma)$ and $E(\Gamma)$ the sets of its vertices and edges, respectively. For a vertex $i\in V(\Gamma)$, we define its weight as
\begin{equation}\label{eq:Wvert}
	W_{i}=-e^{-w_i(\cS(w_i\hbar\partial_{z_i})-1)x(z_i)}\;\frac{dz_i}{w_i}.
\end{equation}
For an edge $(i,j)\in E(\Gamma)$, we define its weight as
\begin{equation}
	\begin{aligned}
		W_{i,j}&=\tfrac{e^{\hbar^2 w_iw_j\cS(w_i\hbar\partial_{z_i})\cS(w_j\hbar\partial_{z_j}){(z_i-z_j)^{-2}}}-1}{\hbar^2}
		=\frac{w_iw_j}{(z_i-z_j)^2-\frac{\hbar^2}{4}(w_i+w_j)^2}.
	\end{aligned}
\end{equation}
By construction, $W_{i}$ is a meromorphic form in $z_i$ and $W_{i,j}$ is a function. Theorem~\ref{thm:simpleformulaxyswap} has the following direct corollary:

\begin{theorem}\label{cor:y=z}
	Let $x(z)$ be an arbitrary rational function. Then all differentials $\omega^{(g)}_{m,0}$ of topological recursion with the spectral curve $(x=x(z),~y=z)$ and the initial form of the recursion
	\begin{equation}
		\omega^{(g)}_{1,0}(z)=-y\,dx=-z\,x'(z)dz
	\end{equation}
	are given by the following explicit closed formula:
	\begin{equation}\label{eq:y=z}
		\omega^{(g)}_{m,0}=[\hbar^{2g}]\sum_{r_1,\dots,r_m\ge0}\prod_{i=1}^m D_{x_i}^{r_i}[w_i^{r_i}]
		\sum_{\Gamma}\hbar^{2g(\Gamma)}\;\prod_{(i,j)\in E(\Gamma)}W_{i,j}\prod_{i\in V(\Gamma)}W_{i}.
	\end{equation}
	The sum here is taken over the set of connected simple graphs on $m$ numbered vertices. The operator $D_{x}$ acts on the space of meromorphic differentials by $D_{x}\,\omega(z)=d\frac{\omega}{dx(z)}$.
\end{theorem}
\begin{remark}
	Let us stress that, as opposed to formulas~\eqref{eq:MainFormula} and~\eqref{eq:MainFormulaSimple}, which represent one of the sets of $\omega^{(g)}_{m,0}$'s and $\omega^{(g)}_{0,n}$'s in terms of the other, the RHS of~\eqref{eq:y=z} is an explicit expression involving only $z_i$ and $x(z_i)$.
\end{remark}

\begin{remark} Note also that while it is typically problematic to describe the limit behavior of topological recursion in a family where several branching points of $x$ collide in the limit, and it often requires some extra conditions, cf.~\cite{limit}, in our approach the formulas that we obtain extend naturally to the discriminants of the Hurwitz spaces for $x$. This observation goes far beyond the case $y=z$ and naturally leads to a new treatment of taking limits in topological recursion that we plan to discuss in a subsequent publication. 
\end{remark}

For example, for small $m$ we have
\begin{align}
	\sum_{g=1}^\infty\hbar^{2g}\omega^{(g)}_{1,0}&=\sum_{r\ge0}D_{x_1}^r[w_1^r]W_{1},\\
	\sum_{g=0}^\infty\hbar^{2g}\omega^{(g)}_{2,0}&=\sum_{r_1,r_2\ge0}D_{x_1}^{r_1}D_{x_2}^{r_2}[w_1^{r_1}w_2^{r_2}]W_{1}W_{2}W_{1,2},\\
	\sum_{g=0}^\infty\hbar^{2g}\omega^{(g)}_{3,0}&=\sum_{r_1,r_2,r_3\ge0}D_{x_1}^{r_1}D_{x_2}^{r_2}D_{x_3}^{r_3}[w_1^{r_1}w_2^{r_2}w_3^{r_3}]\prod_{i=1}^3W_{i}\times\\
	&\qquad\Bigl( W_{1,2}W_{1,3}+W_{1,2}W_{2,3}+W_{1,3}W_{2,3}+\hbar^2W_{1,2}W_{2,3}W_{1,3}\Bigr),\notag
\end{align}
and so on. These formulas compute the differentials of topological recursion without applying the recursion itself. Remark that in spite of appearance of $z_i-z_j$ in the denominators, the differentials $\omega^{(g)}_{m,0}$ have no poles on the diagonals (except for the unstable form $\omega^{(0)}_{2,0}$) due to some miracle cancellation (as follows from Proposition~\ref{prop:regular-at-diagonals}).

\begin{remark}
	Instead of Equation~\eqref{eq:MainFormulaSimple} in the computation of the differentials of Theorem~\ref{cor:y=z} one could equally apply Equation~\eqref{eq:MainFormula}. This would provide a formally different but also a valid closed formula for the differentials $\omega^{(g)}_{m,0}$ of Theorem~\ref{cor:y=z}.
\end{remark}

\subsection{Examples of vertex contributions} The final expressions for all examples that we mentioned in the beginning of this section differ only by the explicit formulas for the vertex weight and $D_x$. Let us compute them explicitly:

\begin{itemize}
	\item Deformed Witten's $r$-spin class (where $x(z)$ is given by~\eqref{eq:xWitten}):
	\begin{align}
		W_i & =  - \frac{dz_i}{w_i} \exp\Big( \frac{(z_i-\frac{w_i\hbar}{2})^{r+1}-(z_i+\frac{w_i\hbar}{2})^{r+1}+(r+1)w_i\hbar z_i^r}{\hbar(r+1)} \Big);
		\\ \notag
		D_x & = d \; \frac{1}{(rz^{r-1}-r\epsilon)dz}.
	\end{align}
	Note that in this case  the vertex weight  does not depend on the deformation parameter.
	\item Enumeration of hypermaps (where $x(z)$ is given by~\eqref{eq:xhyper}):
		\begin{align}
		W_i & =  - \frac{dz_i}{w_i} \Bigg( \frac{1-\frac{w_i\hbar}{2z_i}}{1+\frac{w_i\hbar}{2z_i}}\, e^{\frac{w_i\hbar }{z_i}}\Bigg)^{\frac 1\hbar}
		 \exp\Big( \frac{(z_i-\frac{w_i\hbar}{2})^{r}-(z_i+\frac{w_i\hbar}{2})^{r}+rw_i\hbar z_i^{r-1}}{\hbar r} \Big) ;
		\\ \notag
		D_x & = d \; \frac{z^2}{((r-1)z^{r}-1)dz}.
	\end{align}
	\item Deformed $\Theta$-class (where $x(z)$ is given by~\eqref{eq:xTheta}):
	\begin{align}
		W_i & =  - \frac{dz_i}{w_i} \Bigg( \frac{1-\frac{w_i\hbar}{2z_i}}{1+\frac{w_i\hbar}{2z_i}}\, e^{\frac{w_i\hbar }{z_i}}\Bigg)^{-\frac {r\lambda^{r-1}}\hbar} \!\!
		\exp\Big( \frac{(1+\frac{w_i\hbar}{2z_i})^{1-r}-(1-\frac{w_i\hbar}{2z_i})^{1-r}+(r-1)w_i\hbar z_i^{-1}}{\hbar (r-1)z_i^{r-1}} \Big);
		\\ \notag
		D_x & = - d \; \frac{z^{r+1}}{r(1-(\lambda z)^{r-1})dz}.
	\end{align}
\end{itemize}

\subsection{Specialization to the Witten--Kontsevich case}

One can consider also a further specialization when $x(z)=z^2/2$.  The differentials of topological recursion with the spectral curve $(x=z^2/2, y=z)$ enumerate intersection numbers of $\psi$ classes on the moduli spaces of curves~\cite{EynardOrantin-toporec,ZhouWK}: they are given by
\begin{equation}
	\begin{gathered}
		\omega^{(g)}_{m,0}=\sum_{k_1+\dots+k_m=3g-3+m}\langle \tau_{k_1}\dots\tau_{k_m}\rangle_g\;
		\textstyle\prod_{i=1}^m\frac{(2k_i+1)!!}{z_i^{2k_i+2}}dz_i,\\
		\langle \tau_{k_1}\dots\tau_{k_m}\rangle_g=\int_{\overline{\mathcal M}_{g,m}}\psi_1^{k_1}\dots\psi_m^{k_m}\;.
	\end{gathered}
\end{equation}
On the other hand, substituting $x(z)=z^2/2$ to~\eqref{eq:Wvert} we obtain $W_i=-e^{-w_i^3\hbar^2/24}\frac{dz_i}{w_i}$. This provides a further simplification of~\eqref{eq:y=z}. For example, for $m=1$ we obtain (changing the sign of the parameter $w$)
\begin{equation}
	\omega^{(g)}_{1,0}=[\hbar^{2g}]\sum_{r\ge0}D_{-x_1}^r[w^r]e^{w^3\hbar^2/24}\frac{dz_1}{w}=\frac{1}{24^g g!}D_{-x_1}^{3g-1}dz_1.
\end{equation}
Since $D_{-x}=-d\frac{1}{z\, dz}$ and $\frac{(2k+1)!!}{z^{2k+2}}dz_i=D_{-x}^{k+1}dz$, we recover thereby the well-known equality
\begin{equation}
	\langle\tau_{3g-2}\rangle_g=\frac1{24^g g!}.
\end{equation}
In a similar way, for $m=2$ we obtain
\begin{equation}
	\begin{aligned}
		\omega^{(g)}_{2,0}&=
		\sum_{r_1,r_2\ge0}D_{-x_1}^{r_1}D_{-x_2}^{r_2}[w_1^{r_1}w_2^{r_2}]e^{(w_1^3+w_2^3)\hbar^2/24}
		\tfrac{dz_1dz_2}{(z_i-z_j)^2-\frac{\hbar^2}{4}(w_i+w_j)^2}\\
		&=\sum_{i+k=g}\frac{1}{24^{i} i!4^k}
		(D_{-x_1}^{3}+D_{-x_1}^3)^{i}(D_{-x_1}+D_{-x_2})^{2k}\frac{dz_1dz_2}{(z_1-z_2)^{2(k+1)}}.
	\end{aligned}
\end{equation}
Remark that the pole on the diagonal $z_1=z_2$ cancel on the right hand side for $g>0$. This can be seen explicitly using the identity
\begin{equation}\label{eq:identity-2k}
	\textstyle
	(D_{-x_1}+D_{-x_2})^{k+1}\frac{dz_1dz_2}{(z_1-z_2)^{2(k+1)}}=\frac{(2k+1)!!dx_1dx_2}{(x_1x_2)^{2k+2}}=\frac{1}{(2k+1)!!}(D_{-x_1}D_{-x_2})^{k+1}dx_1dx_2.
\end{equation}
With this identity, we get
\begin{equation}
	\omega^{(g)}_{2,0}=\sum_{i+k=g}\frac{1}{24^{i} i!4^k(2k+1)!!}
	(D_{-x_1}^{3}+D_{-x_1}^3)^{i}(D_{-x_1}+D_{-x_2})^{k-1}(D_{-x_1}D_{-x_2})^{k+1}dz_1dz_2,
\end{equation}
or, equivalently,
\begin{equation}
	\sum_{i+j=3g-1}u_1^iu_2^j\langle\tau_i\tau_j\rangle_g=
	\sum_{i+k=g}\frac{(u_1^{3}+u_2^3)^{i}(u_1+u_2)^{k-1}(u_1u_2)^{k}}{24^{i} i!4^k(2k+1)!!},
\end{equation}
which is equivalent to Dijkgraaf's 2-point function formula \cite{Dijk}.

Note that a similar simplification can be done for any $m\geq 1$.

 

\printbibliography

@misc{BCGFLS-Free,
	title={Functional relations for higher-order free cumulants}, 
	author={Gaëtan Borot and Séverin Charbonnier and Elba Garcia-Failde and Felix Leid and Sergey Shadrin},
	year={2023},
	eprint={2112.12184},
	archivePrefix={arXiv},
}

@misc{borot2021topological,
	title={Topological recursion for fully simple maps from ciliated maps}, 
	author={Gaëtan Borot and Séverin Charbonnier and Elba Garcia-Failde},
	year={2021},
	eprint={2106.09002},
	archivePrefix={arXiv},
}

@article {BEO-loop,
	AUTHOR = {Borot, Ga\"{e}tan and Eynard, Bertrand and Orantin, Nicolas},
	TITLE = {Abstract loop equations, topological recursion and new
	applications},
	JOURNAL = {Commun. Number Theory Phys.},
	FJOURNAL = {Communications in Number Theory and Physics},
	VOLUME = {9},
	YEAR = {2015},
	NUMBER = {1},
	PAGES = {51--187},
	ISSN = {1931-4523},
	MRCLASS = {81T45},
	MRNUMBER = {3339853},
	DOI = {10.4310/CNTP.2015.v9.n1.a2},
	URL = {https://doi.org/10.4310/CNTP.2015.v9.n1.a2},
}

@article {BGF-fullysimple,
	AUTHOR = {Borot, Ga\"{e}tan and Garcia-Failde, Elba},
	TITLE = {Simple maps, {H}urwitz numbers, and topological recursion},
	JOURNAL = {Comm. Math. Phys.},
	FJOURNAL = {Communications in Mathematical Physics},
	VOLUME = {380},
	YEAR = {2020},
	NUMBER = {2},
	PAGES = {581--654},
	ISSN = {0010-3616},
	MRCLASS = {57K20 (14H81 15B52 46L54 60B20)},
	MRNUMBER = {4170288},
	MRREVIEWER = {Dorin Andrica},
	DOI = {10.1007/s00220-020-03867-1},
	URL = {https://doi.org/10.1007/s00220-020-03867-1},
}

@article {BKS,
	AUTHOR = {Borot, Ga\"{e}tan and Kramer, Reinier and Sch\"{u}ler, Yannik},
	TITLE = {Higher {A}iry structures and topological recursion for
	singular spectral curves},
	JOURNAL = {Ann. Inst. Henri Poincar\'{e} D},
	FJOURNAL = {Annales de l'Institut Henri Poincar\'{e} D. Combinatorics, Physics
	and their Interactions},
	VOLUME = {11},
	YEAR = {2024},
	NUMBER = {1},
	PAGES = {1--146},
	ISSN = {2308-5827},
	MRCLASS = {81-XX (14-XX 17-XX 51-XX)},
	MRNUMBER = {4705584},
	DOI = {10.4171/aihpd/168},
	URL = {https://doi.org/10.4171/aihpd/168},
}

@article {BS-blobbed,
	AUTHOR = {Borot, Ga\"{e}tan and Shadrin, Sergey},
	TITLE = {Blobbed topological recursion: properties and applications},
	JOURNAL = {Math. Proc. Cambridge Philos. Soc.},
	FJOURNAL = {Mathematical Proceedings of the Cambridge Philosophical
	Society},
	VOLUME = {162},
	YEAR = {2017},
	NUMBER = {1},
	PAGES = {39--87},
	ISSN = {0305-0041},
	MRCLASS = {14H81 (32G99 81T30)},
	MRNUMBER = {3581899},
	MRREVIEWER = {Xiaobin Li},
	DOI = {10.1017/S0305004116000323},
	URL = {https://doi.org/10.1017/S0305004116000323},
}

@misc{limit,
	title={Taking limits in topological recursion}, 
	author={Gaëtan Borot and Vincent Bouchard and Nitin Kumar Chidambaram and Reinier Kramer and Sergey Shadrin},
	year={2023},
	eprint={2309.01654},
	archivePrefix={arXiv},
}

@article {BouchardEynard,
	AUTHOR = {Bouchard, Vincent and Eynard, Bertrand},
	TITLE = {Think globally, compute locally},
	JOURNAL = {J. High Energy Phys.},
	FJOURNAL = {Journal of High Energy Physics},
	YEAR = {2013},
	NUMBER = {2},
	PAGES = {143, front matter + 34},
	ISSN = {1126-6708},
	MRCLASS = {81T45 (30F10 32G15)},
	MRNUMBER = {3046532},
	MRREVIEWER = {Lee-Peng Teo},
	DOI = {10.1007/JHEP02(2013)143},
	URL = {https://doi.org/10.1007/JHEP02(2013)143},
}

@article {BE-2,
	AUTHOR = {Bouchard, Vincent and Eynard, Bertrand},
	TITLE = {Reconstructing {WKB} from topological recursion},
	JOURNAL = {J. \'{E}c. polytech. Math.},
	FJOURNAL = {Journal de l'\'{E}cole polytechnique. Math\'{e}matiques},
	VOLUME = {4},
	YEAR = {2017},
	PAGES = {845--908},
	ISSN = {2429-7100},
	MRCLASS = {81R10 (14H70 14H81 30F30)},
	MRNUMBER = {3694097},
	MRREVIEWER = {Hsian-Hua Tseng},
	DOI = {10.5802/jep.58},
	URL = {https://doi.org/10.5802/jep.58},
}

@article {bouchard-sulk,
	AUTHOR = {Bouchard, Vincent and Su{\l}kowski, Piotr},
	TITLE = {Topological recursion and mirror curves},
	JOURNAL = {Adv. Theor. Math. Phys.},
	FJOURNAL = {Advances in Theoretical and Mathematical Physics},
	VOLUME = {16},
	YEAR = {2012},
	NUMBER = {5},
	PAGES = {1443--1483},
	ISSN = {1095-0761},
	MRCLASS = {81T30 (14N35)},
	MRNUMBER = {3056954},
	URL = {http://projecteuclid.org/euclid.atmp/1408561555},
}

@article {BDKS-toporec-KP,
	AUTHOR = {Bychkov, Boris and Dunin-Barkowski, Petr and Kazarian, Maxim
	and Shadrin, Sergey},
	TITLE = {Topological recursion for {K}adomtsev--{P}etviashvili tau
	functions of hypergeometric type},
	JOURNAL = {J. Lond. Math. Soc. (2)},
	FJOURNAL = {Journal of the London Mathematical Society. Second Series},
	VOLUME = {109},
	YEAR = {2024},
	NUMBER = {6},
	PAGES = {Paper No. e12946},
	ISSN = {0024-6107},
	MRCLASS = {Prelim},
	MRNUMBER = {4760443},
	DOI = {10.1112/jlms.12946},
	URL = {https://doi.org/10.1112/jlms.12946},
}

@article {BDKS-OrlovScherbin,
	AUTHOR = {Bychkov, Boris and Dunin-Barkowski, Petr and Kazarian, Maxim
	and Shadrin, Sergey},
	TITLE = {Explicit closed algebraic formulas for {O}rlov-{S}cherbin
	{$n$}-point functions},
	JOURNAL = {J. \'{E}c. polytech. Math.},
	FJOURNAL = {Journal de l'\'{E}cole polytechnique. Math\'{e}matiques},
	VOLUME = {9},
	YEAR = {2022},
	PAGES = {1121--1158},
	ISSN = {2429-7100},
	MRCLASS = {05E14 (14H30 14N10 33C80 37K10)},
	MRNUMBER = {4453412},
	DOI = {10.5802/jep.202},
	URL = {https://doi.org/10.5802/jep.202},
}

@article {BDKS-FullySimple,
	AUTHOR = {Bychkov, Boris and Dunin-Barkowski, Petr and Kazarian, Maxim
	and Shadrin, Sergey},
	TITLE = {Generalised ordinary vs fully simple duality for {$n$}-point
	functions and a proof of the {B}orot-{G}arcia-{F}ailde
	conjecture},
	JOURNAL = {Comm. Math. Phys.},
	FJOURNAL = {Communications in Mathematical Physics},
	VOLUME = {402},
	YEAR = {2023},
	NUMBER = {1},
	PAGES = {665--694},
	ISSN = {0010-3616},
	MRCLASS = {14H81 (46N50 57K20)},
	MRNUMBER = {4616685},
	DOI = {10.1007/s00220-023-04732-7},
	URL = {https://doi.org/10.1007/s00220-023-04732-7},
}

@misc{BDKS-symplectic,
	title={Symplectic duality for topological recursion}, 
	author={Boris Bychkov and Petr Dunin-Barkowski and Maxim Kazarian and Sergey Shadrin},
	year={2023},
	eprint={2206.14792},
	archivePrefix={arXiv},
	primaryClass={math-ph}
}

@article {ShiftedWitten,
	AUTHOR = {Charbonnier, S\'{e}verin and Chidambaram, Nitin and Garcia-Failde,
	Elba and Giacchetto, Alessandro},
	TITLE = {Shifted {W}itten classes and topological recursion},
	JOURNAL = {Trans. Amer. Math. Soc.},
	FJOURNAL = {Transactions of the American Mathematical Society},
	VOLUME = {377},
	YEAR = {2024},
	NUMBER = {2},
	PAGES = {1069--1110},
	ISSN = {0002-9947},
	MRCLASS = {14H10 (14H70 34E05 81R10)},
	MRNUMBER = {4688543},
	DOI = {10.1090/tran/9046},
	URL = {https://doi.org/10.1090/tran/9046},
}

@misc{ShiftedTheta,
	title={Relations on {$\overline{\mathcal{M}}_{g,n}$} and the negative $r$-spin {W}itten conjecture}, 
	author={Nitin Kumar Chidambaram and Elba Garcia-Failde and Alessandro Giacchetto},
	year={2023},
	eprint={2205.15621},
	archivePrefix={arXiv},
}

@incollection {Dijk,
    AUTHOR = {Faber, Carel},
     TITLE = {A conjectural description of the tautological ring of the
              moduli space of curves},
 BOOKTITLE = {Moduli of curves and abelian varieties},
    SERIES = {Aspects Math.},
    VOLUME = {E33},
     PAGES = {109--129},
 PUBLISHER = {Friedr. Vieweg, Braunschweig},
      YEAR = {1999},
      ISBN = {3-528-03125-5},
   MRCLASS = {14H10 (14C15 14C17 14N35)},
  MRNUMBER = {1722541},
MRREVIEWER = {Elham\ Izadi},
}

@article {DKPS-rspin,
	AUTHOR = {Dunin-Barkowski, Petr and Kramer, Reinier and Popolitov,
	Alexandr and Shadrin, Sergey},
	TITLE = {Loop equations and a proof of {Z}vonkine's {$qr$}-{ELSV}
	formula},
	JOURNAL = {Ann. Sci. \'{E}c. Norm. Sup\'{e}r. (4)},
	FJOURNAL = {Annales Scientifiques de l'\'{E}cole Normale Sup\'{e}rieure. Quatri\`eme
	S\'{e}rie},
	VOLUME = {56},
	YEAR = {2023},
	NUMBER = {4},
	PAGES = {1199--1229},
	ISSN = {0012-9593},
	MRCLASS = {14H30 (05E14 14H10 14N10 20C30)},
	MRNUMBER = {4650153},
	DOI = {10.24033/asens.2553},
	URL = {https://doi.org/10.24033/asens.2553},
}

@article {DNOPS,
	AUTHOR = {Dunin-Barkowski, P. and Norbury, P. and Orantin, N. and
	Popolitov, A. and Shadrin, S.},
	TITLE = {Dub\-ro\-vin's superpotential as a global spectral curve},
	JOURNAL = {J. Inst. Math. Jussieu},
	FJOURNAL = {Journal of the Institute of Mathematics of Jussieu. JIMJ.
	Journal de l'Institut de Math\'{e}matiques de Jussieu},
	VOLUME = {18},
	YEAR = {2019},
	NUMBER = {3},
	PAGES = {449--497},
	ISSN = {1474-7480},
	MRCLASS = {53D45 (14H81 32G15 81T45)},
	MRNUMBER = {3936638},
	MRREVIEWER = {Felix Janda},
	DOI = {10.1017/s147474801700007x},
	URL = {https://doi.org/10.1017/s147474801700007x},
}

@article {DOPS,
	AUTHOR = {Dunin-Barkowski, Petr and Orantin, Nicolas and Popolitov,
	Aleksandr and Shadrin, Sergey},
	TITLE = {Combinatorics of loop equations for branched covers of sphere},
	JOURNAL = {Int. Math. Res. Not. IMRN},
	FJOURNAL = {International Mathematics Research Notices. IMRN},
	YEAR = {2018},
	NUMBER = {18},
	PAGES = {5638--5662},
	ISSN = {1073-7928},
	MRCLASS = {14H81 (14H57 57M12)},
	MRNUMBER = {3862116},
	MRREVIEWER = {Hsian-Hua Tseng},
	DOI = {10.1093/imrn/rnx047},
	URL = {https://doi.org/10.1093/imrn/rnx047},
}

@book {Eynard-Book,
	AUTHOR = {Eynard, Bertrand},
	TITLE = {Counting surfaces},
	SERIES = {Progress in Mathematical Physics},
	VOLUME = {70},
	NOTE = {CRM Aisenstadt chair lectures},
	PUBLISHER = {Birkh\"{a}user/Springer, [Cham]},
	YEAR = {2016},
	PAGES = {xvii+414},
	ISBN = {978-3-7643-8796-9; 978-3-7643-8797-6},
	MRCLASS = {81-02 (05C10 05C30 14H60 30Fxx 57M50 81T20 81T30)},
	MRNUMBER = {3468847},
	MRREVIEWER = {Daniel David Moskovich},
	DOI = {10.1007/978-3-7643-8797-6},
	URL = {https://doi.org/10.1007/978-3-7643-8797-6},
}

@article {EynardOrantin-toporec,
	AUTHOR = {Eynard, B. and Orantin, N.},
	TITLE = {Invariants of algebraic curves and topological expansion},
	JOURNAL = {Commun. Number Theory Phys.},
	FJOURNAL = {Communications in Number Theory and Physics},
	VOLUME = {1},
	YEAR = {2007},
	NUMBER = {2},
	PAGES = {347--452},
	ISSN = {1931-4523},
	MRCLASS = {14H15 (14N35 32A27 37K10 37K20 81T45)},
	MRNUMBER = {2346575},
	MRREVIEWER = {Vincent Bouchard},
	DOI = {10.4310/CNTP.2007.v1.n2.a4},
	URL = {https://doi.org/10.4310/CNTP.2007.v1.n2.a4},
}

@article {EynardOrantin-xysymmetry,
	AUTHOR = {Eynard, B. and Orantin, N.},
	TITLE = {Topological expansion of mixed correlations in the {H}ermitian
	2-matrix model and {$x$}-{$y$} symmetry of the {$F_g$}
	algebraic invariants},
	JOURNAL = {J. Phys. A},
	FJOURNAL = {Journal of Physics. A. Mathematical and Theoretical},
	VOLUME = {41},
	YEAR = {2008},
	NUMBER = {1},
	PAGES = {015203, 28},
	ISSN = {1751-8113},
	MRCLASS = {14H15 (14N10 14N35 82B41)},
	MRNUMBER = {2450700},
	MRREVIEWER = {Brad Safnuk},
	DOI = {10.1088/1751-8113/41/1/015203},
	URL = {https://doi.org/10.1088/1751-8113/41/1/015203},
}

@misc{eynard2013xy,
	title={About the $x-y$ symmetry of the $F_g$ algebraic invariants}, 
	author={B. Eynard and N. Orantin},
	year={2013},
	eprint={1311.4993},
	archivePrefix={arXiv},
	primaryClass={math-ph}
}

@phdthesis{ElbaPhD,
	author  = "E. Garcia-Failde",
	title   = "On discrete surfaces. Enumerative geometry, matrix models
	and universality classes via topological recursion",
	school  = "Rheinischen Friedrich-Wilhelms-Universitaet Bonn",
	year    = "2018",
	url = "https://hdl.handle.net/20.500.11811/7888",
}

@misc{Hockx-x-y,
	title={On the $x$-$y$ Symmetry of Correlators in Topological Recursion via Loop Insertion Operator}, 
	author={Alexander Hock},
	year={2022},
	eprint={2201.05357},
	archivePrefix={arXiv},
	primaryClass={math-ph}
}

@article {Hock-FullFormula,
	AUTHOR = {Hock, Alexander},
	TITLE = {A simple formula for the {$x$}-{$y$} symplectic transformation
	in topological recursion},
	JOURNAL = {J. Geom. Phys.},
	FJOURNAL = {Journal of Geometry and Physics},
	VOLUME = {194},
	YEAR = {2023},
	PAGES = {Paper No. 105027, 26},
	ISSN = {0393-0440},
	MRCLASS = {46L54 (15B52 16R60)},
	MRNUMBER = {4659813},
	DOI = {10.1016/j.geomphys.2023.105027},
	URL = {https://doi.org/10.1016/j.geomphys.2023.105027},
}

@phdthesis{KramerPhD,
	author  = "R. Kramer",
	title   = "Cycles of curves, cover counts, and central invariants",
	school  = "Universiteit van Amsterdam",
	year    = "2019",
	url = "https://hdl.handle.net/11245.1/588b7706-082d-4836-9cf7-7ca5025f3d64"
}

@article {ZhouWK,
	AUTHOR = {Zhou, Jian},
	TITLE = {Topological recursions of {E}ynard-{O}rantin type for
	intersection numbers on moduli spaces of curves},
	JOURNAL = {Lett. Math. Phys.},
	FJOURNAL = {Letters in Mathematical Physics},
	VOLUME = {103},
	YEAR = {2013},
	NUMBER = {11},
	PAGES = {1191--1206},
	ISSN = {0377-9017},
	MRCLASS = {14H10 (14N35 53D45)},
	MRNUMBER = {3095142},
	MRREVIEWER = {Shengmao Zhu},
	DOI = {10.1007/s11005-013-0632-7},
	URL = {https://doi.org/10.1007/s11005-013-0632-7},
}

\end{document}